\documentclass[english]{amsart}
\usepackage{babel}
\usepackage[]{geometry}
\usepackage{caption}

\usepackage{enumerate}

\usepackage{caption}
\usepackage{subcaption}
\usepackage{dsfont}
\usepackage{amsmath,amssymb,amsfonts,latexsym,cancel}
\usepackage{pictexwd}
\usepackage{tikz}
\usepackage[colorlinks, linkcolor=blue, citecolor=blue]{hyperref}           

\usepackage{tikz-cd}
\usepackage{pgfplots}
\pgfplotsset{compat=newest}
\pgfplotsset{plot coordinates/math parser=false}
\newlength\figureheight
\newlength\figurewidth

\usepackage{color}
\usepackage{epsfig}
\usepackage{comment}
\usepackage{booktabs}

\newcommand{\R}{\mathbb{R}}
\newcommand{\N}{\mathbb{N}}
\newcommand{\Z}{\mathbb{Z}}

\DeclareMathOperator{\Cov}{Cov}
\DeclareMathOperator{\Var}{Var}
\DeclareMathOperator{\mse}{{MSE}}
\DeclareMathOperator{\Ge}{Geo}

\def\1{\raisebox{2pt}{\rm{$\chi$}}}

\newcommand{\p}[1]{\left(#1\right)}

\newcommand{\norm}[1]{\left\|#1\right\|}

\newcommand{\abs}[1]{\left| #1 \right|}

\newcommand{\dist}{\operatorname{dist}}

\renewcommand{\P}{\mathbb{P}}

\newcommand{\supp}{\operatorname{supp}}

\theoremstyle{plain}
\newtheorem{definition}{Definition}[section]

\newtheorem{theorem}[definition]{Theorem}
\newtheorem{corollary}[definition]{Corollary}
\newtheorem{lemma}[definition]{Lemma}
\newtheorem{remark}[definition]{Remark}
\newtheorem{assumption}[definition]{Assumption}

\newtheorem{conjecture}[definition]{Conjecture}

\theoremstyle{definition}

\theoremstyle{remark}

\RequirePackage{mathtools}
\DeclarePairedDelimiter{\braces}{\{}{\}}
\DeclarePairedDelimiter{\brackets}{(}{)}
\DeclarePairedDelimiter{\floor}{\lfloor}{\rfloor}
\DeclareMathOperator{\bias}{bias}

\begin{document}

\title[Sample Complexity Analysis of MTD via Markovian and Hard-Core MRA]{Sample Complexity Analysis of Multi-Target Detection via Markovian and Hard-Core Multi-Reference Alignment}

\author{Kweku Abraham}
\address{Kweku Abraham
\newline \indent
{Department of Applied Mathematics and Theoretical Physics \newline \indent
University of Cambridge}
\newline \indent
{Wilberforce Road, CB3 0WA}
\newline\indent
Cambridge, {United Kingdom}}
\email{lkwa2@cam.ac.uk}

\author{Amnon Balanov}
\author{Tamir Bendory}
\address{Amnon Balanov and Tamir Bendory
\newline\indent
School of Electrical and Computer Engineering \newline\indent Tel Aviv University \newline\indent Tel Aviv, Israel}
\email{amnonba15@gmail.com}
\email{bendory@tauex.tau.ac.il}

\author{Carlos Esteve-Yag\"{u}e}
\address{Carlos Esteve-Yag\"{u}e
\newline \indent
{Departamento de Matem\'{a}ticas} 
\newline \indent
{Universidad de Alicante}
\newline \indent
{03690 San Vicente del Raspeig}
\newline \indent
{Alicante, Spain}}
\email{c.esteve@ua.es}

\begin{abstract}

Motivated by single-particle cryo-electron microscopy, we study the sample complexity of the multi-target detection (MTD) problem, in which an unknown signal appears multiple times at unknown locations within a long, noisy observation. We propose a patching scheme that reduces MTD to a non-i.i.d. multi-reference alignment (MRA) model. In the one-dimensional setting, the latent group elements form a Markov chain, and we show that the convergence rate of any estimator matches that of the corresponding i.i.d. MRA model, up to a logarithmic factor in the number of patches. Moreover, for estimators based on empirical averaging, such as the method of moments, the convergence rates are identical in both settings. We further establish an analogous result in two dimensions, where the latent structure arises from an exponentially mixing random field generated by a hard-core placement model. As a consequence, if the signal in the corresponding i.i.d. MRA model is determined by moments up to order $n_{\min}$, then in the low-SNR regime the number of patches required to estimate the signal in the MTD model scales as  $\sigma^{2n_{\min}}$, where $\sigma^2$ denotes the noise variance.
\end{abstract}

\date{\today}

\maketitle

\section{Introduction}

\subsection{Problem formulation.}
\label{sec: 1d MTD formulation intro}
In this work, we analyze the sample complexity of the \emph{multi-target detection} (MTD) problem.
The goal is to recover a signal $X \in \mathbb{R}^L$ that appears multiple times at unknown locations within a noisy observation $Z \in \mathbb{R}^{LM}$~\cite{bendory2019multi}. The observation model is  
\begin{align}
    Z = \sum_{i=1}^{q} S(t_i) * X + \varepsilon, \label{eqn:MTDmodelHomogenous}
\end{align}
where $*$ denotes linear convolution, $q$ is the number of signal occurrences, and $\varepsilon \sim \mathcal{N}(0, \sigma^2 I_{LM})$ is additive Gaussian noise.  
Each unknown location $t_i \in \{0,1,\ldots, L(M-1)\}$ is represented by a one-hot vector $S(t_i) \in \{0,1\}^{L(M-1)+1}$, which shifts $X$ so that its first entry is aligned at position $t_i$ in $Z$. While both the signal $X$ and the set of positions $\{t_i\}_{i=1}^q$ are unknown, 
we treat the positions as nuisance parameters and focus on estimating $X$. 
Throughout the introduction, we describe the one-dimensional MTD model with $X \in \mathbb{R}^L$ and measurement $Z \in \mathbb{R}^{LM}$ for clarity of exposition. However, our main results also apply in a two-dimensional setting, where the signal is $X \in \mathbb{R}^{L \times L}$ and the measurement is $Z \in \mathbb{R}^{ML \times ML}$. This extension is introduced in Section~\ref{sec:2DhomogeneousMTD}. 

At low noise levels, signal occurrences can be reliably detected and aligned, allowing for averaging to suppress noise and improve recovery~\cite{dadon2024detection}. However, as the noise level increases, localizing signal instances becomes increasingly unreliable. In such high-noise regimes, one must estimate the underlying signal without access to positional information (see Figure~\ref{fig:1}(a) for illustration).
This work aims to characterize the sample complexity of the MTD model, namely the observation length ($LM$) needed to recover the underlying signal reliably, in the presence of high noise.

\begin{figure}[t!]
    \centering
    \includegraphics[width=1.0\linewidth]{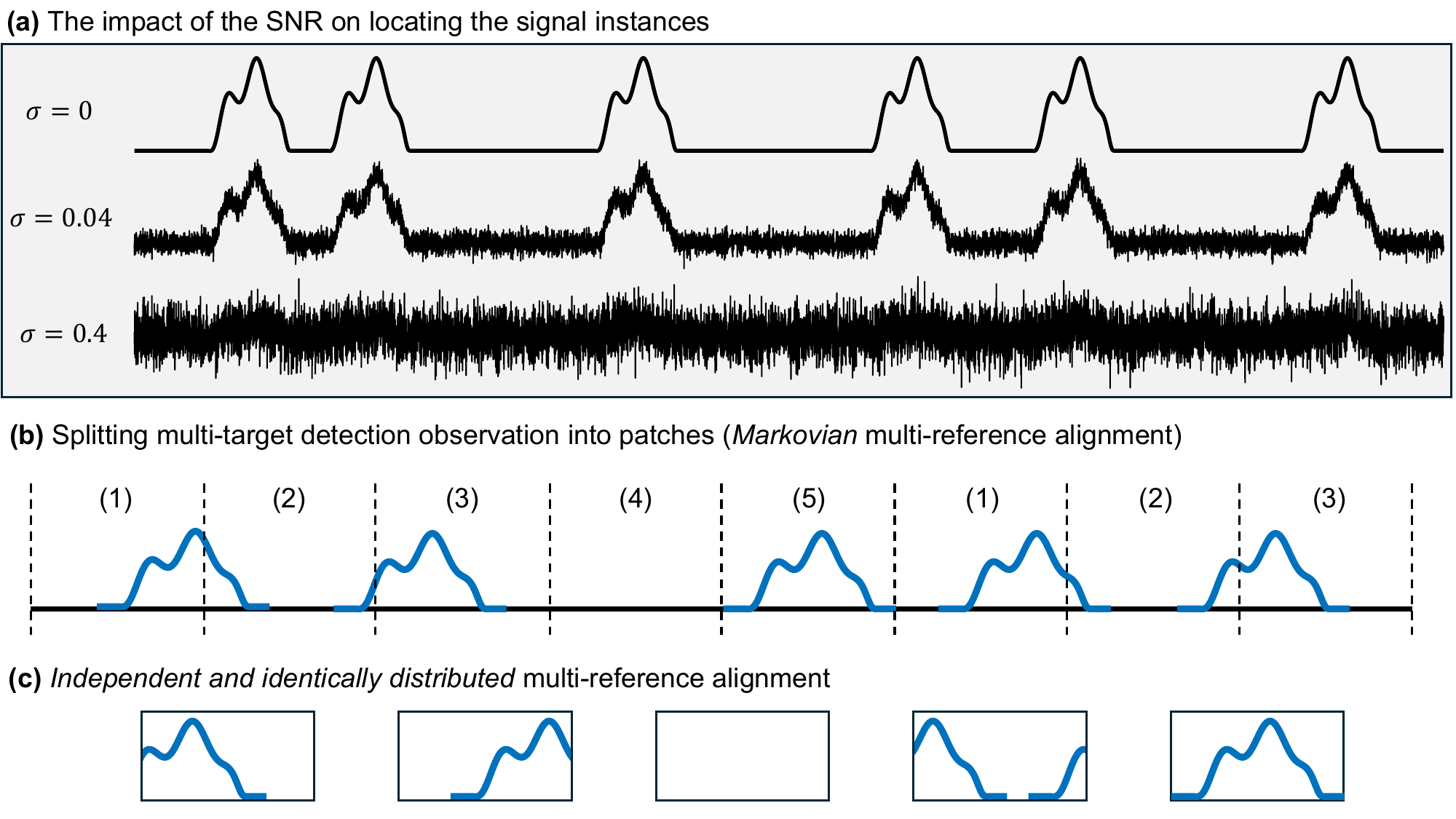}
    \caption{\textbf{The multi-target detection (MTD) model, and its relation to the Markovian and i.i.d.\ multi-reference alignment (MRA) models in one dimension.}
    \textbf{(a)} The MTD model contains signal instances that are embedded at random locations within a noisy observation. In the low-noise regime, it is possible to identify the locations of these instances. However, in the high-noise regime, which is our motivation in this work, the signal locations cannot be reliably recovered.
    \textbf{(b)} To address this challenge, we divide the MTD observation into non-overlapping patches. These patches form what we call the \emph{Markovian MRA} model, where adjacent patches exhibit statistical Markovian dependencies. Each patch can take one of five possible forms: (1) a cropped prefix of a signal instance; (2) two cropped segments, consisting of the suffix of one signal instance followed by the prefix of another; (3) a cropped suffix of a signal instance; (4) an empty patch; or (5) a complete signal instance.
    \textbf{(c)} In contrast, the \emph{i.i.d.\ MRA} model assumes the patches are independent and identically distributed, with no statistical dependency between them. A central contribution of this work is establishing a connection between the Markovian MRA model and the i.i.d.\ MRA model, enabling analysis via known tools and results for the i.i.d.\ case.}
    \label{fig:1}
\end{figure}

\subsection{Motivating application: Cryo-electron microscopy} 
The primary motivation for the MTD model arises in cryo-electron microscopy (cryo-EM)~\cite{nogales2016development, bendory2020single,singer2020computational} and cryo-electron tomography (cryo-ET)~\cite{chen2019complete, schaffer2019cryo, zhang2019advances}. 
Over the past decade, single-particle cryo-EM has become a leading method for determining the spatial structures of biological macromolecules, particularly proteins~\cite{bai2015cryo,glaeser2016good,cheng2018single}. It enables reconstruction of molecular structures in their native state and has achieved atomic-level resolution for numerous proteins. Beyond static structure determination, cryo-EM also offers the potential to capture conformational dynamics, which are critical for understanding protein function~\cite{toader2023methods}.

In a cryo-EM experiment, dozens of micrographs are recorded, each containing multiple 2D projections of an unknown 3D molecular structure at unknown orientations and positions. To avoid radiation damage, imaging is performed under extremely low electron doses, yielding a very low signal-to-noise ratio (SNR), often around $1/100$~\cite{bendory2020single}. The central challenge is to reconstruct the 3D structure from these noisy projections without knowledge of their locations or viewing directions. As the molecular size decreases, the SNR drops further, making particle detection increasingly difficult and, in extreme cases, infeasible~\cite{henderson1995potential}. Under such low-SNR conditions, downstream analyses can be strongly affected, sometimes producing misleading structural inferences~\cite{henderson2013avoiding,balanov2025structure}.

The MTD framework was recently proposed as a computational approach for recovering small structures that lie beyond the capabilities of current methodologies. Although heuristic algorithms have been developed, most notably those based on the method of moments and approximate expectation–maximization~\cite{bendory2023toward,kreymer2022approximate, kreymer2023stochastic,bendory2019multi}, the theoretical and statistical foundations of signal recovery in this regime remain largely unexplored~\cite{balanov2025note,balanov2025orbit}. To be more precise, some statistical properties of cryo-EM have been analyzed in~\cite{bandeira2023estimation,bendory2024sample,bendory2024transversality}, under the assumption that particle locations are known. However, the full problem with \emph{unknown locations} has not been addressed, and the present work constitutes a first step toward this more general and realistic formulation.

The MTD model analyzed in this work~\eqref{eqn:MTDmodelHomogenous} serves as a simplified abstraction of the signal-recovery process underlying the cryo-EM computational problem. Specifically, it isolates the challenge of unknown translations while omitting other sources of uncertainty that arise in more realistic settings, such as rotational variability and tomographic projection effects. This simplification allows us to focus on the core statistical and computational difficulties of recovery in the low-SNR regime under unknown locations of the signal instances, while serving as a first step toward more advanced models that more closely reflect cryo-EM. 
In Section~\ref{sec:outlook}, we discuss how the insights developed here could extend to more general models, including those directly relevant to cryo-EM.

\subsection{Signal recovery using Markovian modeling}
In this work, we develop an estimation framework based on patch extraction and Markovian modeling, illustrated in Figure~\ref{fig:1}(b). The central idea is to recover $X$ by dividing the long MTD measurement $Z$ (defined in~\eqref{eqn:MTDmodelHomogenous}) into $M$ non-overlapping patches, each of length $L$. These patches, denoted by $\{Z_k\}_{k=1}^M \subset \mathbb{R}^L$, are defined by 
\begin{equation}
    \label{patches-1d}
    Z_k[\xi] := Z[L(k-1) + \xi], \qquad \text{for } \xi \in \{0, \ldots, L-1\}, \quad k = 1, \ldots, M.
\end{equation}

\subsubsection{Description of each patch} \label{sec:descriptionOfEachPatch}
Under the non-overlapping assumption, where no two signal instances share overlapping coordinates, each patch has five possible forms: (i) a cropped prefix of a signal instance; (ii) two cropped segments, consisting of the suffix of one signal instance followed by the prefix of another; (iii) a cropped suffix of a signal instance; (iv) an empty patch; or (v) a complete signal instance.
Let $i_k^{(0)}, i_k^{(1)} \in \{0,\dots,L\}$ denote the starting indices of the possible occurrence of $X$ within the patches $Z_{k-1}$ and $Z_k$ respectively, considering $i_k^{(0)} = 0$ in case the patch $Z_{k-1}$ does not contain the starting pixel of $X$, and $i_k^{(1)}=L$ in case the patch $Z_k$ does not contain the starting pixel of $X$.
The non-overlapping condition is equivalent to $i_k^{(1)}\geq i_k^{(0)}$.
We define the padded version of $X$ given by $\tilde{X} = [X, 0_L]$, where $0_L \in \mathbb{R}^L$ denotes the vector of $L$ zeros. Then, each patch $Z_k \in \mathbb{R}^L$ can be expressed as
\begin{align}
    Z_k[\xi] = \tilde{X}\left[\xi - i_k^{(1)} \bmod 2L\right] + \tilde{X}\left[\xi + L - i_k^{(0)} \right] + \varepsilon_k, \qquad \text{for} \quad  \xi \in \{0, \dots, L-1\},\label{eqn:patch}
\end{align}
where $\varepsilon_k\stackrel{i.i.d.}{\sim}\mathcal{N}(0,\sigma^2 I_{L})$, $\xi \in \{0, \dots, L-1\}$. Here, $\tilde{X}\left[\xi - i_k^{(1)} \bmod 2L\right]$ corresponds to the prefix entries of the first (possible) instance of $X$ starting in the patch $Z_k$, while $\tilde{X}\left[\xi + L - i_k^{(0)}\right]$ corresponds to the suffix entries of the (possible) cropped instance of $X$ that started in patch $Z_{k-1}$. The indices $\{(i_k^{(0)}, i_k^{(1)})\}_{k=1}^M$ are random variables determined by the unknown placements of $X$ in $Z$, and they introduce statistical dependencies between patches. See Figure~\ref{fig:1_2} for illustration of the patch construction.

\begin{figure}[t!]
    \centering
    \includegraphics[width=0.9\linewidth]{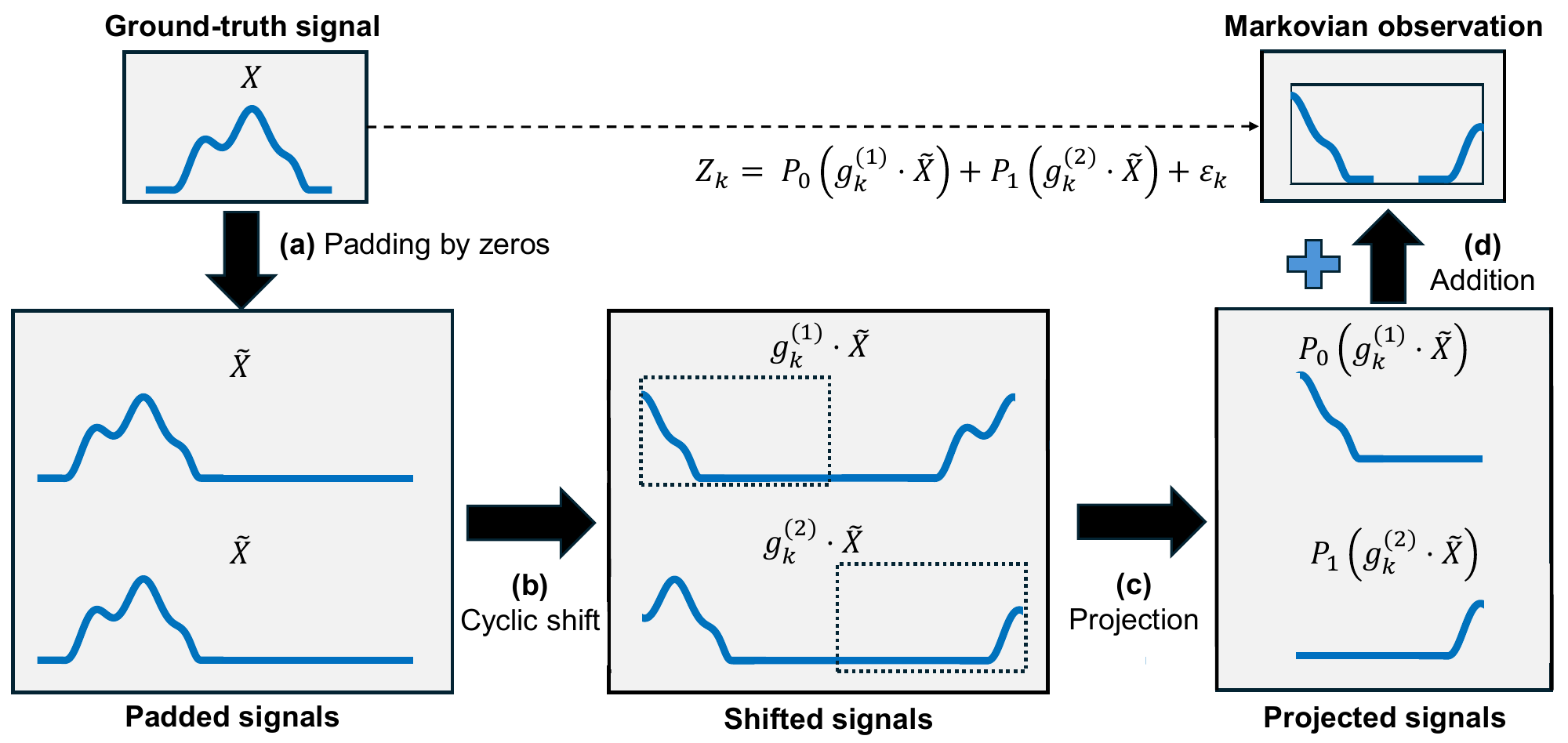}    \caption{\textbf{Illustration of the patch construction in \eqref{eqn:patch}.} 
    Starting from the ground-truth signal $X$, we form its zero-padded version $\tilde X = [X,0_L] \in \mathbb{R}^{2L}$ (panel (a)). Two copies of $\tilde X$ are then cyclically shifted (panel (b)), corresponding to the two possible signal contributions to a patch: a suffix of a signal instance that started in the previous patch, and a prefix of a signal instance that starts in the current patch. 
    Applying the projection operators $P_0$ and $P_1$ extracts the relevant $L$-dimensional portions of these shifted padded signals (panel (c)). Their sum yields the \emph{noiseless} patch of the observed Markovian patch $Z_k$ (panel (d)).}
    \label{fig:1_2}
\end{figure}

\subsubsection{From i.i.d.\ MRA to Markovian MRA via MTD patches}

The structure of each MTD patch $Z_k$ defined in~\eqref{eqn:patch} admits a natural group-theoretic description: a cyclic shift of a zero-padded signal, followed by a linear projection operator. This places the model in direct analogy with the classical i.i.d.\ \emph{multi-reference alignment (MRA)} problem~\cite{bandeira2014multireference,bandeira2023estimation,bendory2017bispectrum,balanov2025expectation}. In the classical \emph{i.i.d.\ MRA model}, the task is to estimate $X \in \mathbb{R}^L$ from $M$ independent noisy observations of the form
\begin{equation}
\label{eqn:classical-MRA}
    Y_k \;=\; P(g_k \cdot X) + \varepsilon_k, 
    \qquad g_k \sim \pi, \ \ \varepsilon_k \sim \mathcal{N}(0,\sigma^2 I),
\end{equation}
where $g_k$ is drawn i.i.d.\ from a distribution $\pi$ over a compact group $G$ and $P$ is a linear projection operator.
More broadly, this model belongs to the general framework of \emph{orbit recovery}, in which one seeks to estimate an unknown signal from noisy observations of transformed copies of that signal under a group action. Classical examples include MRA, where the group consists of cyclic shifts, and cryo-EM, where the group is the rotation group. In this language, the patching construction used here reduces the original MTD problem to an orbit-recovery problem associated with a suitable group action on a zero-padded signal. Indeed, each patch is expressed as a combination of projected transformed copies of the padded signal, placing the estimation problem within the broader orbit-recovery paradigm.
 
 A central insight from the orbit recovery and MRA literature is that if the orbit of $X$ is uniquely determined by its moments up to order $n_{\min}$—the \emph{moment order cutoff}—then reliable recovery in the low-SNR regime requires
\[
    M = \omega(\sigma^{2n_{\min}}),
\]
that is, the number of observations must grow faster than $\sigma^{2n_{\min}}$ as $\sigma^2 \to \infty$~\cite{abbe2018estimation}.

In contrast, \emph{MTD patches are not i.i.d.} . Although each $Z_k$ can be written in the MRA form~\eqref{eqn:classical-MRA}, the latent group elements $\{g_k\}_{k=1}^M$ inherit dependencies from the random placement of signals in $Z$ (i.e., the group elements $g_k$ are associated with the indices $(i_k^{(0)}, i_k^{(1)})$), forming a hidden Markov model~\cite{cappe2005inference}. This motivates the following generalization.

\begin{definition}[Markovian and i.i.d.\ MRA models] \label{def:markovianAndMRAmodels} 
We present two versions of the MRA model:  
    \begin{enumerate} 
        \item \textbf{Markovian MRA model} (MTD model with dependent patches, Figure~\ref{fig:1}(b)): Let $\{ Z_k \}_{k=1}^{M}$ denote the sequence of patches extracted from the MTD observation $Z$, as defined in~\eqref{patches-1d}. These patches are statistically dependent, governed by a hidden Markov model~\cite{cappe2005inference}, induced by the random placement of signals in $Z$. More precisely, $\{ Z_k \}_{k=1}^{M}$ have the form of \eqref{eqn:classical-MRA}, with $\{ g_k \}_{k=1}^M$ being the first $M$ states of a Markov chain.
        \item \textbf{IID MRA model} (model with independent observations, Figure~\ref{fig:1}(c)): $\{ Y_k \}_{k=1}^{M}$ generated according to~\eqref{eqn:classical-MRA}, with the parameters $\{ g_k\}_{k = 1}^{M}$ sampled i.i.d.\ from a fixed distribution. When this distribution is the stationary distribution of the Markovian MRA model in~(i), we refer to it as the \emph{induced MRA model}.  
    \end{enumerate} 
\end{definition}

Markov chains are among the most fundamental dependent-data models. Under mild conditions, they exhibit \emph{exponential mixing}, meaning that samples taken sufficiently far apart become nearly independent and approximately follow the stationary distribution~\cite{levin2017markov,meyn2012markov}. The rate of this convergence is governed by the \emph{absolute spectral gap}, $\Delta = 1 - |\lambda_2|$, where $\lambda_2$ is the second-largest eigenvalue of the transition matrix in absolute value. This mixing property will allow us to transfer sample-complexity results from the well-studied i.i.d.\ MRA to the MTD setting.

Our main results also extend to the two-dimensional setting, motivated by applications in cryo-EM, where the target signal is $X \in \mathbb{R}^{L \times L}$ and the measurement is $Z \in \mathbb{R}^{ML \times ML}$. In this case, the simple Markov description no longer applies. Instead, we adopt a \emph{hard-core model} (see Section~\ref{sec:hard-core-model}) to describe the spatial distribution of signal occurrences, which retains the essential features of the one-dimensional formulation, while capturing the geometry of the 2D setting.

\subsubsection{Main assumptions and relation to prior work} 
\label{sec:mainAssmp}

All prior analyses of the MTD model have focused on the \emph{well-separated} regime, in which signal occurrences are assumed to be separated by at least $L$ coordinates~\cite{bendory2019multi, balanov2025note,balanov2025orbit}. This assumption greatly simplifies the problem, and under this setting, it has been shown that the sample complexity of the MTD model can be upper bounded by $\omega(\sigma^6)$, via explicit recovery from third-order moments~\cite{bendory2019multi}. However, it remains unclear whether such guarantees continue to hold in more general settings that relax the separation condition, namely, when signal occurrences do not overlap but are otherwise arbitrarily positioned.

In this work, we focus on the more challenging \emph{non-overlapping} regime, where signal instances may appear arbitrarily close to each other, subject only to the condition that they do not overlap. For the one-dimensional MTD problem, we model the distances between consecutive signal occurrences as independent draws from a Geometric distribution with parameter $\lambda \in (0,1)$. This assumption is natural for modeling sparsely distributed particles and endows the resulting process with a Markov structure. We emphasize that key parts of our analysis, most notably Theorems~\ref{thm:MTD-1D} and \ref{thm:sampleComplexity1DMTD}, remain valid under different assumptions on the particle location distribution, requiring only that patches sufficiently far apart exert limited statistical influence on each other (see Definition~\ref{def:exponential-mixing}). 
A natural extension to two dimensions is obtained by assuming a \emph{hard-core} spatial model, in which the same principle of weak long-range correlations ensures that distant patches exert only limited statistical influence on each other (Sections~\ref{sec:2DhomogeneousMTD}--\ref{sec:mainResults2D}).

\subsection{Main contributions}

Our main results establish a direct connection between the sample complexity of the MTD model (Definition~\ref{def:markovianAndMRAmodels}(i)) and that of the i.i.d.\ MRA model (Definition~\ref{def:markovianAndMRAmodels}(ii)).
We show that under natural mixing conditions, the sample complexity of MTD is no greater than that of its corresponding i.i.d.\ MRA, up to a possible logarithmic factor.
In particular, under the assumptions in Section \ref{sec:mainAssmp} for the 1D and 2D MTD models, we are able to prove the following:
\begin{enumerate}
    \item \emph{Structure of extracted patches.} 
    In Theorem~\ref{thm:MTD-1D} (for the 1D case) and Theorem~\ref{thm:2d-main-result} (for the 2D case), we show that the patches $\{Z_k\}_{k=1}^{M}$ extracted from the MTD measurement admit an MRA-type representation. 
    In the 1D case, this representation takes the exact form of~\eqref{eqn:classical-MRA} with $X \in \mathbb{R}^L$, and the group elements $\{g_k\}_{k=1}^{M}$ being the first $M$ states of a Markov chain with positive absolute spectral gap. 
    In the 2D case, where the signal is $X \in \mathbb{R}^{L \times L}$, the representation is a natural two-dimensional analogue of~\eqref{eqn:classical-MRA}; here the group elements follow a stochastic process that satisfies a strong mixing property.    
    \item \emph{Sample complexity equivalence.} 
    In Theorem~\ref{thm:sampleComplexity1DMTD} (for the 1D case) and Theorem~\ref{thm:2d-general-result} (for the 2D case), we show that the convergence rate of any estimator under the Markovian MRA model matches the rate under the i.i.d.\ MRA model, up to a logarithmic factor in the number of patches $M$. 
    Moreover, when the estimator is based on empirical averaging, the convergence rates in the two settings match up to a constant factor.
\end{enumerate}

These results imply that, under the stated assumptions, the sample complexity of the MTD model is no greater than that of its associated i.i.d.\ MRA model. In particular, if orbit recovery for the associated i.i.d.\ MRA model is possible from $n_{\min}$-th order moments with sample complexity $\omega(\sigma^{2n_{\min}})$, then the same bound holds for the MTD model. In Section~\ref{sec:empirical_simulation} we provide empirical evidence that moments up to order $3$ suffice for recovery in the induced i.i.d.\ MRA model, supporting $n_{\min}=3$. Hence, when particle location is infeasible (e.g. in the low-SNR regime), recovery from MTD measurements requires $\omega(\sigma^{6})$ patches.

\subsection{Work structure}
The paper is organized as follows. 
Section~\ref{sec:pre} introduces the method of moments together with the feature-map framework and estimator definitions that we use throughout. 
Section~\ref{sec:mainResults1D} develops the one-dimensional MTD model: we formalize the measurement model, split it into patches, and show that the induced sequence of patches forms a Markovian MRA model (Theorem~\ref{thm:MTD-1D}). 
Section~\ref{sec:convergenceToMRA1D} proves that the Markovian MRA model converges to the classical i.i.d.\ MRA model under exponential mixing time and derives the resulting sample-complexity guarantees; we also state the implications of this result for the method of moments (Corollary~\ref{cor:MoM-MTD}). 
Section~\ref{sec:2DhomogeneousMTD} turns to two dimensions, introducing the hard-core placement model, admissibility on the grid graph, and the exponential mixing property for the induced random field. 
Section~\ref{sec:mainResults2D} establishes the corresponding 2D results: an induced (non-i.i.d.) MRA representation with exponential mixing (Theorem~\ref{thm:2d-main-result}) and convergence to the i.i.d.\ MRA setting with matching sample complexity up to logarithmic factors (Theorem~\ref{thm:2d-general-result}). 
Section~\ref{sec:empirical_simulation} presents empirical simulations that (i) demonstrate recovery in the induced i.i.d.\ MRA model from second- and third-order moments and (ii) illustrate the convergence of Markovian patch moments to their i.i.d.\ counterparts, supporting the theoretical scaling laws.
Finally, Section~\ref{sec:outlook} discusses implications, open questions, and extensions, including algebraic-structure variants and connections to cryo-EM and cryo-ET. 
All proofs are deferred to the appendices.

\section{Preliminaries} \label{sec:pre}

The goal of this work is to characterize the sample complexity of the MTD model, defined as the number of patches $M$ needed to reliably recover the signal $X$ in the low-SNR regime. Although the measurement contains $q$ occurrences of $X$, we express complexity in terms of $M$. Under our asymptotic framework, the distribution of inter-signal spacings---and hence the stationary law of the induced Markov chain---remains fixed. Equivalently, the density parameter $\lambda$, which specifies the expected number of signal occurrences per patch, is treated as constant as $M \to \infty$. Thus, $q \asymp \lambda M$, so asymptotics in $M$ translate directly to asymptotics in $q$. The stationary distribution determined by $\lambda$ governs the typical patch structure and thereby  shapes the sample complexity.

\subsection{Method of moments}\label{sec:method-of-moments}
A key statistical property governing the sample complexity in the i.i.d.\ MRA model presented in~\eqref{eqn:classical-MRA} is the moment order cut-off, i.e., the minimal moment order of the observations that uniquely determines the orbit of the signal. We distinguish between two related notions of moments:  

\begin{enumerate}
    \item \emph{Noiseless group-transformed moments.}  
    For a compact group $G$ acting on $\mathbb{R}^L$, let $g \sim \pi$ be a random element drawn from a distribution $\pi$ over $G$, and let $P$ be a fixed projection operator. The $n$-order moment of the noiseless group-transformed signal is  
    \begin{align}
        T_{X,\pi}^{(n)} 
        \;\triangleq\; \mathbb{E}_{g \sim \pi}\!\left[\big(P(g \cdot X)\big)^{\otimes n}\right], 
        \label{eqn:MoMdefinition}
    \end{align}
    which captures the $n$-th order correlations of the projected transformations of $X$, averaged over the group distribution. Here $v^{\otimes n}$ denotes the $n$-fold tensor power of a vector $v \in \mathbb{R}^p$, i.e., the rank-one tensor with entries
    \[
      (v^{\otimes n})_{i_1,\ldots,i_n} \;=\; v_{i_1}\cdots v_{i_n},
    \]
    so that, componentwise,
    \[
      \big(T_{X,\pi}^{(n)}\big)_{i_1,\ldots,i_n}
      \;=\; \mathbb{E}_{g \sim \pi}\!\left[ \big(P(g \cdot X)\big)_{i_1}\cdots \big(P(g \cdot X)\big)_{i_n} \right].
    \]

    \item \emph{Empirical moments.}  
    In practice, we only observe noisy samples $Y_k = P(g_k \cdot X) + \varepsilon_k$, as defined in \eqref{eqn:classical-MRA},
    and thus the relevant population moment is
    \begin{align}
        \label{eqn:def-empirical-moments}
        T_{Y}^{(n)} 
        \;\triangleq\; \mathbb{E}[Y^{\otimes n}]
        \;=\; \mathbb{E}_{g \sim \pi,\, \varepsilon}\!\left[\big(P(g \cdot X)+\varepsilon\big)^{\otimes n}\right],
    \end{align}
    which depends explicitly on the noise level $\sigma$.  
    Given $M$ independent observations $\{Y_k\}_{k=1}^M$, the \emph{empirical $n$-th moment} is defined as
    \[
        \widehat{T}^{(n)} 
        \;\triangleq\; \frac{1}{M} \sum_{k=1}^M (Y_k)^{\otimes n},
    \]
    serving as an estimator of $T_{Y}^{(n)}$.
\end{enumerate}

A fundamental result in the MRA literature is that if the smallest integer $n$ such that $T_{X,\pi}^{(n)}$ uniquely determines the orbit of $X$ is $n_{\min}$, then the sample complexity in the low-SNR regime scales as
\[
    M = \omega(\sigma^{2 n_{\min}}).
\]  
Thus, $n_{\min}$ serves as the key statistical parameter controlling the asymptotic difficulty of orbit recovery.

The minimal moment order $n_{\min}$ (and thus the sample complexity) depends on the group action $G$, the group distribution $\pi$, the projection operator $P$, and the availability of prior information about $X$. For example, when $G = \mathbb{Z}_L$ with a uniform distribution $\pi$ and no prior, recovery requires third-order moments, leading to the scaling $\omega(\sigma^{6})$~\cite{bendory2017bispectrum,bandeira2023estimation}. In contrast, for generic non-uniform $\pi$~\cite{abbe2018multireference,bendory2022dihedral} or when suitable prior information is available~\cite{bendory2024transversality,amir2025stability}, second-order moments suffice, reducing the scaling to $\omega(\sigma^{4})$. In higher-dimensional settings, such as $G=\mathrm{SO}(3)$ with tomographic projections, the uniform case is also believed to require $n_{\min}=3$~\cite{bandeira2023estimation}, though the precise characterization remains an open problem.

\subsection{Feature maps} \label{sec:MoMandTargetMaps}
To state our main results in a general way, it is useful to introduce the notion of a \emph{feature map}. Formally, a feature map is a function $\mathcal{F} : \mathbb{R}^L \to \mathbb{R}^d$ that associates to each signal $X \in \mathbb{R}^L$ a vector $\mathcal{F}(X) \in \mathbb{R}^d$ representing the quantities of interest for the estimation task. For instance, if the goal is to characterize the $n$-th order moment of the observation $Y$ in~\eqref{eqn:classical-MRA} associated with the underlying signal $X$, then the feature map is given by
\begin{align}
\mathcal{F}(X) := \mathbb{E}[Y^{\otimes n}]
        \;=\; \mathbb{E}_{g \sim \pi,\, \varepsilon}\!\left[\big(P(g \cdot X)+\varepsilon\big)^{\otimes n}\right],
\end{align}
which allows recovery of the noiseless moment $T_{X,\pi}^{(n)}$. On the other hand, if the objective is to directly reconstruct $X$ itself, then the natural choice of feature map is the identity transformation $\mathcal{F}(X) \triangleq X$. This general framework allows us to express our results in a way that applies both to moment-based recovery and to direct signal estimation, thereby unifying the treatment of these problems under a common formalism.

\subsubsection{Estimators in Markovian and i.i.d.\ MRA models}
With this formulation of the feature maps, we now extend the definition of the estimators of the Markovian and i.i.d.\ MRA model introduced in Definition \ref{def:markovianAndMRAmodels}. 
Let us denote by
\begin{align}
    \widehat{F}: \bigcup_{M \geq 1} (\mathbb{R}^L)^M \longrightarrow \mathbb{R}^d, \label{eqn:estimatorDef}
\end{align}
any estimator that associates, to any collection of $M$ patches of length $L$, a prediction of the target features $\mathcal{F}(X)$ of the underlying signal $X$. 
Depending on the nature of the patches, we define two estimation models associated to the same map $\widehat{F}(\cdot)$. 

\begin{definition}[Feature estimation in the Markovian and i.i.d.\ MRA models]  
\label{def:markovianAndMRAmodelsFeatures}  
Let $M \in \mathbb{N}$ and let $\widehat{F}(\cdot)$ be the estimator defined in~\eqref{eqn:estimatorDef}, which maps a collection of patches or observations to an approximation of the target features $\mathcal{F}(X) \in \mathbb{R}^d$ of the signal $X$. We define two associated estimators:  
\begin{enumerate}  
    \item \textbf{Markovian MRA model estimation.}  
    Given the Markovian model in Definition~\ref{def:markovianAndMRAmodels}(i) and the sequence of patches $\{ Z_k \}_{k=1}^M$ extracted from the MTD observation $Z$, we define, for any $m \in \mathbb{N}$,  
    \[
        \widehat{F}_{\mathsf{MTD}}(M, m) := \widehat{F}\big(\{ Z_{km} \}_{k=1}^{M'}\big),  
        \qquad M' = \lfloor M/m \rfloor,
    \]  
    which applies $\widehat{F}(\cdot)$ to a subsampled set of patches, retaining only one out of every $m$ patches to mitigate dependence.
    \item \textbf{i.i.d.\ MRA model estimation.}  
    Given the i.i.d.\ MRA model in Definition~\ref{def:markovianAndMRAmodels}(ii) and the sequence of independent observations $\{ Y_k \}_{k=1}^M$, we define  
    \[
        \widehat{F}_{\mathsf{MRA}}(M) := \widehat{F}\big(\{ Y_k \}_{k=1}^M\big),
    \]  
    which applies $\widehat{F}(\cdot)$ directly to the i.i.d.\ observations.  
\end{enumerate}  
\end{definition}

\section{The one-dimensional MTD model: Forming a Markovian MRA model} \label{sec:mainResults1D}

This section introduces the preliminaries, definitions, and assumptions underlying the analysis of the one-dimensional MTD model. Section~\ref{sec:1DhomogenounsMTD} presents a detailed formulation of the model, along with the assumption governing the stochastic placement of the signal $X$ within the MTD observation. In Section~\ref{sec:splitingIntoPatches}, we describe the process of splitting the observation into patches using group actions and projection operators, and in Section~\ref{sec:1dimentionalMTDresults} we establish that the sequence of extracted patches forms a Markovian MRA model.

\subsection{The 1-D MTD model} \label{sec:1DhomogenounsMTD}
In this subsection, we introduce the one-dimensional MTD model and the stochastic assumptions governing the placement of signal occurrences within the measurement.

\subsubsection{Model formulation} 
Let $X\in \R^{L}$ be the signal to be recovered and $Z\in \R^{LM}$ be the measurement introduced in \eqref{eqn:MTDmodelHomogenous}. 
We recall from  \eqref{eqn:MTDmodelHomogenous} that the measurement $Z\in \R^{LM}$ is given by 
\begin{equation*}
Z = \sum_{i=1}^q S(t_i) \ast  X  + \varepsilon,
\end{equation*}
where, for each $i\in \{1, \ldots , q\}$, the location of the $i$-th occurrence of the signal $X$ in the measurement is determined by the parameter $t_i\in \{0, \ldots, L(M-1)\}$. The linear convolution of the one-hot vector $S(t_i)$ with $X$ is given by the vector $S(t_i) \ast X \in \R^{LM}$ defined, for any $\xi \in \{ 0, \ldots , LM - 1 \}$, as
\begin{equation}
    \label{shift measurement}
S(t_i) \ast X [\xi] =
\begin{cases}
    X [\xi - t_i], & \xi\in \{t_i, \ldots , L - 1 + t_i\}, \\
    0, & \text{else}.
\end{cases}
\end{equation}
The parameter $t_i$ denotes the location of the first pixel of the $i$-th occurrence of the signal $X$, specifically the position of $X[0]$. An example of such a measurement is shown in Figure~\ref{fig:1}.

\subsubsection{The distribution of signal appearances}
In the 1D MTD model introduced in \eqref{eqn:MTDmodelHomogenous}, a natural way to model the distribution of signal appearances is to assume that the number of empty pixels between two consecutive signals follows a geometric distribution with parameter $\lambda \in (0,1)$, which controls the density of signal occurrences, as we state in the following assumption. 

\begin{assumption}[Distribution of occurrences in 1D MTD model]
    \label{assumption MTD}
    The number of empty pixels between two signal appearances follows a \emph{zero-indexed} Geometric distribution with parameter $\lambda \in (0,1)$, denoted by $\Ge(\lambda)$, that is, for $k \in \{0,1,2,\ldots\}$, $Pr[\Ge(\lambda) = k] = (1-\lambda)^k \lambda$. The finite set $\{ t_i\}_{i=1}^q$ in \eqref{eqn:MTDmodelHomogenous} is generated as the first $q$ terms of the stochastic process
    \[
        t_1 \sim \Ge(\lambda), 
        \qquad
        t_i - t_{i-1} - L \sim \Ge(\lambda) 
        \quad \forall i = 2,3,\ldots,
    \]
    where 
    \[
        q := \max \left\{ i \in \N \ : \ t_i \leq L(M-1) \right\}.
    \]
\end{assumption}

It is important to note that this assumption ensures that signal occurrences do not overlap within the measurement. The parameter $\lambda$ quantifies the density of signal appearances in the MTD measurement: values of $\lambda$ close to 1 correspond to a high density of $X$ within $Z$, whereas smaller values of $\lambda$ indicate sparser occurrences.

We emphasize that our convergence results are not specific to Assumption \ref{assumption MTD}. Indeed, any assumption on the distribution of occurrences satisfying the exponential mixing property (see  Definition~\ref{def:exponential-mixing}) would provide the same convergence guarantees. For instance, we could have chosen a 1D analogue of the hard-core model presented in Section~\ref{sec:hard-core-model}. However, we chose to consider Assumption~\ref{assumption MTD} for clarity and presentation purposes, since it induces a simple Markov process, and explicit computations of the convergence rate and the associated stationary distribution are possible. In addition, we note that the patching construction does not require exact knowledge of the signal length: it is enough to choose $L$ as an upper bound on the signal support. In that case, however, the non-overlapping requirement is enforced relative to this upper bound, and is therefore potentially more restrictive than if the exact support size were known.

\subsection{Splitting the measurement into patches} \label{sec:splitingIntoPatches}
In order to reduce the MTD model to an MRA model, the first step is to divide the measurement into $M$ patches of length $L$, denoted by $\{Z_k\}_{k = 1}^{M}\subset \R^{L}$, and given by
\begin{equation}
Z_k [\xi] := Z[L(k-1) + \xi], \qquad \text{for all} \ \xi\in \{ 0, \ldots, L-1\}
\ \text{and} \ k = 1, \ldots , M. \label{eqn:sequenceOneDomention}
\end{equation}
Under the geometric distribution assumption, Assumption~\ref{assumption MTD}, each patch may contain up to two shifted, non-overlapping instances of the signal $X \in \mathbb{R}^L$. To represent each patch within the MRA framework (as in \eqref{eqn:patch}), next we describe each patch using cyclic shifts and projection operators.

We embed the signal into a zero-padded vector $\tilde{X} := (X, \, 0_L) \in \R^{2L}$ and consider the action of a finite cyclic group. Consider the cyclic group $ \mathbb{Z}_{2L} $, acting on $ \mathbb{R}^{2L} $ by cyclically shifting the coordinates of any vector in $ \mathbb{R}^{2L} $. Specifically, for any $ \tilde{g} \in \mathbb{Z}_{2L} $, the group action is defined as:
\begin{equation}
    \label{group action 1d simple}
    \tilde{g} \cdot \tilde{X}[\xi] = \tilde{X} [\xi - \tilde{g} \bmod 2L],
    \qquad \forall \xi\in \{0, \ldots , 2L -1\}.
\end{equation}
Since each patch can include up to two instances of the signal, we must account for this by defining the relative shift for each instance and applying the appropriate projection. Let $ \overline{X} $ represent the concatenation of the two instances: 
\begin{align}
    \overline{X} = \p{\tilde{X}, \tilde{X}} \in \mathbb{R}^{2L} \times \mathbb{R}^{2L},
    \quad \text{with} \ \tilde{X} = [X, 0_L]\in \R^{2L}.
    \label{eqn:barXdef}
\end{align}
Now, consider the direct product group $ \mathbb{Z}_{2L} \times \mathbb{Z}_{2L} $, which acts on $ \overline{X} \in \mathbb{R}^{2L} \times \mathbb{R}^{2L} $ by applying the group action of $ \mathbb{Z}_{2L} $ on each component of $ \overline{X} $. Specifically, for any $ g = (g^{(1)}, g^{(2)}) \in \mathbb{Z}_{2L} \times \mathbb{Z}_{2L} $ and any $  (\tilde{X}_0, \tilde{X}_1) \in \mathbb{R}^{2L} \times \mathbb{R}^{2L} $, the action is defined as:
\begin{equation}
    g \cdot (\tilde{X}_0, \tilde{X}_1) = (g^{(1)}, g^{(2)})\cdot (\tilde{X}_0, \tilde{X}_1) = (g^{(1)}\cdot \tilde{X}_0, \, g^{(2)} \cdot \tilde{X}_1). \label{eqn:cyclicGroupActionDef}
\end{equation}

It remains to define the projection operator. We introduce the linear operator $P: \mathbb{R}^{2L} \times \mathbb{R}^{2L} \to \mathbb{R}^L$ defined as
\begin{align}
    P\left(\tilde{X}_0, \tilde{X}_1\right) = P_0 \tilde{X}_0 + P_1 \tilde{X}_1, \qquad \text{for any } (\tilde{X}_0, \tilde{X}_1) \in \mathbb{R}^{2L}\times \mathbb{R}^{2L}, \label{eqn:projecionOperatorDef}
\end{align}
where $P_0, P_1: \mathbb{R}^{2L} \to \mathbb{R}^L$, are linear operators represented by the matrices
\begin{align*}
    P_0 = [0_{L\times L} \ \  I_L]
    \quad \text{and} \quad
    P_1 = [I_L \ \  0_{L\times L}],
\end{align*}
with $0_{L \times L}$ denoting the $L \times L$ zero matrix and $I_L$ the identity matrix in $\mathbb{R}^L$.
In words, $P_0 \tilde{X}_0$ extracts the second half of the vector $\tilde{X}_0$, while $P_1 \tilde{X}_1$ extracts the first half of the vector $\tilde{X}_1$.

We point out that, in view of this construction and the non-overlapping assumption, not all the group elements $g_k = (g_k^{(1)}, g_k^{(2)})\in \Z^{2L} \times \Z^{2L}$ are admissible. In our construction, ignoring the noise variable $\varepsilon_k$, the patches $Z_k$ are modeled by 
$$
Z_k = P (g_k\cdot \overline{X}) = P_0 (g_k^{(1)} \cdot \tilde{X}) + P_1 (g_k^{(2)} \cdot \tilde{X}).
$$
Since, in this construction, the term $ P_0 (g_k^{(1)} \cdot \tilde{X})$ represents the suffix of the possible signal appearance starting in patch $Z_{k-1}$, the variable $g_{k}^{(1)}$ can only take values in $\{0, \ldots, L-1\}$.
On the other hand, since  $P_1 (g_k^{(2)} \cdot \tilde{X})$ represents the prefix of the possible signal starting in the current patch, the variable $g_k^{(2)}$ can only take values in $\{0, \ldots, L\}$. Moreover, the non-overlapping assumption also imposes the condition $g_k^{(1)} \leq g_k^{(2)}$. 
These admissibility conditions on the group elements $g_k$ are inherited from our construction; see Figure~\ref{fig:1_2} for an illustration.
A different choice of the zero-padded signal $\tilde{X}$ and the projection operators $P_0$ and $P_1$ would result in similar but different admissibility conditions on the group elements $g_k^{(1)}$ and $g_k^{(2)}$.
We made this choice for clarity and notational simplicity.

\subsection{From MTD to a Markovian MRA model} \label{sec:1dimentionalMTDresults}
We are now ready to define the MRA model associated with the sequence of patches presented in \eqref{eqn:sequenceOneDomention}.

\begin{definition}[The induced MRA model]
\label{def:inducedMRAmodel}
    Let us consider the finite abelian group $G = \mathbb{Z}_{2L} \times \mathbb{Z}_{2L}$ acting on $\overline{X}$ as described in \eqref{eqn:cyclicGroupActionDef} and \eqref{group action 1d simple}. Let $P: \mathbb{R}^{2L} \times \mathbb{R}^{2L} \to \mathbb{R}^L$ be the linear projection operator defined in \eqref{eqn:projecionOperatorDef}. For a signal $X\in \R^{L}$, let $\overline{X} \in \mathbb{R}^{2L} \times \mathbb{R}^{2L}$ be defined as in \eqref{eqn:barXdef}. We consider the problem of recovering $\overline{X}$ from noisy group-transformed and projected observations of $\overline{X}$ with the form:
    \begin{equation}
    \label{eqn:generalMRA1D}
        Z_k = P(g_k \cdot \overline{X}) + \varepsilon_k, \qquad k \in \{1, 2, \ldots, M\},
    \end{equation}
    where $\{g_k\}_{k = 1}^{M} \subset G$ is a random variable in the product space $G^M = G\times \cdots \times G$, and $\{\varepsilon_k\}_{k = 1}^{M}$ is an i.i.d.\ noise sequence, $\varepsilon_k \sim \mathcal{N}(0, \sigma^2 I)$.
\end{definition}

\begin{remark}[Recovering $X$ vs.\ recovering $\overline{X}$]
\label{rem:recoveringX}
Although Definition~\ref{def:inducedMRAmodel} is stated in terms of the lifted object $\overline{X}$ from~\eqref{eqn:barXdef}, our ultimate goal is to recover the original signal $X \in \mathbb{R}^L$ (or its features $\mathcal{F}(X)$). 
By construction, the measurements determine $\overline{X}$ only up to a cyclic shift, reflecting the group action $g_k \cdot \overline{X}$. However, because $\overline{X}$ has the special form of the signal $X$ padded with zeros, this ambiguity can be resolved in a canonical way: the zero-padding breaks the symmetry and uniquely identifies $X$. 
Hence, recovery of $X$ from the observations in~\eqref{eqn:generalMRA1D} is well-defined and does not require addressing group invariance. This distinction will be important in the analysis of the mean square error (MSE) for estimators based on the induced MRA model of MTD patches.
\end{remark}

In the following theorem (proved in Appendix~\ref{sec:proof-of-thm-MTD-1D}), we show that the extracted patches $\{Z_k\}_{k=1}^M$ can be expressed in the form~\eqref{eqn:generalMRA1D}, where the sequence $\{g_k\}_{k = 1}^{M}$ forms a Markov chain over the group $G$. 

\begin{theorem}[Markov structure of the induced model]
\label{thm:MTD-1D}
Let $L, M \in \mathbb{N}$, and let $X \in \mathbb{R}^L$ be the signal to be recovered. Consider the measurement $Z \in \mathbb{R}^{LM}$ defined by the MTD model in~\eqref{eqn:MTDmodelHomogenous}, where $q \in \mathbb{N}$ denotes the number of signal occurrences and the placement parameters $\{t_i\}_{i=1}^q$ satisfy Assumption~\ref{assumption MTD}. Suppose $Z$ is partitioned into $M$ non-overlapping patches $\{Z_k\}_{k=1}^{M}$ as described in~\eqref{eqn:sequenceOneDomention}.
Then, the following holds,
\begin{enumerate}
    \item The sequence $\{ Z_k \}_{k=1}^{M}$ can be expressed in the form of Definition~\ref{def:inducedMRAmodel}, with associated group elements $\{ g_k \}_{k=1}^{M} = \{ (g_k^{(1)}, g_k^{(2)})\}_{k=1}^{M}$ belonging to the group $G = \mathbb{Z}_{2L} \times \mathbb{Z}_{2L}$.
    \item The sequence $\{g_k\}_{k=1}^{M}$ forms a time-homogeneous Markov chain on $G$.
    \item There exists a unique stationary distribution $\pi$ over $G$ such that for any $1 \leq k \leq m \leq M$ and any $g' \in G$, the total variation distance between the conditional distribution of $g_m$ given $g_k = g'$ and the stationary distribution satisfies
    \begin{equation}
    \label{exponential mixing MC thm}
        \sum_{g \in G} \left| \mathbb{P}(g_m = g \mid g_k = g') - \pi(g) \right| \leq C \cdot \left(1 - (1 - \lambda)^{L-1} \right)^{m - k},
    \end{equation}
for some constant $C > 0$ depending only on $\lambda$ and $L$.
\end{enumerate}

\end{theorem}

\subsubsection{The stationary distribution of the Markov chain }

In Lemma \ref{lem: stationary distribution MC}, we prove that the stationary distribution $\pi$ associated to the Markov chain $\{ g_k \}_{k=1}^{M}$ constructed in the proof of Theorem~\ref{thm:MTD-1D} is given explicitly by
\begin{align} \label{eqn:stationaryDistrbution}
\pi (x,y) :=
\begin{cases}
    \dfrac{(1-\lambda)^L}{1 + (L-1)\lambda} & \text{if} \ (x,y)
 = (0,L) \\
 \dfrac{\lambda (1-\lambda)^{L - x}}{1+(L-1)\lambda} & \text{if} \ 0 < x <L \ \text{and} \ y = L \\
 \dfrac{\lambda (1-\lambda)^{y}}{1+ (L-1)\lambda} & \text{if} \ x = 0 \ \text{and} \ 0 \leq  y < L \\
 \dfrac{\lambda^2 (1 - \lambda)^{y - x}}{1 + (L-1)\lambda} & \text{if} \  0 < x \leq y < L \\
 0 & \text{else,}
 \end{cases}
\end{align}
for any $g=(x,y)\in G = \Z_{2L\times 2L}$.

Note that, in the definition of $\pi$, the five cases correspond, respectively, to (see Figure~\ref{fig:1}(b) for illustration):
\begin{enumerate}
    \item the probability of a patch not containing any part of the signal $X$;
    \item the case in which the patch contains only a part of the signal $X$ remaining from the previous patch;
    \item the case in which the patch contains only a part of the signal starting in the current patch;
    \item the case in which the patch contains two parts of the signal;
    \item in the last case, if $x>y$, then two instances of the signal $X$ overlap, and thus it has probability $0$ due to Assumption \ref{assumption MTD}. The cases $x\geq L$ and $y>L$ have also probability $0$ due to the construction of the Markov chain $g_k = (g_k^{(1)}, g_k^{(2)})$ in the proof of Theorem \ref{thm:MTD-1D} (see the choice of $(g_k^{(1)},g_k^{(2)})$ in \eqref{g1 def} and \eqref{gk def}), which only takes values in $\{ 0,1, \ldots, L-1 \} \times \{ 0, 1, \ldots, L \}$.
\end{enumerate}

Theorem~\ref{thm:MTD-1D} together with the stationary distribution~\eqref{eqn:stationaryDistrbution} reveals a trade-off governed by the density parameter $\lambda$. Our reduction of the MTD problem to a Markovian MRA model relies on two key ingredients: first, the convergence of the dependent patch process to its associated stationary i.i.d.\ model; and second, the ability to recover the signal from that stationary model. On the one hand, as follows from the mixing estimate~\eqref{exponential mixing MC thm}, smaller values of $\lambda$ lead to faster convergence to stationarity; see Figure~\ref{fig:3} for an empirical illustration. On the other hand, $\lambda$ also affects the sample complexity of the stationary problem itself. Indeed, by \cite[Theorem II.2]{abbe2018estimation} and the form of the stationary distribution $\pi$ in~\eqref{eqn:stationaryDistrbution}, when $\lambda$ is small, patches contain signal less frequently, so the induced random variable $Y=P(g\cdot\overline{X})$, with $g\sim\pi$, has weaker low-order moments. As a result, more patches are needed to achieve a given estimation accuracy. Thus, although smaller $\lambda$ improves the mixing properties of the underlying Markov chain, larger $\lambda$ is generally more favorable for signal recovery because it increases the signal density in the stationary model. In other words, the main information-theoretic bottleneck is not estimating the stationary distribution $\pi$, but recovering $X$ from the induced stationary observations. We emphasize, however, that in practical applications $\lambda$ is typically determined by the data-generation mechanism and is not under direct control.

\section{The one-dimensional MTD model: Convergence to i.i.d.\ MRA and sample complexity} \label{sec:convergenceToMRA1D}

In this section, we show that the Markovian MRA model introduced in Theorem~\ref{thm:MTD-1D} converges to the classical i.i.d.\ MRA model in the limit as the number of patches tends to infinity. Building on this connection, we derive the sample complexity of the Markovian MRA model and compare it to that of its i.i.d.\ counterpart. Specifically, we use Theorem~\ref{thm:MTD-1D} to formally relate the distribution of patches extracted from the MTD measurement to the stationary distribution of the i.i.d.\ MRA model, in which observations are sampled independently.

As discussed in Section~\ref{sec:MoMandTargetMaps}, a common objective is to estimate specific features of the signal $X$ (e.g., its moments), denoted by $\mathcal{F}(X) \in \mathbb{R}^d$. Let $\widehat{F}$ be the estimator defined in~\eqref{eqn:estimatorDef}, which maps any collection of $M$ patches $\{Z_k\}_{k=1}^{M}$, each of the form~\eqref{eqn:generalMRA1D}, to an estimate of $\mathcal{F}(X)$. Recalling Definition~\ref{def:markovianAndMRAmodelsFeatures}, we define $\widehat{F}_{\mathsf{MTD}}(M, m)$ as the estimator obtained by applying $\widehat{F}$ to the subsampled sequence $\{Z_{km}\}_{k=1}^{M'}$, where $m \in \mathbb{N}$ denotes the spacing between selected patches and $M' = \lfloor M/m \rfloor$ ensures the correct index range. Likewise, recall the definition from Definition~\ref{def:markovianAndMRAmodelsFeatures} of $\widehat{F}_{\mathsf{MRA}}(M)$ as the same function applied to $M$ i.i.d.\ observations $\{Y_k\}_{k=1}^{M}$ drawn from the stationary distribution associated with the Markovian model.

Our objective is to compare the MSE of the Markovian estimator $\widehat{F}_{\mathsf{MTD}}(M, m)$ to that of the i.i.d.\ estimator $\widehat{F}_{\mathsf{MRA}}(M)$ in the asymptotic regime where $M \to \infty$.
To make this comparison precise, we now formalize the notion of a convergence rate for the estimator $\widehat{F}(\cdot)$ in the i.i.d.\ MRA setting. This provides a rigorous characterization of how the estimation error decays as a function of the sample size and noise level.

    \begin{definition}[Convergence rate] \label{def:convergenceRate1D}
    Let $M \in \mathbb{N}$, and let $\{Y_k\}_{k=1}^{M}$ be i.i.d.\ observations of the form~\eqref{eqn:generalMRA1D}, corresponding to i.i.d.\ samples $\{\tilde{g}_k\}_{k=1}^{M}$ drawn from a distribution $\pi$ over a compact group $G$, corrupted by i.i.d.\ Gaussian noise $\{\varepsilon_k\}_{k=1}^{M}$ with variance $\sigma^2$ and zero mean. Let $\widehat{F}$ be an estimator as defined in~\eqref{eqn:estimatorDef}, mapping observations to an estimate of the features $\mathcal{F}(X) \in \mathbb{R}^d$. Define the i.i.d.\ estimator as $\widehat{F}_{\mathsf{MRA}}(M) := \widehat{F}(\{Y_k\}_{k=1}^{M})$. Given a function $a:(0, \infty) \to (0,\infty)$, we say that $\widehat{F}_{\mathsf{MRA}}(M)$ has convergence rate $a(\sigma)$ if
        \[
        \frac{ \mathbb{E} \left[ \left\| \widehat{F}_{\mathsf{MRA}} (M) - \mathcal{F} (X) \right\|^2 \right] }{a(\sigma)/M} \xrightarrow[M \to \infty]{} 1.
    \]
\end{definition}

Recall that, following~\cite{abbe2018estimation}, the convergence rate $a(\sigma)$ is polynomial in the noise level $\sigma$. In standard orbit-recovery models, the optimal estimation rate is of the form $a(\sigma)=\sigma^{2n_{\min}}$, where $n_{\min}$ is the order of the lowest moment that determines the orbit of the signal under the group action.

\subsection{Convergence to the IID MRA model.}
We are now in a position to present the main theorem of the 1D MTD model (proved in Appendix~\ref{sec:proof-of-sample-complexity}), which establishes the sample complexity relationship between MTD patches and the associated MRA model with i.i.d.\ observations.

\begin{theorem}[Convergence to MRA model]
    \label{thm:sampleComplexity1DMTD}
    Let the assumptions of Theorem~\ref{thm:MTD-1D} hold. Consider two observation models:
    \begin{enumerate}
        \item Let $\{Z_k\}_{k = 1}^{M}$ denote the sequence of patches extracted from the MTD measurement $Z$ as in~\eqref{eqn:sequenceOneDomention}, which, by Theorem~\ref{thm:MTD-1D}, satisfy the Markovian MRA model~\eqref{eqn:generalMRA1D} with group elements $\{g_k\}_{k = 1}^{M}$ forming a Markov chain over a compact group $G$.
        \item Let $\{Y_k\}_{k = 1}^{M}$ be i.i.d.\ samples of the form~\eqref{eqn:generalMRA1D}, with group elements $\{\tilde{g}_k\}_{k = 1}^{M}$ drawn i.i.d.\ from the stationary distribution $\pi$ associated to the Markov chain in Theorem~\ref{thm:MTD-1D}.
    \end{enumerate}

    Let $\mathcal{F}: \R^L \to \R^d$ be a bounded continuous feature map, and let $\widehat{F}: \bigcup_{M\geq 1} (\R^{L})^M \to \R^d$ be a bounded continuous estimator map such that the estimation associated to i.i.d.\ measurements, denoted by $\widehat{F}_{\mathsf{MRA}} (M) := \widehat{F} \left[ \{ Y_k \}_{k=1}^{M}\right]$ has convergence rate $a(\sigma)$ as per Definition \ref{def:convergenceRate1D}, for some function $a:(0,\infty)\to (0,\infty)$. 
    Then the following statements hold:
    \begin{enumerate}
        \item There exists $c_0>0$ such that, for any $m=\lfloor c\log(M) \rfloor$ with $c>c_0$, and taking $M' = \lfloor M/m\rfloor$, we have
        \[
            \frac{\mathbb{E} \left[ \left\| \widehat{F} \left[ \{{Z}_{km}\}_{k=1}^{M'} \right] - \mathcal{F}(X) \right\|^2 \right]}{ a(\sigma)/M'} \longrightarrow 1, \quad \text{as } M' \to \infty.
        \]
        In particular, since $M'\sim \frac{M}{c\log (M)},$ the sample complexity in the MTD model matches that of the i.i.d.\ MRA model up to a logarithmic factor.

        \item Suppose the estimator is of the form 
        $$
        \widehat{F} \left[ \{{Y}_k\}_{k=1}^{M} \right] = \frac{1}{M} \sum_{k=1}^{M} F ({Y}_k),
        $$ 
        for some bounded continuous function $F: \R^L\to \R^d$, and assume $a(\sigma) \geq \tau$ for some constant $\tau > 0$. Then,
        \[
            \limsup_{M \to \infty} \frac{\mathbb{E} \left[\left\| \widehat{F}\left[ \{Z_{k}\}_{k=1}^{M} \right] - \mathcal{F}(X) \right\|^2 \right]}{a(\sigma)/M} \leq K,
        \]
        where $K$ is a constant depending on the mixing rate from Theorem \ref{thm:MTD-1D} (i.e., on $\lambda$ and $L$), the norm $\|F\|_\infty$, the lower bound $\tau$, and the dimension of $X$.
    \end{enumerate}
\end{theorem}

Theorem~\ref{thm:sampleComplexity1DMTD} establishes that statistical estimation from the dependent MTD model (Markovian MRA) is asymptotically equivalent to estimation under the i.i.d.\ MRA model, up to a (possibly) logarithmic overhead. The result covers two regimes, distinguished by the structure of the estimator:
\begin{enumerate}
    \item \emph{General estimators via subsampling:} In the general case where $\widehat{F}$ is not necessarily an empirical average, statistical guarantees can still be obtained by subsampling the MTD patches at logarithmic intervals. Specifically, selecting every $m = \lfloor c \log M \rfloor$-th patch ensures sufficient decorrelation due to the exponential mixing of the underlying Markov chain. The resulting subsampled sequence of size $M' \asymp M / \log M$ behaves approximately as an i.i.d.\ sample, allowing the estimator to achieve the same convergence rate (up to constants) as in the i.i.d.\ MRA model. This approach requires only boundedness and mixing, without structural assumptions on $\widehat{F}$.
    \item \emph{Empirical average estimators without subsampling:} When the estimator is an empirical mean $\widehat{F}({Z_k}) = \frac{1}{M} \sum F(Z_k)$, one can retain all $M$ patches without subsampling. This result implies that empirical average estimators can fully exploit the dependent data, achieving the same sample complexity (up to a constant factor) as if the observations were independent.
\end{enumerate}

The distinction between the two regimes is one of \emph{efficiency}. For general estimators, subsampling at logarithmic intervals is necessary for our theorem to apply: without subsampling, we cannot guarantee the desired statistical behavior, even though the estimator itself is well-defined. This restriction means that only a fraction $\sim 1/\log M$ of the available data is effectively used. By contrast, empirical average estimators can leverage all $M$ patches without any loss, achieving the same sample complexity as in the i.i.d.\ setting. Consequently, Theorem~\ref{thm:sampleComplexity1DMTD}(ii) provides strictly stronger guarantees whenever such estimators are applicable.

We emphasize that the convergence rate is defined relative to the stationary distribution of the group elements ${g_k}$ in the MTD model. In the one-dimensional case, this stationary distribution is analytically characterized in terms of the density parameter $\lambda$ and the patch length $L$ (see Theorem~\ref{thm:MTD-1D}). In higher dimensions, the spatial interactions are governed by more complex models, such as the hard-core process in 2D that we propose in sections \ref{sec:2DhomogeneousMTD} and \ref{sec:mainResults2D}, and understanding the induced distribution becomes substantially more involved. Nevertheless, the general equivalence principle established in Theorem~\ref{thm:sampleComplexity1DMTD} ensures that, under exponential mixing conditions, estimators that achieve a given statistical accuracy in the i.i.d.\ model will exhibit the same asymptotic rate when applied to MTD data, even in the presence of these dependencies.

\subsection{Implication to the method of moments and the MTD sample complexity} \label{subsec:implicationsToMoM}

A central application of Theorem~\ref{thm:sampleComplexity1DMTD}(ii) is for the method of moments, as described in Section~\ref{sec:method-of-moments}. Recall that in the i.i.d.\ setting, the $n$-th order moment of the random variable $Y \in \mathbb{R}^L$ defined in~\eqref{eqn:patch}, is defined by \eqref{eqn:MoMdefinition}. This moment can be empirically estimated using $M$ i.i.d.\ observations $\{Y_k\}_{k=1}^{M}$ via:
\begin{align}
    \widehat{T}_{\mathrm{MRA}}^{(n)}(\{Y_k\}_{k=1}^{M}) := \frac{1}{M} \sum_{k=1}^M (Y_k)^{\otimes n}.
\end{align}
Similarly, we define the $n$-th order empirical moment computed from the $M$ patches $\{Z_k\}_{k = 1}^{M}$ extracted from the MTD observation $Z$ (similarly to Theorem~\ref{thm:sampleComplexity1DMTD}(ii)), as follows:
\begin{align}
    \widehat{T}_{\mathrm{MTD}}^{(n)}(\{Z_k\}_{k=1}^{M}) := \frac{1}{M} \sum_{k=1}^M (Z_k)^{\otimes n}. \label{eqn:mtdempiricalMoMDef}
\end{align}

We now formalize the consequence of Theorem~\ref{thm:sampleComplexity1DMTD}(ii) for the method of moments.

\begin{corollary}[Method of moments for MTD]
\label{cor:MoM-MTD}
Let $n \in \mathbb{N}$, and denote by $T^{(n)}_{\mathrm{MRA}} = \mathbb{E}[Y^{\otimes n}]$ the $n$-th order population moment of the i.i.d.\ MRA model, where $Y$ is as defined in~\eqref{eqn:patch}. Let $\widehat{T}_{\mathrm{MTD}}^{(n)}$ be the empirical $n$-th order moment computed from $M$ consecutive patches $\{Z_k\}_{k=1}^{M}$ extracted from the MTD model, as defined in~\eqref{eqn:mtdempiricalMoMDef}.
Assume that the conditions of Theorem~\ref{thm:sampleComplexity1DMTD}(ii) hold, and that the moment map $(\cdot)^{\otimes n}: \mathbb{R}^L \to \mathbb{R}^{L^{\otimes n}}$ is bounded on the support of the observations.

Then:
\begin{enumerate}
    \item There exists a constant $C > 0$ such that
    \[
        \limsup_{M \to \infty} \frac{\mathbb{E} \left[\left\| \widehat{T}_{\mathrm{MTD}}^{(n)}(\{Z_k\}_{k = 1}^{M}) - T^{(n)}_{\mathrm{MRA}} \right\|^2 \right]}{a(\sigma)/M} \leq C,
    \]
    where $a(\sigma)$ is the convergence rate from Definition~\ref{def:convergenceRate1D}, and $C$ depends on $n$, $\|(T^{(n)}_{\mathrm{MRA}})^{\otimes n}\|_\infty$, and the mixing properties of the underlying Markov chain.
    \item Furthermore, suppose that $n_{\min}$ is the minimal moment order required for orbit recovery in the i.i.d.\ MRA model~\eqref{eqn:generalMRA1D}, in the sense that the orbit of $\overline{X}$ under the group action is uniquely determined by $\{T^{(n)}_{\mathrm{MRA}}\}_{n=0}^{n_{\min}}$. Then, in the low-SNR regime $\sigma^2 \to \infty$, the sample complexity of the method of moments applied to MTD data satisfies:
    \[
        M = \omega\left( \sigma^{2n_{\min}} \right),
    \]
    matching the asymptotic sample complexity of the corresponding i.i.d.\ MRA model up to constants.
\end{enumerate}
\end{corollary}

Corollary~\ref{cor:MoM-MTD} shows that the method of moments retains its asymptotic optimality when applied to dependent MTD data. Specifically, the empirical moment estimator $\widehat{T}^{(n)}_{\mathrm{MTD}}$ converges to the true i.i.d.\ moment $T^{(n)}_{\mathrm{MRA}} = \mathbb{E}[Y^{\otimes n}]$ with the same convergence rate as in the i.i.d.\ MRA model. This follows directly from Theorem~\ref{thm:sampleComplexity1DMTD}(ii), which guarantees that empirical average estimators behave similarly under Markovian and i.i.d.\ sampling.
The practical consequence is significant: for signals identifiable from moments up to order $n$, it suffices to estimate those $n$ moments using the MTD data. Thus, the method of moments applied to the MTD patches achieves the same sample complexity scaling as in the i.i.d.\ model, namely, $M = \omega(\sigma^{2n_{\min}})$ in the low-SNR regime.
We stress that in this section we focus on the number of samples required to recover the moments from empirical data. Previous results (see~\cite{abbe2018estimation,bendory2017bispectrum,bandeira2023estimation}) show how to recover the signal $X$ itself from sufficiently accurate estimators of the moments, and therefore our sample complexity results apply to estimation of $X$ itself by plugin.

\section{The two-dimensional MTD model: Preliminaries}
\label{sec:2DhomogeneousMTD}

In the two-dimensional MTD model, the observation is a large noisy image containing multiple copies of $X \in \mathbb{R}^{L \times L}$ (Figure~\ref{fig:2}(a)). A key step in extending the MTD framework to two dimensions is to specify the placement model for these signal occurrences. Our requirements are twofold: (i) copies of $X$ must not overlap, and (ii) sufficiently distant regions should behave nearly independently, so that local statistics do not induce long-range correlations. Among the various models that satisfy these conditions, the \emph{hard-core model} provides a natural and parsimonious choice. It enforces only the exclusion constraint, that is preventing overlap, while otherwise leaving placements unconstrained. This minimal assumption guarantees that the induced random field exhibits strong mixing, with correlations between distant patches decaying rapidly, thereby allowing us to treat well-separated patches as nearly independent. Such exponential mixing is the structural property required for our sample-complexity analysis. 

Originating in statistical physics, the hard-core model describes systems of finite-size, non-interpenetrating objects, where any configuration placing two objects closer than a prescribed exclusion distance has zero probability. This abstraction has broad relevance across domains, including: (i) finite-size particles in statistical mechanics~\cite{hansen2013theory}, (ii) non-overlapping targets in imaging~\cite{moller2003statistical,baddeley2016spatial}, and (iii) exclusion processes in networks and ecology~\cite{haenggi2013stochastic,wiegand2013handbook}; see also~\cite{georgii2011,weitz2006}.
While alternative spatial processes, such as Markov random fields or interacting particle systems, can also exhibit mixing, the hard-core model remains analytically tractable, conceptually simple, and widely studied.

\begin{figure}[t!]
    \centering
    \includegraphics[width=1.0\linewidth]{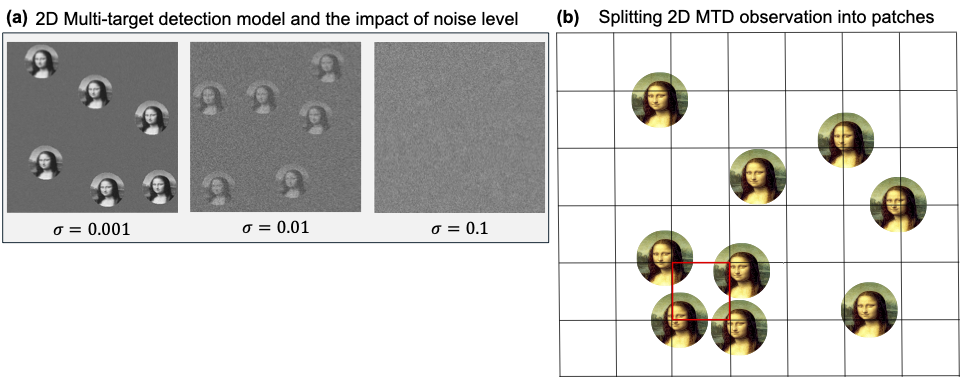}
    \caption{\textbf{The MTD model in the 2D setting.} 
    \textbf{(a)} The 2D MTD model consists of multiple instances of an image (e.g., the Mona Lisa image) embedded at random locations within a noisy 2D observation. In the low-noise (high-SNR) regime, these instances can be reliably detected and localized. However, in the high-noise regime, the signal becomes indistinguishable from the background, making accurate localization infeasible.
    \textbf{(b)} To overcome this challenge, the observation is partitioned into non-overlapping patches. These patches form a structured variant of the MRA model, where neighboring patches exhibit statistical dependencies. We model these dependencies using a hard-core interaction, ensuring that overlapping signal occurrences are avoided. Each patch may intersect up to four distinct signal instances at its corners (appearing cropped in the patch), as illustrated in the red square.}    
    \label{fig:2}
\end{figure}

\subsection{Model formulation} 
\label{sec:2D-homogeneous-MTD-formulation}

Let $X\in\mathbb{R}^{L\times L}$ be the target image  and let $Z\in\mathbb{R}^{LM\times LM}$ be the observed image. We model $Z$ as
\begin{align}\label{eq:2d-mtd}
    Z \;=\; \sum_{i=1}^{q} \bigl(S(t_i)\ast X\bigr) \;+\; \varepsilon,
\end{align}
where $\varepsilon\in\mathbb{R}^{LM\times LM}$ has i.i.d.\ $\mathcal{N}(0,\sigma^2)$ entries and $q\in\mathbb{N}$ is the number of appearances of $X$. The placement parameters $t_i$ take values in
$\{0,1,\ldots,L(M-1)\}^2$, so that an $L\times L$ copy of $X$ placed at $t_i$ lies entirely within the image. For $t\in\{0,\ldots,LM-1\}^2$, the indicator image (Kronecker delta at $t$), $S(t)\in\{0,1\}^{LM\times LM}$ is defined by $S(t)[\xi] = \mathbf{1}\{\xi=t\}$ for $\xi\in\{0,\ldots,LM-1\}^2$, and the convolution in~\eqref{eq:2d-mtd} implements a zero-padded placement (no wrap-around):
\begin{align}\label{eq:2d-translation}
    \bigl(S(t)\ast X\bigr)[\xi_1,\xi_2]
    \;=\;
\begin{cases}
X\bigl[(\xi_1,\xi_2)-t\bigr], & \text{if } (\xi_1,\xi_2)-t \in \{0,\ldots,L-1\}^2,\\[2pt]
0, & \text{otherwise}.
\end{cases}
\end{align}
Thus, each $t_i$ specifies the top-left anchor pixel of a copy of $X$ inserted into $Z$.
 
\subsection{Hard-core model}
\label{sec:hard-core-model}
The hard-core model consists of three components: (i) a graph $\mathcal{G}_M=(V_M,E_M)$, whose vertices $V_M$ represent potential placement sites for signal occurrences (i.e., the locations $\{t_i\}_{i=1}^q$) and whose edges $E_M$ encode conflicts (i.e., pairs of sites that cannot be simultaneously occupied); (ii) a collection of \emph{admissible configurations}, defined as subsets of $V_M$ that respect the non-overlap constraint; and (iii) a probability measure over admissible configurations, parameterized by an activity $\lambda>0$ that controls the expected density of occupied sites. In other words, in our setting, the placement parameters $\{t_i\}_{i=1}^q$ in~\eqref{eq:2d-mtd} are modeled by a hard-core process ensuring that no two signal copies overlap.

In our setting, the vertices of the graph are given by
\[
    V_M = \{0,1,\ldots,L(M-1)\}^2,
\]
corresponding to all possible anchor positions of $L\times L$ signal copies within the image. The edges are defined as
\[
    (t_1,t_2)\in E_M \quad \Longleftrightarrow \quad \|t_1-t_2\|_\infty \leq L-1,
\]
where, for $t_i=(t_i^{(1)},t_i^{(2)})$, we set
\[
    \|t_1-t_2\|_\infty = \max\big(|t_1^{(1)}-t_2^{(1)}|,\;|t_1^{(2)}-t_2^{(2)}|\big).
\]
Thus, two vertices are adjacent whenever their associated $L\times L$ patches overlap.  

A configuration is a subset $\eta\subseteq V_M$. It is \emph{admissible} if no two occupied vertices are adjacent, i.e., if $(t_j,t_k)\notin E_M$ for all $j\neq k$. Equivalently, admissibility requires that any two occupied sites be separated by at least $L$ pixels in either the horizontal or vertical direction, ensuring that no two copies of the signal overlap. We will formalize this property next.

\begin{definition} [Admissible configuration]
A configuration $\eta=\{t_i\}_{i=1}^q \subseteq V_M$ is \emph{admissible} if it spans no edges in $\mathcal{G}_M$, i.e.
\begin{align}\label{eq:hc-admissibility}
    (t_j,t_k)\notin E_M, \qquad \forall j\neq k.
\end{align}
Equivalently, if the placement parameters in the 2D MTD model $\eta=\{t_i\}_{i=1}^q\subseteq V_M$ denotes the occupied set, an admissible configuration requires 
\begin{align}
    \|t_j-t_k\|_\infty \ge L, \qquad \text{for all} \quad j\neq k.
\end{align}
\end{definition}

Finally, the hard-core process is defined as a probability measure supported only on admissible configurations. For $\lambda > 0$, the probability of observing $\eta \subseteq V_M$ is  
\begin{align} \label{eqn:partition-function}
    \mu_\lambda(\eta) \;=\; \frac{1}{Z_\lambda}\,\mathbf{1}\{\eta \text{ admissible}\}\,\lambda^{|\eta|}, 
    \qquad 
    Z_\lambda \;=\; \sum_{\tau\subseteq V_M}\mathbf{1}\{\tau \text{ admissible}\}\,\lambda^{|\tau|}.
\end{align}
Here, $\lambda$ (the activity or fugacity parameter) controls the density of occupied sites: larger values favor denser configurations. In this formulation, all admissible subsets of the same cardinality receive equal weight, proportional to $\lambda^{|\eta|}$, while inadmissible ones have probability zero. In our setting, the graph $\mathcal{G}_M$ is chosen so that admissibility corresponds precisely to excluding overlapping $L \times L$ signal copies in the image. We next summarize this formally in Definition~\ref{def:hardCoreModelGraph}.

\begin{definition} [The MTD hard-core model] \label{def:hardCoreModelGraph}
      For $L\in \N$ fixed, and any $M\in \N$, let $\mathcal{G}_M = \p{V_M,E_M}$ be the MTD graph with set of vertices
     \begin{equation}
        \label{graph vertex}
        V_M := \{ 0, 1, \ldots, L(M-1)  \}^2,
    \end{equation}
    and set of edges $E_M \subset V_M \times V_M$, such that for every $(x_1, y_1) \neq (x_2, y_2)$, the following is satisfied,
        \begin{equation}
            \label{graph edges}
            \left( (x_1, y_1) ,\,  (x_2, y_2) \right) \in E_M \ \text{whenever} \ \max \{ |x_1 - x_2|, \, |y_1-y_2|  \} \leq L-1.
        \end{equation} 
    For any subset $\eta = \{ t_i \}_{i=1}^q$ of $V_M$, we say that a vertex $v\in V_M$ is occupied whenever $v=t_i$ for some $i = 1, \ldots, q$.

    For any $\lambda > 0$, the hard-core model over the graph $\mathcal{G}_M$ with parameter $\lambda$ is a probability measure over the subsets of $V_M$, such that, for any $q\in \{ 0 ,\ldots, |V_M|\}$, any subset $\eta = \{ t_i \}_{i=1}^q$ with $q$ elements has probability
    $$
    \mu_\lambda \left( \eta \right) = \frac{1}{Z_\lambda} \lambda^q \mathds{1}\{ \eta \ \text{is admissible} \},
    $$
    where $Z_\lambda$ is defined in \eqref{eqn:partition-function}, and the admissibility of $\eta$ is given by the condition \eqref{eq:hc-admissibility}. 
\end{definition}

The admissibility condition \eqref{eq:hc-admissibility} ensures the non-overlap of $L\times L$ signal footprints.
Analogously to Assumption~\ref{assumption MTD} in 1D, we assume signal placements $\{t_j\}_{j=1}^q$ are drawn from a hard-core model. 

\begin{assumption}[Hard-core model for distribution of placements]
\label{assumption hard-core}
There exists $\lambda>0$ such that the placement parameters $\{t_i\}_{i=1}^q$ in~\eqref{eq:2d-mtd} are distributed according to the hard-core model on $\mathcal{G}_M$ with activity $\lambda$ from Definition \ref{def:hardCoreModelGraph}.
\end{assumption}

At low activity $\lambda$, short-range repulsion yields weak long-range dependence and, under suitable conditions, exponential mixing~\cite{dobrushin1968problem,georgii2011}, which we formalize next.

\subsection{Random fields and exponential mixing}
\label{sec:exponential-mixing}

The reduction of the 2D MTD problem to an MRA model follows the same idea as in the 1D case (see Section~\ref{sec:1dimentionalMTDresults}): we partition the measurement $Z$ into non-overlapping patches $Z_k$. In two dimensions, it is convenient to index patches by their spatial location in the original image. For $M\in\mathbb{N}$, define
\begin{align}
\mathcal{I}_M \;:=\; \{1,\ldots,M\}^2 \;=\; \bigl\{\,k=(k_1,k_2):\ k_1,k_2\in\{1,\ldots,M\}\,\bigr\},
\end{align}
so that a total of $M^2$ disjoint $L\times L$ patches are extracted. Because the index set is two-dimensional, the collection $\{Z_k\}_{k\in\mathcal{I}_M}$ cannot be modeled as a hidden Markov process. Instead, we treat the group elements $\{g_k\}_{k\in\mathcal{I}_M}$ acting on the target signal as a \emph{random field} on $\mathcal{I}_M$.

\begin{definition}[Random field and admissibility] \label{def:random-field}
Let $G$ be a finite set and $\mathcal{I}\subset\mathbb{Z}^2$. A $G$-valued random field indexed by $\mathcal{I}$ is a random variable $g=\{g_k\}_{k\in\mathcal{I}}$ which takes values in $ G^{\mathcal{I}}$. For any index subset $S\subseteq\mathcal{I}$, we define the restriction of $g$ to $S$ as $g_{|S}=\{g_k\}_{k\in S}$, which takes values in $G^S$. We say that a configuration $\Psi\in G^S$ is \emph{admissible} if $\mathbb{P}(g_{|S}=\Psi)>0$.
\end{definition}

\begin{definition}[Exponential mixing]
\label{def:exponential-mixing}
Let $\mathcal{I}\subset \Z^2$ be a two-dimensional index set, and define the complement of the index set as $\mathcal{I}^c := \Z^2\setminus \mathcal{I}$. A $G$-valued random field $\{g_k\}_{k\in\mathcal{I}}$ satisfies \emph{exponential mixing} if there exist $c_1,\gamma>0$ and a probability distribution $\pi$ on $G$ such that for every $k\in\mathcal{I}$, $S\subseteq\mathcal{I}$, and admissible $\Psi\in G^S$,
\begin{align} \label{eq:mixing-property}
\max_{\varphi\in G}\big|\mathbb{P}(g_k=\varphi\mid g_{|S}=\Psi)-\pi(\varphi)\big|
    \;\leq\; c_1\exp\!\Big(-\gamma\cdot\mathrm{dist}(k,S\cup\mathcal{I}^c)\Big),
\end{align}
where, for $d(k,k')=\displaystyle\max_{j=1,2}|k_j-k'_j|$, we define $\mathrm{dist}(k,A)=\displaystyle\min_{a\in A}d(k,a)$. 
\end{definition}

Intuitively, exponential mixing ensures that the statistical influence of distant sites decays exponentially, so well-separated patches behave nearly independently.
The distance to the complement of the index set in the right hand side of \eqref{eq:mixing-property} ensures that the statistical influence of the measurement boundary also decays exponentially as the patch is away from the edges of the image.

\section{The two-dimensional MTD model: Main results}
\label{sec:mainResults2D}

Analogously to Section~\ref{sec:mainResults1D}, we present the main results for the two-dimensional hard-core MTD model introduced in the previous section. In Section~\ref{sec:sec:2dimentionalMTDhardCoreForming}, we show that an MTD measurement can be partitioned into $L\times L$ patches, each of which has the form of a two-dimensional MRA observation (see Figure~\ref{fig:2}). Although these patches are not i.i.d., they are generated by an underlying random field that exhibits exponential mixing, governed by the distribution placement of the hard-core model  (Assumption~\ref{assumption hard-core}). In Section~\ref{sec:convergenceTo2DMRA}, we prove that this dependent MRA model converges to the classical i.i.d.\ MRA model; consequently, the sample complexity of the induced 2D hard-core model can be analyzed by comparison with the i.i.d.\ MRA framework.

\subsection{From MTD to a non-i.i.d.\ MRA model with exponential mixing} \label{sec:sec:2dimentionalMTDhardCoreForming}

As in the 1D case, we partition the measurement $Z\in\mathbb{R}^{LM\times LM}$ (see \eqref{eq:2d-mtd}) into $M^{2}$ non-overlapping $L\times L$ patches. Let $\mathcal{I}_M\triangleq\{1,\ldots,M\}^2$ and, for $k=(k_1,k_2)\in\mathcal{I}_M$, define
\begin{equation}
    \label{eq:2d-patches}
    Z_{k}[\xi]\;=\;Z\bigl[(k_1-1)L+\xi_1,\ (k_2-1)L+\xi_2\bigr],
    \qquad \forall\,\xi=(\xi_1,\xi_2)\in\{0,\ldots,L-1\}^2.
\end{equation}

Under Assumption~\ref{assumption hard-core} on the locations $\{t_i\}_{i=1}^q$ in \eqref{eq:2d-mtd} (which forbids overlap between signal occurrences), each patch $Z_k\in\mathbb{R}^{L\times L}$ can contain portions of at most four instances of the signal $X\in\mathbb{R}^{L\times L}$.

To model these contributions, we introduce the zero–padded signal
\begin{equation}
\label{X tilde 2D}
\tilde{X}\;\triangleq\;
\begin{bmatrix}
X & 0_{L\times L}\\[2pt]
0_{L\times L} & 0_{L\times L}
\end{bmatrix}\in\mathbb{R}^{2L\times 2L},
\end{equation}
and define the (componentwise) circular–shift action of $\mathbb{Z}_{2L}\times\mathbb{Z}_{2L}$ on $\mathbb{R}^{2L\times 2L}$ by
\begin{equation}
\label{group action 2d simple}
(\alpha,\beta)\cdot\tilde{X}\,[\xi_1,\xi_2]
\;=\;\tilde{X}\bigl[(\xi_1-\alpha,\ \xi_2-\beta)\!\!\!\!\!\pmod{2L}\,\bigr],
\qquad \forall\,(\xi_1,\xi_2)\in\{0,\ldots,2L-1\}^2,
\end{equation}
for any $(\alpha, \beta)\in \Z_{2L}\times \Z_{2L}$.

Since up to four shifted copies may contribute to a single patch (see Figure~\ref{fig:2}), set
\begin{equation}
\label{eqn:barXdef 2D}
\overline{X}\;\triangleq\;(\tilde{X},\tilde{X},\tilde{X},\tilde{X})\in(\mathbb{R}^{2L\times 2L})^{4},
\end{equation}
and let $G\triangleq(\mathbb{Z}_{2L}\times\mathbb{Z}_{2L})^{4}$ act componentwise on $(\mathbb{R}^{2L\times 2L})^{4}$:
\begin{align}
    \label{eq:group-action-2d-double}
    g\cdot(\tilde{X}_0,\tilde{X}_1,\tilde{X}_2,\tilde{X}_3)
    &=\bigl(g^{(1)}\cdot\tilde{X}_0,\ g^{(2)}\cdot\tilde{X}_1,\ g^{(3)}\cdot\tilde{X}_2,\ g^{(4)}\cdot\tilde{X}_3\bigr),
\end{align}
for $g=(g^{(1)},g^{(2)},g^{(3)},g^{(4)})\in G$.

Finally, define the linear operator $P:(\mathbb{R}^{2L\times 2L})^{4}\to\mathbb{R}^{L\times L}$ by
\begin{equation}
\label{P def 2D}
P(\tilde{X}_0,\tilde{X}_1,\tilde{X}_2,\tilde{X}_3)
\;=\;P_{00}\tilde{X}_0\;+\;P_{01}\tilde{X}_1\;+\;P_{10}\tilde{X}_2\;+\;P_{11}\tilde{X}_3,
\end{equation}
where $P_{00},P_{01},P_{10},P_{11}:\mathbb{R}^{2L\times 2L}\to\mathbb{R}^{L\times L}$ extract, respectively, the top–left, top–right, bottom–left, and bottom–right $L\times L$ blocks. Namely, for
\[
\tilde{X}=\begin{bmatrix}A & B\\ C & D\end{bmatrix},\quad A,B,C,D\in\mathbb{R}^{L\times L},
\]
we have $P_{00}\tilde{X}=A$, $P_{01}\tilde{X}=B$, $P_{10}\tilde{X}=C$, and $P_{11}\tilde{X}=D$.

We are now in a position to introduce the MRA model that describes the patches $Z_k$, defined in \eqref{eq:2d-patches}, extracted from the MTD measurement $Z\in \R^{LM\times LM}$ in \eqref{eq:2d-mtd}. 

\begin{definition} [The induced hard-core two-dimensional MRA model]
\label{def:induced MRA 2D}
    Let $X\in \R^{L\times L}$ be the target signal, and let $\overline{X}\in (\R^{2L\times 2L})^4$ be given by \eqref{eqn:barXdef 2D}.
    Let us consider the index set $\mathcal{I}_M:= \{ 1, \ldots, M\}^2$, and the finite abelian group $G := (\Z_{2L} \times \Z_{2L})^4$ acting on $ (\R^{2L\times 2L})^4$ as described in \eqref{eq:group-action-2d-double}. Let $P: (\R^{2L\times 2L})^4 \to \R^{L\times L}$ be the linear operator defined in \eqref{P def 2D}.
    We consider the problem of recovering $\overline{X}$ from noisy group-transformed and projected observations of $\overline{X}$ with the form
    \begin{equation}
    \label{eq:general-mra}
    Z_k = P (g_k\cdot \overline{X})
        + \varepsilon_k, \qquad k\in \mathcal{I}_M,
    \end{equation}
    where $\{g_k\}_{k\in \mathcal{I}_M}$ is a random field over $G$ indexed by $\mathcal{I}_M$, and $\{ \varepsilon_k\}_{k\in \mathcal{I}_M}$ is i.i.d.\ Gaussian noise.
\end{definition}

The next result proves that the patch sequence ${Z_k}$ defined in \eqref{eq:2d-patches} admits the MRA representation of Definition~\ref{def:induced MRA 2D}, with latent group elements forming an exponentially mixing random field. The construction of this random field relies on the assumption that the placement of the signal occurrences $\{t_i\}_{i=1}^q$ is governed by a hard-core model on a 2D graph, where the vertices are the pixels of the image (Assumption \ref{assumption hard-core}). We refer to this model as the hard-core MRA model. The proof is given in Appendix~\ref{sec:proof-of-2d-converence}.

\begin{theorem}[Induced MRA model with exponential mixing]
\label{thm:2d-main-result}
    Let $L, M \in \mathbb{N}$ and let $X \in \mathbb{R}^{L \times L}$. Let the measurement $Z \in \mathbb{R}^{LM \times LM}$ be defined as in~\eqref{eq:2d-mtd}, with position parameters $\{t_i\}_{i=1}^q$ satisfying Assumption~\ref{assumption hard-core} for some $\lambda \in (0,1)$. Suppose that $Z$ is partitioned into non-overlapping patches $\{Z_k\}_{k \in \mathcal{I}_M}$ as in~\eqref{eq:2d-patches}.
    Then, the patch sequence $\{Z_k\}_{k \in \mathcal{I}_M}$ can be expressed in the form~\eqref{eq:general-mra}, as in Definition~\ref{def:induced MRA 2D}, where $\{g_k\}_{k \in \mathcal{I}_M}$ is a random field over the group $G$. Moreover, there exists a constant $\lambda_0 \in (0,1)$, depending only on $L$, such that if $\lambda \in (0, \lambda_0)$, then the random field $\{g_k\}_{k \in \mathcal{I}_M}$ satisfies the exponential mixing property (Definition~\ref{def:exponential-mixing}) for some probability distribution $\pi$ over $G$.
\end{theorem}

In contrast to the one-dimensional Markov case (cf.\ \eqref{eqn:stationaryDistrbution}), the distribution $\pi$ does not admit a simple closed form in 2D. Instead, it is constructed in the proof of Theorem~\ref{thm:2d-main-result} as follows:
\begin{enumerate}
    \item \emph{Local marginal and encoding.}
    In Step 2 of the proof of Theorem~\ref{thm:2d-main-result}, in Section \ref{sec:proof-of-2d-converence}, we express the group elements $\{g_k\}_{k\in \mathcal{I}_M}$ associated to each patch as a function of the marginal of the hard-core process on a $2L\times 2L$ neighborhood of the corresponding patch.
    \item \emph{Existence of a Gibbs measure in the infinite graph $\Z^2$.} In Step 3 we prove that, by Corollary~2.6 of~\cite{weitz2006} (see also Lemma \ref{lem: strong mixing hard-core} in Section \ref{sec:proof-of-2d-converence}), if $\lambda$ is sufficiently small (with a threshold depending only on the graph connectivity, determined by $L$), the hard–core process on the finite graph $\mathcal{G}_M$ (see Definition~\ref{def:hardCoreModelGraph}) satisfies a strong spatial mixing property, which implies the existence of a unique Gibbs measure $\mu$ on the infinite graph $\mathbb{Z}^2$. The hard-core model can equivalently be expressed as the restriction of the Gibbs measure to the finite set of vertices $V_M$, conditionally on having no particles in the complement $\Z^2\setminus V_M$.
    \item \emph{Construction of $\pi$.} In Step 4, we combine the two previous points to construct the distribution $\pi$ as a function of the marginal of the Gibbs measure $\mu$ in a $2L\times 2L$ finite graph. Let $V_0\triangleq\{-L,\ldots,L-1\}^2$ and let $\mu_0$ be the marginal of $\mu$ on $V_0$ (so $|V_0|=(2L)^2$). The distribution $\pi$ is defined as the pushforward of $\mu_0$ under the encoding map $\Phi:\{0,1\}^{V_0}\to G$, from (i), that associates each admissible configuration $\eta$ with the corresponding group element in $G$. 
\end{enumerate}
Under this construction, we can prove that the random field $\{g_k\}_{k\in \mathcal{I}_M}$ satisfies the exponential mixing property in Definition \ref{def:exponential-mixing} with respect to $\pi$.

As a noticeable difference with respect to the 1D case, in the 2D case we require the activity parameter $\lambda$ to be smaller than a constant $\lambda_0$, which depends on the degree of the graph defining the hard core model.
In \cite[Corollary~2.6]{weitz2006}), the threshold $\lambda_0$ is given explicitly as $\lambda_0 = \frac{b^b}{(b-1)^{b+1}}$, where $b+1$ is the degree of the graph (the maximum number of neighbors of each node), and one can deduce the simpler lower estimate $\lambda_0 \geq \frac{1}{b-1}$.
In view of Definition~\ref{def:hardCoreModelGraph}, the degree of the graph is $(2L-1)^2-1$, hence
$$
\lambda_0 \geq \dfrac{1}{(2L - 1)^2 - 2}.
$$
The result in  \cite[Corollary~2.6]{weitz2006}) applies to general graphs, and we are not able to determine whether this condition is sharp for the specific type of graph that we are considering here. The necessity of the smallness condition on the activity parameter $\lambda$ to ensure exponential mixing is left as an open problem.

\subsection{Convergence to i.i.d.\ MRA and sample complexity} \label{sec:convergenceTo2DMRA}

This section is the two–dimensional analogue of Section~\ref{sec:convergenceToMRA1D}.
We show that the induced 2D MRA model of Definition~\ref{def:induced MRA 2D}, obtained by extracting patches $\{Z_k\}_{k\in\mathcal{I}_M}$ from an MTD measurement $Z$, matches the standard i.i.d.\ MRA reconstruction problem with patches $\{Y_k\}_{k\in\mathcal{I}_M}$ of the form~\eqref{eq:general-mra}, where the group elements $\{\tilde g_k\}_{k\in\mathcal{I}_M}$ are i.i.d.\ from the distribution $\pi$ specified in Theorem~\ref{thm:2d-main-result}.

Analogously to the 1D setting, we reuse the feature map, estimators, and convergence-rate notion for the 2D model.
Let $\mathcal{F}:\mathbb{R}^{L\times L}\to\mathbb{R}^d$ be the feature map from Section~\ref{sec:MoMandTargetMaps}, adjusted to a domain of 2D images with size $L \times L$. For any estimator function $\widehat{F}$ acting on collections of patches, define (analogously to Definition~\ref{def:markovianAndMRAmodelsFeatures}),
\begin{equation}
    \label{eqn:2d-mra-and-mtd-est}
    \widehat{F}_{\mathsf{MRA}}(M)\;\triangleq\;\widehat{F}\!\left[\{Y_k\}_{k\in\mathcal{I}_M}\right],
    \qquad
    \widehat{F}_{\mathsf{MTD}}(M,m)\;\triangleq\;\widehat{F}\!\left[\{Z_{km}\}_{k\in\mathcal{I}_{\lfloor M/m\rfloor}}\right],
\end{equation}
where $\{Y_k\}$ are MRA observations and $\{Z_k\}$ are 2D MTD patches sampled every $m$ sites in each coordinate; both aim to estimate $\mathcal{F}(X)$.
Due to the strong spatial mixing of the underlying hard-core model (see Lemma \ref{lem: strong mixing hard-core}) for the spatial placement of signal occurrences, subsampling with $m$ large enough ensures enough decorrelation between the selected patches.

The convergence-rate notion mirrors the 1D case (see Definition \ref{def:convergenceRate1D}) with the natural replacement of the sample size by the number of 2D patches, i.e., we say that $\widehat{F}_{\mathsf{MRA}} (M)$ has convergence rate $a(\sigma)$ if
\begin{equation}
    \frac{\mathbb{E}\bigl[\|\widehat{F}_{\mathsf{MRA}}(M)-\mathcal{F}(X)\|^2\bigr]}{a(\sigma)/M^2}
    \;\xrightarrow[M\to\infty]{}\;1.
    \label{eqn:2d-convergence-rate}
\end{equation}

The next theorem (proved in Appendix~\ref{sec:proof-of-sample-complexity}) establishes the connection, in terms of MSE and sample complexity, between the induced MRA model from Theorem \ref{thm:2d-main-result} and the associated i.i.d.\ MRA model.

\begin{theorem} [Convergence to two-dimensional MRA model]
    \label{thm:2d-general-result}
    Under the assumptions of Theorem \ref{thm:2d-main-result}, 
    \begin{enumerate}
        \item let $\{ Z_k \}_{k \in \mathcal{I}_M}$ be the patches extracted from $Z$ as in \eqref{eq:2d-patches}, which, by virtue of Theorem \ref{thm:2d-main-result}, can be written as in \eqref{eq:general-mra}, with $\{ g_k \}_{k\in \mathcal{I}_M}$ a random field over $G$ satisfying the exponential mixing property from Definition \ref{def:exponential-mixing};
        \item let $\{ Y_k \}_{k \in \mathcal{I}_M}$ be measurements of the form \eqref{eq:general-mra}, with $\{ \tilde{g}_k \}_{k\in \mathcal{I}_M}$ i.i.d.\ from the probability distribution $\pi$ over $G$ from Theorem \ref{thm:2d-main-result}.
    \end{enumerate}
    Let $\mathcal{F}: \R^{L\times L}\to \R^d$ be a bounded continuous feature map, and let $\widehat{F}(\cdot)$ be a bounded continuous estimator map such that the estimator $\widehat{F}_{\mathsf{MRA}} (M)$ defined in \eqref{eqn:2d-mra-and-mtd-est} has convergence rate $a(\sigma)$, i.e., \eqref{eqn:2d-convergence-rate} holds for some function $a: (0,\infty) \to (0,\infty)$.
    Then, the following statements hold:
\begin{enumerate}
\item There exists $c_0>0$ such that, for any $m=\lfloor c\log(M) \rfloor$ with $c>c_0$, and taking $M' = \lfloor M/m\rfloor$, we have
$$
    \frac{\mathbb{E} \left[ \left\|
\widehat{F} \left[ \{Z_{km}\}_{k\in\mathcal{I}_{M'}} \right] - \mathcal{F} (X)
    \right\|^2 \right]}{ a(\sigma)/(M')^2} \longrightarrow 1, \quad \text{as} \ M'\to \infty.
    $$
In particular, since $(M')^2\asymp\frac{M^2}{(c\log(M))^2}$, the sample complexity of dependent data matches the sample complexity in the i.i.d.\ case up to a logarithmic factor squared.
\item Suppose $\widehat{F}(\cdot)$ is of the form $\widehat{F} \left[ \{{Y}_{k}\}_{k\in\mathcal{I}_{M}} \right] = \frac{1}{M^2} \sum_{k\in \mathcal{I}_M} F ({Y}_k)$ for some bounded continuous function $F:\R^{L\times L}\to \R^d$, and assume $a(\sigma)\geq \tau$ for some $\tau>0$. Then, 
\[ \limsup_{M\to \infty} \frac{ \mathbb{E} \left[\left\| \widehat{F}\left[ \{Z_{k}\}_{k\in\mathcal{I}_{M}} \right] - \mathcal{F} (X) \right\|^2\right]}{a(\sigma)/M^2}\leq K,\] 
for a constant $K$ that depends only on $\lambda$, $L$, $\norm{F}_\infty,$ $\tau$ and the dimension of $X$.
\end{enumerate}    
\end{theorem}

Theorem~\ref{thm:2d-general-result} extends all one–dimensional results to the two–dimensional MTD setting. Under the hard–core process assumption, the $L\times L$ patches $\{Z_k\}_{k\in\mathcal{I}_M}$ admit the induced 2D MRA representation with latent group elements $\{g_k\}$ forming an exponentially–mixing random field. Consequently, any estimator that attains MSE convergence rate $a(\sigma)$ in the i.i.d.\ 2D MRA model achieves the same rate on MTD data up to a squared logarithmic overhead arising from decorrelation: choosing a grid spacing $m\simeq C\log M$ in \emph{each} coordinate yields an effective sample size $M'^2\asymp M^2/(\log M)^2$, so the gap to i.i.d.\ is $(\log M)^2$.

As in 1D, statement (ii) in Theorem~\ref{thm:2d-general-result} shows that empirical–average estimators (when applicable) need no subsampling and thus exploit all $M^2$ patches. Their MSE continues to scale as $a(\sigma)/M^2$ under exponential mixing, matching the i.i.d.\ rate up to constants.
The same implication holds for the method of moments: empirical $n$th–order moments computed from the 2D MTD patches converge to their i.i.d.\ counterparts at rate $a(\sigma)/M^2$. If $n_{\min}$ is the minimal moment order needed for orbit recovery in the i.i.d.\ model, then in the low-SNR regime the required \emph{number of patches} $N=M^2$ satisfies $N=\omega(\sigma^{2n_{\min}})$, identical scaling to the i.i.d.\ case.

\section{Numerical experiments}
\label{sec:empirical_simulation}

In Section~\ref{sec:convergenceToMRA1D}, we established a precise statistical connection between the MTD model, its Markovian MRA formulation, and the corresponding induced i.i.d.\ MRA model. In particular, Theorem~\ref{thm:sampleComplexity1DMTD} and Corollary~\ref{cor:MoM-MTD} show that, for empirical-average estimators such as empirical moments, the sample complexity in the MTD/Markovian MRA setting matches that of the induced i.i.d.\ model up to constant factors. Moreover, this connection relies on the convergence of the latent Markov chain to its stationary distribution, which governs the induced i.i.d.\ model and underlies the moment-based transfer of sample-complexity bounds.

The numerical experiments in this section serve three purposes. First, we empirically validate the stationary distribution and quantify the convergence of the latent Markov chain to stationarity. Second, we provide evidence that moments up to order $3$ suffice for recovery in the induced i.i.d.\ model, whereas second-order information alone is insufficient in the regimes tested here, thereby supporting the conclusion that $n_{\min}=3$. Third, we compare the observable moments of the Markovian and induced i.i.d.\ models under additive noise, illustrating the predicted scaling laws.

\subsection{Convergence to the stationary distribution}

We next validate the explicit stationary distribution $\pi$ in~\eqref{eqn:stationaryDistrbution} empirically, and quantify how quickly the Markov chain $\{g_k\}_{k=1}^M$ approaches stationarity as the number of patches increases. 
For an observation composed of $M$ patches, we form the empirical distribution
\begin{align}
    \widehat{\pi}_M(g) \triangleq\frac{1}{M}\sum_{k=1}^{M}\mathbf{1}\{g_k=g\},
    \qquad g\in \{0,\ldots,L-1\}\times\{0,\ldots,L\},
    \label{eq:pi_hat_def}
\end{align}
where, in the simulations, the latent states $g_k=(g_k^{(1)},g_k^{(2)})$ are computed from the ground-truth signal start times $\{t_i\}$ used to generate the MTD measurement (see the construction in Theorem~\ref{thm:MTD-1D} and the discussion in Section~\ref{sec:splitingIntoPatches}). 
We quantify the discrepancy between $\widehat{\pi}_M$ and the theoretical stationary distribution $\pi$ using the total variation distance
\begin{align}
    \mathrm{TV}\!\left(\widehat{\pi}_M,\pi\right)
    \;\triangleq\;
    \frac{1}{2}\sum_{g}\left|\widehat{\pi}_M(g)-\pi(g)\right|.
    \label{eq:TV_def}
\end{align}
To connect convergence of $\widehat{\pi}_M$ to the method-of-moments viewpoint (Section~\ref{sec:method-of-moments}), we also track convergence at the level of induced \emph{noiseless} moments. Specifically, we form the empirical moment estimator $\widehat{T}_{\mathrm{MTD}}^{(2)} \triangleq \frac{1}{M}\sum_{k=1}^{M} \big(Z_k\big)^{\otimes 2}$, as defined in~\eqref{eqn:mtdempiricalMoMDef}, for the observations $\{Z_k\}_{k=1}^M$ defined in~\eqref{patches-1d} (without the noise term), and compare them to the corresponding noiseless population moments $T^{(2)}_{\overline{X},\pi}$ defined in~\eqref{eqn:MoMdefinition}, which are evaluated exactly by summation over $\supp(\pi)$.
For each $(\lambda,M)$, we repeat the simulation over multiple independent Monte-Carlo trials and report the empirical means of $\mathrm{TV}(\widehat{\pi}_M,\pi)$ and of the MSE of $\mathbb{E}\big\|\widehat{T}_{\mathrm{MTD}}^{(2)}-T^{(2)}_{\overline{X},\pi}\big\|_F^2$.

Figure~\ref{fig:3} summarizes the results.
Panel~(a) plots the Monte-Carlo average of $\mathrm{TV}(\widehat{\pi}_M,\pi)$ as a function of $M$, for several values of $\lambda$.
Across all tested $\lambda$ values, the curves exhibit an approximately linear decay with slope close to $-1/2$ on log--log axes, indicating the scaling
\begin{align}
    \mathbb{E}\,\mathrm{TV}\!\left(\widehat{\pi}_M,\pi\right)
    \;\asymp\;
    \frac{C(\lambda,L)}{\sqrt{M}},
    \label{eq:TV_scaling_empirical}
\end{align}
for a constant $C(\lambda,L)$ depending on the mixing properties of the chain (and therefore on $\lambda$ and $L$).
Panel~(b) report the corresponding moment estimation error and exhibit a $M^{-1}$ scaling for the moment MSEs:
\begin{align}
    \mathbb{E}\Big\|\widehat{T}_{\mathrm{MTD}}^{(2)}-T^{(2)}_{\overline{X},\pi}\Big\|_F^2 \;\asymp\; \frac{C_2(\lambda,L,X)}{M},    \label{eq:moment_MSE_scaling_empirical}
\end{align}
with constants depending on $(\lambda,L)$ (and on the fixed signal $X$). In addition, it can be seen that a smaller $\lambda$ enhances the convergence to the stationary distribution, as discussed in Section~\ref{sec:1dimentionalMTDresults}.
These empirical observations are consistent with the exponential mixing bound in~\eqref{exponential mixing MC thm}: as $M$ grows, the effective number of approximately independent samples increases linearly with $M$, yielding  $M^{-1/2}$ fluctuations for TV distance and the corresponding $M^{-1}$ scaling for MSE of the moment estimation.

\begin{figure}[t]
    \centering
    \includegraphics[width=0.85\linewidth]{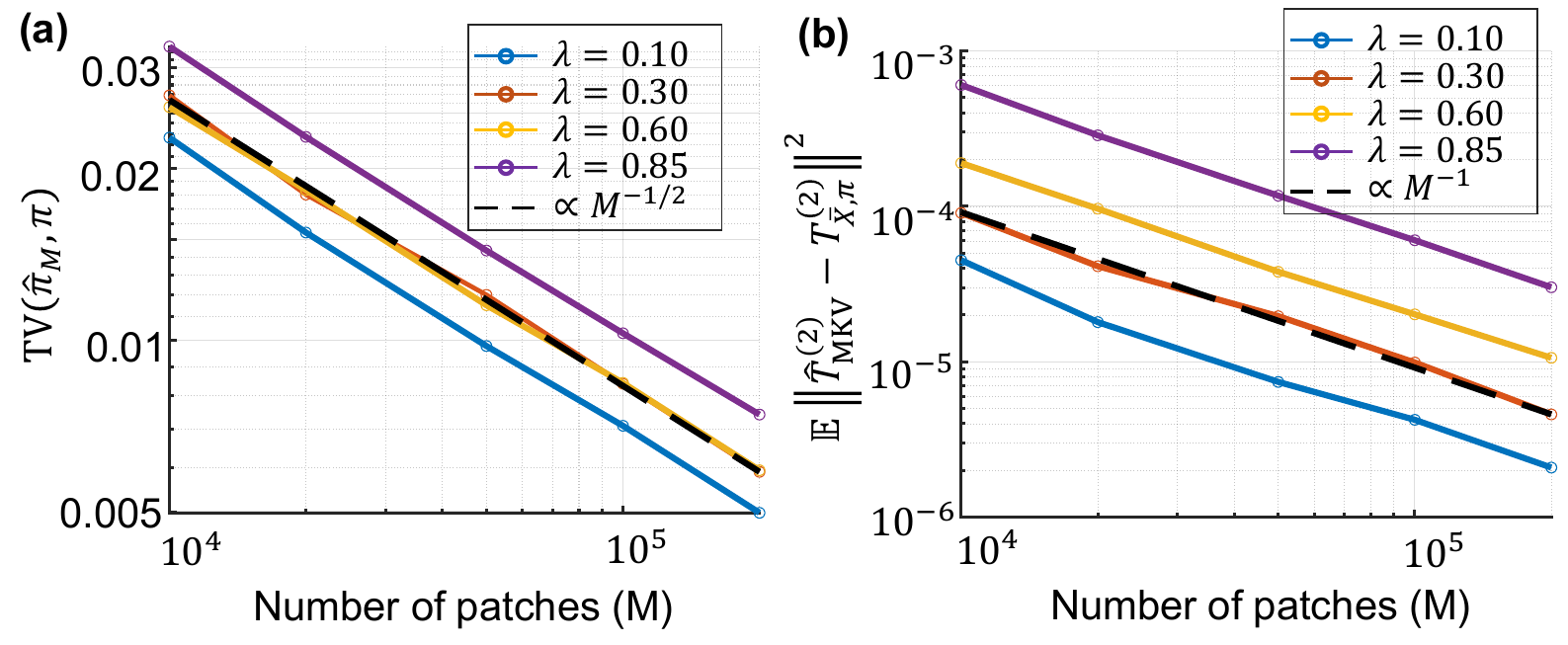}   \caption{\textbf{Empirical convergence to stationarity and induced moment scaling.}
    \textbf{(a)} For each $\lambda$, we simulate the Markov chain $\{g_k\}$ induced by the 1D MTD patching procedure (Theorem~\ref{thm:MTD-1D}), estimate its stationary distribution by $\widehat{\pi}_M$ in~\eqref{eq:pi_hat_def}, and report the Monte-Carlo mean of the total variation distance $\mathrm{TV}(\widehat{\pi}_M,\pi)$ in~\eqref{eq:TV_def}.  \textbf{(b)} Monte-Carlo means of  $\mathbb{E}\|\widehat{T}^{(2)}-T^{(2)}_{\overline{X},\pi}\|_F^2$, where $T^{(2)}_{\overline{X},\pi}$ denotes the noiseless population moments and $\widehat{T}_{\mathrm{MTD}}^{(2)}$ the corresponding empirical noiseless estimators. The dashed reference lines indicate slopes $-1/2$ in (a) and $-1$ in (b), consistent with the observed $M^{-1/2}$ scaling for TV and $M^{-1}$ scaling for moment MSE.}
    \label{fig:3}
\end{figure}

\subsection{Empirical identification of the minimal moment order in the induced i.i.d.\ MRA model}
\label{subsec:empirical_setup}

This subsection provides empirical evidence regarding the moment cut-off order $n_{\min}$ of the induced i.i.d.\ MRA model associated with the MTD construction. In particular, our goal is to examine whether recovery is already possible from moments up to order $3$, or whether higher-order moment information may be necessary.

To isolate identifiability from additive noise and finite-sample effects, we work in the noise-free population setting and evaluate the noiseless group-transformed moments~\eqref{eqn:MoMdefinition}. Specifically, let $X_\star\in\R^{L}$ be the ground-truth signal, and let $\overline{X}_\star=(\tilde X_\star,\tilde X_\star)$ denote its lifted representation, where $\tilde X_\star=(X_\star,0_L)$. For a candidate $X\in\R^L$ (with $\overline{X}$ defined analogously), we denote by $T^{(2)}_{\overline{X},\pi}$ and $T^{(3)}_{\overline{X},\pi}$ the corresponding induced noiseless moments, as defined in~\eqref{eqn:MoMdefinition}.

In our experiments, we generate a ground-truth signal $X_\star$, with i.i.d.\ Gaussian entries, compute the target second-order and third-order moments $\big(T^{(2)}_{\overline{X}_\star,\pi},\,T^{(3)}_{\overline{X}_\star,\pi}\big)$ via~\eqref{eqn:MoMdefinition}, and then attempt recovery by moment matching using either $T^{(2)}_{\overline{X}_\star,\pi}$ alone or the pair $\big(T^{(2)}_{\overline{X}_\star,\pi},\,T^{(3)}_{\overline{X}_\star,\pi}\big)$. More precisely, we estimate $X$ by solving the nonconvex moment-matching problem
\begin{align}
    \widehat{X} \;\in\;
    \arg\min_{X\in\R^{L}} \Bigg\{
    \frac{\big\|T^{(2)}_{\overline{X},\pi}-T^{(2)}_{\overline{X}_\star,\pi}\big\|_F^2}
    {\big\|T^{(2)}_{\overline{X}_\star,\pi}\big\|_F^2} \;+\; w_3\,    \frac{\big\|T^{(3)}_{\overline{X},\pi}-T^{(3)}_{\overline{X}_\star,\pi}\big\|_F^2}
    {\big\|T^{(3)}_{\overline{X}_\star,\pi}\big\|_F^2}
    \Bigg\},
    \label{eq:empirical_obj_m2_m3}
\end{align}
where $\|\cdot\|_F$ denotes the Frobenius norm. We consider both the case $w_3>0$ (specifically, $w_3=3$) and the second-moment-only variant $w_3=0$.

We minimize~\eqref{eq:empirical_obj_m2_m3} using \texttt{fminunc} of Matlab with multiple random restarts. Each restart is initialized with $X^{(0)}\sim\mathcal{N}(0,I_L)$ and then normalized to satisfy $\|X^{(0)}\|_2=1$, thereby fixing the global scale. We run $n_{\mathrm{restarts}}$ independent initializations and report the solution attaining the smallest objective value. Since the optimization problem is non-convex, this procedure does not provide a theoretical guarantee of convergence to the global optimum, even with multiple restarts. Nevertheless, similar non-convex moment-matching formulations, solved via local optimization with random initialization, are commonly used in closely related recovery problems; see, for example,~\cite{boumal2018heterogeneous,bendory2017bispectrum}. The reference implementation is provided in \href{https://github.com/AmnonBa/sample-complexity-mkv-mtd}{https://github.com/AmnonBa/sample-complexity-mkv-mtd}.

The results are shown in Figure~\ref{fig:4}. When matching moments up to order $3$ with $w_3=3$, the optimizer consistently returns an estimate $\widehat{X}$ that matches the target moments to numerical precision and achieves a small recovery error. By contrast, when matching only the second moment ($w_3=0$), the optimizer typically converges to alternative solutions that fit $T^{(2)}_{\overline{X}_\star,\pi}$ well but still deviate substantially from $X_\star$, consistent with non-identifiability from second-order information alone in this induced model. We observe the same qualitative behavior across a range of density parameters $\lambda$ and for multiple signal lengths $L$, using randomly generated ground-truth signals. This indicates that the phenomenon is robust and not tied to a specific choice of $\lambda$ or dimension. Overall, these experiments provide empirical evidence that moments up to order $3$ suffice for recovery in the induced i.i.d.\ model in the regimes tested here, thereby supporting the conclusion that $n_{\min}=3$.

Based on these empirical observations, we make the following conjecture.

\begin{conjecture}[Sample complexity of the induced i.i.d.\ MRA model]
\label{conj:sample_complexity_induced_MRA}
Consider the induced i.i.d.\ MRA model associated with the one-dimensional MTD construction, under the assumptions of Section~\ref{sec:mainResults1D}. Then, for a generic signal $X_\star$, the  moment cut-off order is  $n_{\min}=3$.
Consequently, in the low-SNR regime, recovery of $X_\star$ up to the natural inherent group action requires sample complexity $M=\omega(\sigma^{6})$.
\end{conjecture}

\begin{figure}[t]
    \centering
    \includegraphics[width=1.0\linewidth]{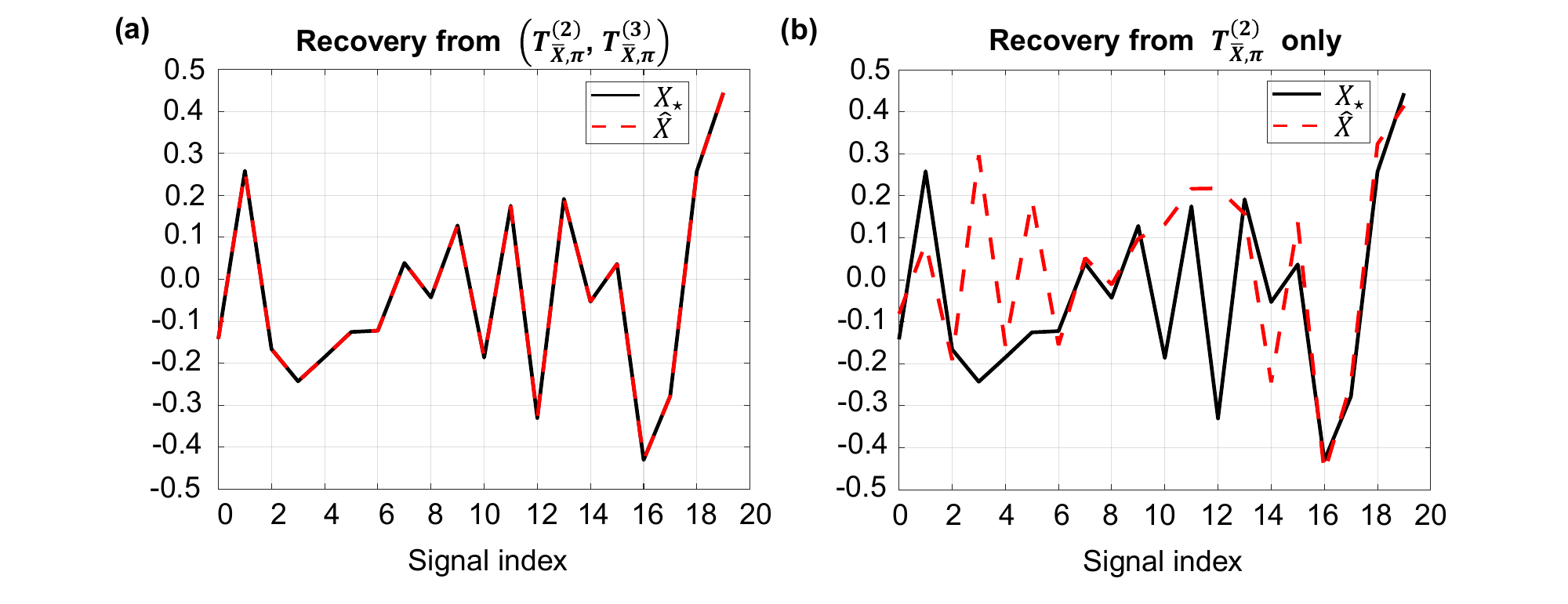}    \caption{\textbf{Recovery from second- and third-order moments versus second-order moments only.}
    Noise-free population experiment in the induced i.i.d.\ MRA model.
    The black curve shows the ground-truth signal $X_\star$, and the red dashed curve shows the recovered estimate $\widehat{X}$ after optimal cyclic alignment.
    \textbf{(a)} Using both $\big(T^{(2)}_{\overline{X}_,\pi},T^{(3)}_{\overline{X},\pi}\big)$ yields accurate recovery.
    \textbf{(b)} Using $T^{(2)}_{\overline{X},\pi}$ alone leads to an ambiguous reconstruction that can match $T^{(2)}_{\overline{X},\pi}$ well but deviates substantially from $X_\star$, illustrating the need for third-order information in this induced model. In this simulation, we use $\lambda=0.3$, although similar behavior is observed for other values of $\lambda$.}
    \label{fig:4}
\end{figure}

\subsection{Empirical Markovian-i.i.d.\ moment discrepancy under additive noise}
\label{subsec:empirical_markov_vs_iid_moments}

This subsection empirically illustrates the connection established in Section~\ref{sec:convergenceToMRA1D} by quantifying, at the level of observable moments, how close the dependent MTD patch process (equivalently, the Markovian induced MRA model of Theorem~\ref{thm:MTD-1D}) is to the induced i.i.d.\ MRA model with stationary distribution~$\pi$.
Whereas Theorem~\ref{thm:sampleComplexity1DMTD} guarantees i.i.d.-like variance scaling for empirical averages under Markovian sampling (up to constants determined by mixing), here we explicitly measure the finite-$M$ discrepancy between the two models, including its dependence on the noise level $\sigma$.

\subsubsection{Models and moment estimators.}
Fix $L$ and $\lambda$, and let $\pi$ be the stationary distribution in~\eqref{eqn:stationaryDistrbution}.
We generate $M$ consecutive MTD patches $\{Z_k\}_{k=1}^M$ under Assumption~\ref{assumption MTD} with additive noise $\varepsilon_k\sim\mathcal{N}(0,\sigma^2 I_L)$, and compute empirical moments $\widehat T^{(n)}_{\mathsf{MTD}}(M)$ for $n \in \{2,3\}$, as defined in~\eqref{eqn:mtdempiricalMoMDef}.
We then compare these to the population moments of the induced i.i.d.\ MRA model with $(g_k^{(1)},g_k^{(2)})\sim\pi$,
\[
    T^{(n)}_{\mathsf{MRA}}(\sigma)\;\triangleq\;\mathbb{E}\!\left[Y^{\otimes n}\right],\qquad n\in\{2,3\},
\]
which we evaluate in the population limit by exact summation over $\supp(\pi)$ together with the standard Gaussian corrections for additive noise.
We report the MSE between these moments by
\begin{align}
    \mathcal{E}_n(M,\sigma) \triangleq \mathbb{E}\!\left[\big\|\widehat T^{(n)}_{\mathsf{MTD}}(M)-T^{(n)}_{\mathsf{MRA}}(\sigma)\big\|_F^2\right],
    \qquad n\in\{2,3\}, \label{eq:cross_moment_error_def}
\end{align}
approximated via Monte-Carlo averaging over independent trials.

\subsubsection{Observed scaling laws.}
Figure~\ref{fig:5} reports $\mathcal{E}_2(M,\sigma)$ and $\mathcal{E}_3(M,\sigma)$ as functions of $M$ for several noise levels. Three expected phenomena are visible.
First, for fixed $\sigma$, the curves decay with slope close to $-1$ on log-log axes, consistent with the $1/M$ variance scaling for empirical averages, in agreement with Theorem~\ref{thm:sampleComplexity1DMTD}(ii). Second, the separation across noise levels matches the standard noise-dominated moment-estimation scaling,
\begin{align}
    \mathcal{E}_2(M,\sigma)\;\propto\;\frac{\sigma^{4}}{M},
    \qquad    \mathcal{E}_3(M,\sigma)\;\propto\;\frac{\sigma^{6}}{M},
\end{align}
as the noise increases, as discussed in Section~\ref{sec:method-of-moments}.
Third, at very low noise the curves exhibit a nonzero floor attributable to finite-$M$ convergence of the Markov chain $\{g_k\}$ to stationarity: since $T^{(n)}_{\mathsf{MRA}}(\sigma)$ is defined under the stationary law $\pi$, any residual deviation of the state distribution from $\pi$ induces a bias in the Markovian moments, consistent with the total-variation convergence behavior in Figure~\ref{fig:3}(a).

\begin{figure}[t]
    \centering
    \includegraphics[width=1.0\linewidth]{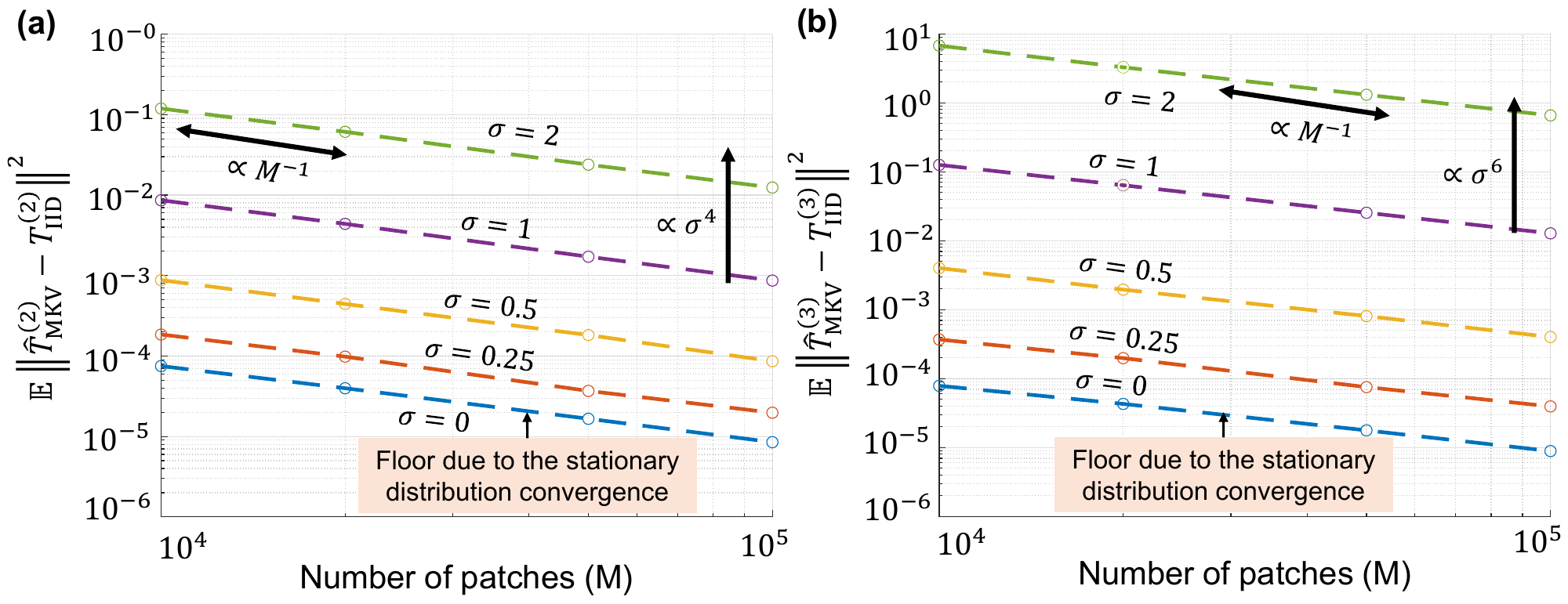}
    \caption{\textbf{Markovian-i.i.d.\ cross-model moment discrepancy under additive noise.}
    We compare the empirical moments computed from $M$ consecutive MTD patches (Markovian induced MRA), $\widehat{T}^{(n)}_{\mathsf{MTD}}(M)=\frac{1}{M}\sum_{k=1}^{M}(Z_k)^{\otimes n}$, to the population moments of the induced i.i.d.\ model under the stationary distribution $\pi$, $T^{(n)}_{\mathsf{MRA}}(\sigma)=\mathbb E[Y^{\otimes n}]$, and plot the Monte-Carlo estimates of $\mathcal{E}_n(M,\sigma)=\mathbb E\|\widehat T^{(n)}_{\mathsf{MTD}}(M)-T^{(n)}_{\mathsf{MRA}}(\sigma)\|_F^2$ for $n=2$ (panel a) and $n=3$ (panel b), across several noise levels $\sigma$.
    The dashed lines indicate the empirical $M^{-1}$ decay (variance of empirical averages), while the vertical separation between curves is consistent with the noise-dominated scaling $\mathcal{E}_2\propto \sigma^4/M$ and $\mathcal{E}_3\propto \sigma^6/M$.
    At very low noise, the curves exhibit a nonzero floor attributable to finite-$M$ convergence of the Markov chain state distribution to $\pi$ (cf.\ Figure~\ref{fig:3}(a)).}
    \label{fig:5}
\end{figure}

\section{Discussion and outlook} \label{sec:outlook}

\subsection{Comparison with maximum-likelihood approaches}
A natural alternative under the fully specified MTD model is maximum-likelihood estimation based on the marginal likelihood, obtained by integrating out the latent particle locations. While statistically appealing, this formulation is computationally challenging in our setting: even after marginalization, evaluating the likelihood requires summing over a combinatorially large collection of admissible latent configurations, which grows rapidly with the problem size. Existing related approaches therefore rely on approximate expectation-maximization / approximate marginal-likelihood formulations on simplified patch-level models rather than on exact maximum-likelihood~\cite{kreymer2022approximate, kreymer2023stochastic}. Such approximations can be effective computationally, but they still lead to nonconvex optimization problems and presently lack the kind of explicit low-SNR sample-complexity guarantees derived here. By contrast, our approach avoids direct likelihood optimization by reducing the problem to an induced MRA model, thereby enabling analysis through low-order moments and yielding explicit statistical guarantees.

\subsection{Sample complexity for the MTD model}

In this work, we establish a formal connection between the MTD model and the classical i.i.d.\ MRA problem. In the one-dimensional setting, the MTD model is cast as a Markovian variant of MRA, while in two dimensions, it is cast as  a hard-core interaction model. We show that the MSE in the Markovian or hard-core setting can be upper bounded by that of the corresponding i.i.d.\ MRA model. The resulting i.i.d.\ MRA model falls within the standard MRA paradigm, where the sample complexity is known to depend on the minimal moment that uniquely determines the orbit of the underlying signal.
While this minimal moment has been explicitly characterized for several standard MRA models, determining it for the i.i.d.\ model induced by our Markovian construction remains an open problem. In Section~\ref{sec:empirical_simulation}, however, we address this question empirically and provide evidence that moments up to order three suffice for recovery in the induced i.i.d.\ MRA model, resulting a sample complexity of $M = \omega (\sigma^6)$. A rigorous characterization of the moment cut-off for this induced model remains an interesting direction for future work, and would lead to a sharper and fully explicit sample-complexity bound for the MTD model.

\subsection{MTD with an algebraic structure}

A natural extension of the MTD model is to consider the case where the signal instances are related through an underlying algebraic structure. In this setting, referred to as the MTD model with algebraic structure, multiple transformed copies of a base signal are embedded at unknown locations within a long, noisy observation $Z$~\cite{bendory2019multi}. Specifically, let $g_i \in G$ be i.i.d.\ group elements sampled from a distribution $\rho$ over a compact group $G$, and define $X_i = g_i \cdot X$, where the group action $\cdot$ encodes the transformation applied to the signal. The observed signal is modeled as
\begin{equation}
Z = \sum_{i=0}^{N-1} S(t_i) * X_i + \varepsilon, \label{eqn:MTDmodelGeneral}
\end{equation}
where $*$ denotes linear convolution, $S(t_i)$ are unknown location indicators determining the positions of each transformed signal copy, and $\varepsilon \sim \mathcal{N}(0, \sigma^2 I_{LM \times LM})$ represents additive Gaussian noise. As in the homogeneous case (see\eqref{eqn:MTDmodelHomogenous}), the goal is to estimate the signal based on the noisy superposition of its transformed instances.
Due to the invariance of the statistical distribution of $Z$ under the action of $G$, i.e., $Z$ has the same distribution whether the underlying signal is $X$ or $g \cdot X$ for any $g \in G$, the task reduces to recovering the $G$-orbit of $X$, namely the set $\{g \cdot X \mid g \in G\}$.

Adapting the current framework to this more general algebraic setting introduces several conceptual challenges. First, the model may involve continuous rather than discrete signals, necessitating a reformulation of the Markov property and the hard-core interaction model in a continuous domain. Second, the composition of group actions under convolution does not necessarily yield a group structure, which complicates both modeling and analysis. One possible resolution is to restrict attention to a suitable subset of the product group, but this requires ensuring that the resulting structure remains compact.

\subsection{Connection to cryo-EM and cryo-ET}

A prominent example of the MTD model with algebraic structure arises in cryo-EM, which can be viewed as a specific instance of the model in~\eqref{eqn:MTDmodelGeneral} with the addition of a tomographic projection. In cryo-EM, the observed data $Z \in \mathbb{R}^{LM \times LM}$ consists of a 2D micrograph containing multiple occurrences of images $X_i \in \mathbb{R}^{L \times L}$, each given by
\[
X_i = P(g_i \cdot X),
\]
where $X$ is the underlying 3D volume to be estimated, $g_i$ are random elements of the three-dimensional special orthogonal group $G = \mathsf{SO}(3)$, and $P$ denotes a tomographic projection operator, which may also incorporate additional linear effects such as the microscope's point spread function and discretization.
A key challenge in this setting is that the reduction technique developed in this work, from the dependent (e.g., Markovian of hard-core) model to an i.i.d.\ MRA model, applies only after the group action and projection have been applied. That is, the observed cropped instances take the form
\[
P_2\left(g_i^{(2)} \cdot P_1\left(g_i^{(1)} \cdot X\right)\right),
\]
where $g_i^{(1)}$ and $g_i^{(2)}$ are distinct group actions and $P_1$, $P_2$ are projection operators. This nested composition of group actions and projections does not naturally fit into the classical MRA framework, and thus falls outside the scope of the current analysis.

While extending our methodology to the cryo-EM setting poses challenges, primarily due to the complexities introduced by projection imaging, a more tractable and closely related direction arises in cryo-electron tomography (cryo-ET), particularly in the context of the \emph{sub-tomogram averaging} problem. In this setting, tomographic projections are absent, and the observations more closely resemble the raw signal model studied in this work, making it a promising candidate for applying the techniques developed here.

\section*{Acknowledgments}
T.B. is supported by the BSF grant no. 2020159, by the NSF-BSF grant no. 2024791, and the ISF grant no. 1924/21. C.E.-Y. is supported by the Ram\'on y Cajal 2022 grant RYC2022-035966-I.

\bibliographystyle{abbrv}
\bibliography{mybibfile}

\begin{appendix}
\section{Proofs}
\subsection{Proof of Theorem \ref{thm:MTD-1D}} \label{sec:proof-of-thm-MTD-1D}
We prove the Theorem in three steps: 
\begin{enumerate}
    \item First we prove that every patch $Z_k$ can be written as
    $$
    Z_k = P \left( g_k\cdot \overline{X} \right) + \varepsilon_k
    $$
    for some $g_ k = (g_k^{(1)}, g_k^{(2)}) \in \Z_{2L}\times \Z_{2L}$ and $\varepsilon_k\in \R^L$. Recall that $\overline{X} = ( \tilde{X}, \,  \tilde{X})$, where $\tilde{X} = (X, 0_L)$.
    \item Then, we prove that Assumption \ref{assumption MTD} implies that the sequence $\{g_k\}_{k = 1}^{M}$ is a Markov Chain over $\Z_{2L}\times \Z_{2L}$.
    \item Finally, we conclude the proof by proving that the aforementioned Markov Chain has a positive absolute spectral gap, and deduce explicitly the convergence estimate to the stationary probability distribution. The explicit expression for the stationary distribution is obtained in Lemma \ref{lem: stationary distribution MC}.
\end{enumerate}

\textbf{Step 1:} Since the patches $\{Z_k\}_{k=1}^M$ into which we divide the measurement $Z$ have length $L$, any patch $Z_k$ can contain, at most, two parts of the signal $X$ (note that Assumption \ref{assumption MTD} implies that there cannot be two signals overlapping): one or both of a piece of $X$ remaining from the previous patch $Z_{k-1}$, and a piece of a signal starting in the current patch $Z_k$.

Recalling the projection operator $P: \R^{2L} \times \R^{2L} \to \R^L$ from \eqref{eqn:projecionOperatorDef}, and the group action of $\Z_{2L}\times \Z_{2L}$ on $\R^{2L} \times \R^{2L}$ defined in \eqref{eqn:cyclicGroupActionDef}, we can write
$$
P \left( g_k\cdot \overline{X} \right) = P \left( g_k^{(1)} \cdot \tilde{X}, \, g_k^{(2)} \cdot \tilde{X} \right) = P_0 \left( g_k^{(1)} \cdot \tilde{X} \right)
+ P_1 \left( g_k^{(2)} \cdot \tilde{X} \right).
$$
Recall that $P_0$ is the linear operator that extracts the last $L$ components of any vector in $\R^{2L}$, whereas $P_1$ extracts the first $L$ components.

If $Z_k$ contains the remainder of a signal which started in the previous patch, let's say at the $l$-th pixel for some $l \in \{1, \ldots, L-1\}$, then we set $g_k^{(1)} = l$, i.e.\ $g_k^{(1)} \cdot \tilde{X}$ is the signal $\tilde{X}$, shifted $l$ pixels to the right. In this way, $P_0 \left[ g_k^{(1)} \cdot \tilde{X} \right]$ contains the remainder of the signal that started in the previous patch. If, on the contrary, $Z_k$ has no signal remaining from the previous patch, then we take $g_k^{(1)}  = 0$. In this case, $P_0 \left[ g_k^{(1)} \cdot \tilde{X} \right] = P_0 \tilde{X} = 0_L$ is the vector of zeros.

If there is a new signal appearance starting in the patch $Z_k$, let's say at the $l$-th pixel for some $l \in \{ 0, \ldots, L-1 \}$, then we take $g_k^{(2)} = l$. In this way, $P_1 \left[ g_k^{(2)}\cdot  \tilde{X} \right]$ contains the first part of the signal (the first $L-l$ pixels), which appears in the patch $Z_k$.
In case there is no signal starting in the patch $Z_k$, we can take $g_k^{(2)} = L$, which yields $g_k^{(2)} \cdot \tilde{X} = [0_L, X]$, and therefore $P_1 \left[ g_k^{(2)} \cdot \tilde{X} \right] = 0_L$ is the vector of zeros.  

Adding up the terms $P_0 \left[ g_k^{(1)} \cdot \tilde{X} \right]$ and $P_1 \left[ g_k^{(2)} \cdot \tilde{X} \right]$ along with the vector of i.i.d.\ Gaussian noise $\varepsilon_k\in \R^L$, we can represent all the possible configurations for a patch $Z_k$. 

\textbf{Step 2:} The next step is to prove that the sequence $\{g_k\}_{k = 1}^{M} =  (g_k^{(1)}, g_k^{(2)})_{k = 1}^{M}$ is a Markov Chain over $\Z_{2L}\times \Z_{2L}$.

Let $\{ \omega_k \}_{k\geq 0}$ be the Markov Chain over $\{0, \ldots, L\}$ such that the first term $\omega_0$ is given by
\begin{equation}
\label{MC first term}
\omega_0 = 0,
\end{equation}
and the matrix of transition probabilities is
\begin{equation}
\label{matrix of prob transitions}
P = \begin{bmatrix}
\lambda & (1-\lambda)\lambda & (1-\lambda)^2\lambda &  \cdots & (1-\lambda)^{L-1}\lambda & (1-\lambda)^L \\
0 & \lambda & (1-\lambda)\lambda & \cdots & (1-\lambda)^{L-2}\lambda & (1-\lambda)^{L-1} \\
\vdots & \vdots & \vdots & \ddots & \vdots & \vdots \\
0 & 0 & 0 & \cdots & \lambda & (1-\lambda) \\
\lambda & (1-\lambda)\lambda & (1-\lambda)^2\lambda &  \ldots & (1-\lambda)^{L-1}\lambda & (1-\lambda)^L
\end{bmatrix}.
\end{equation}
We shall prove that each random variable $ (g_k^{(1)}, g_k^{(2)})$ in the sequence $\{ (g_k^{(1)}, g_k^{(2)}) \}_{k = 1}^{M}$ can be written as a deterministic function of $(\omega_{k-1}, \omega_k)$, for each $k\in \{ 1,2,\ldots, M\}$, extracted from the Markov Chain $\{\omega_k\}_{k\geq 0}$.

Let us start with the first term of the sequence. Since $Y_1$ contains no remainder from previous patches, we have $g_1^{(1)} = 0$.
By Assumption \ref{assumption MTD}, for any $l\in \{0, \ldots, L-1\}$, the probability of the first signal $X$ appearance starting in the $l$-th pixel of the measurement is $(1-\lambda)^l \lambda$,
and the probability of it starting in a pixel beyond the $(L-1)$-th is  $(1-\lambda)^L$. 
Therefore, $g_1^{(2)}$ is the random variable in $\{0, \ldots, L\}$ with probabilities given by the first row in the matrix \eqref{matrix of prob transitions}.
Hence, we take
\begin{equation}
\label{g1 def}
g_1 = \left( g_1^{(1)}, g_1^{(2)}\right) = (\omega_0, \omega_1),
\end{equation}
i.e., $g_1$ is the random variable in $\Z_{2L}\times \Z_{2L}$ given by the first two terms of the Markov chain $\{\omega_k\}_{k\geq 0}$.

For the second patch, $Y_2$, if $g_1^{(2)} = 0$ or $g_1^{(2)} = L$, then $Y_2$ contains no signal remainder from patch $Y_1$, and then we take $g_2^{(1)} = 0$. In this case, $g_2^{(2)}$ follows the same distribution as $g_1^{(2)}$, which corresponds to the first or the last rows in the matrix $P$, depending on whether $g_1^{(2)} = 0$ or $g_1^{(2)} = L$. In case $0 < g_1^{(2)} <L$, we should take $g_2^{(1)} = g_1^{(2)} = \omega_1$, since $P_0 [ \omega_1 \cdot \tilde{X} ]$ selects the second half of the shifted vector $\omega_1 \cdot \tilde{X}$.
Concerning $g_2^{(2)}$, in view of Assumption \ref{assumption MTD}, we need to take $g_2^{(2)} = \omega$, where $\omega$ follows a Geometric distribution starting from the first empty pixel after the signal remainder from the previous patch.
By the construction of the Markov Chain $\{\omega_k\}_{k\geq 1}$ (see the matrix of transition probabilities $P$ in \eqref{matrix of prob transitions}), we can take $g_2^{(2)} = \omega_2$.

In general, applying a recurrent argument, for any $k\geq 2$, we can take
\begin{equation}
g_k^{(1)} = \begin{cases}
    \omega_{k-1} & \text{if} \ \omega_{k-1} \in \{0, \ldots, L-1\} \\
    0 & \text{if} \ \omega_{k-1} = L,
\end{cases}
\quad \text{and} \quad
g_k^{(2)} = \omega_k.
\label{gk def}
\end{equation}
In this way, we have proved that, for any $k\geq 1$, the probability distribution $g_k = (g_k^{(1)}, g_k^{(2)})$ is given by a deterministic function of the random variables $\omega_k$ and $\omega_{k-1}$. Hence, since $\{\omega_k\}_{k\geq 1}$ is a Markov Chain, so is $\{ ( g_k^{(1)}, g_k^{(2)}) \}_{k\geq 1}$.

\textbf{Step 3:} Next, we prove that the Markov Chain $\{ \omega_k\}_{k\geq 0}$ defined in \eqref{MC first term}--\eqref{matrix of prob transitions} has a positive absolute spectral gap. For this, it is enough to verify that the Markov Chain is indecomposable and aperiodic (see \cite[Section 4]{rosenthal1995convergence}).
These two properties hold since, in view of the transition matrix $P$ in \eqref{matrix of prob transitions}, the state $\omega = L$ is accessible (with positive probability) from any other state, and from this state, the Markov Chain can transition (with positive probability) to any other state.

This argument proves the existence of a unique stationary distribution $\rho$ on $\{0, \ldots , L\}$, and that, if we call $\mu_k(\omega | \omega_m )$, with $m,k\geq 0$ the probability distribution of $\omega_{m+k}$ conditionally on $\omega_m$, then the probability distribution $\mu_k(\cdot | \omega_m)$ converges to $\rho$ as $k$ tends to infinity.
Using a coupling argument (see for instance \cite[Section 6]{rosenthal1995convergence}) and the uniform minorization condition from \cite[Section 6.2]{rosenthal1995convergence}, we can obtain explicitly the following convergence estimate: 
\begin{equation}
\| \mu_k(\cdot | \omega_m) - \rho\| = \dfrac{1}{2} \sum_{\omega = 0}^L | \mu_k (\omega | \omega_m) - \rho (\omega) | \leq (1 - \beta)^k, \qquad \text{for all} \ \omega_m\in \{ 0, \ldots , L \},
\label{convergence MC proof}
\end{equation}
where $\beta$ is given in terms of the matrix of probability transition $P$ in \eqref{matrix of prob transitions} by 
$$
\beta = \sum_{j = 0}^L \min_i P_{ij} = (1-\lambda)^{L-1} \lambda + (1-\lambda)^L = (1-\lambda)^{L-1}. 
$$
Using this convergence estimate for the Markov chain $\{\omega_k\}_{k\geq 0}$, we deduce a similar convergence result for the Markov chain $\{ (\omega_k, \omega_{k-1}) \}_{k\geq 1}\subset \{ 0, \ldots, L\}^2$, with a stationary distribution denoted by $\tilde{\rho}$.
Using the construction of $g_k$ in terms of $(\omega_k, \omega_{k-1})$ given in \eqref{g1 def} and \eqref{gk def}, we obtain the convergence estimate in  Theorem \ref{thm:MTD-1D}(iii) for the Markov chain $\{g_k\}_{k\geq 1}\subset G$. See Lemma \ref{lem: stationary distribution MC} for the explicit computations.
The stationary distribution of the Markov chain $\{g_k\}_{k\geq 1}$ is denoted by $\pi$, and is simply the push-forward measure of $\tilde\rho$ by the map in $\{ 0,\ldots, L \}^2\to G$ defined in \eqref{gk def}.
\hfill $\square$

\begin{lemma}
    \label{lem: stationary distribution MC}
    Let $\{ \omega_k \}_{k\geq 0} \subset \{ 0, 1, \ldots , L \}$ be the Markov chain given by $\omega_0 = 0$ and the matrix of probability transitions \eqref{matrix of prob transitions}, and let $\{g_k\}_{k\geq 1} = \{(g_k^{(1)}, g_k^{(2)})\}_{k\geq 1}\subset G=\Z_{2L}\times \Z_{2L}$ be given in terms of $\{\omega_k\}_{k\geq 1}$ by \eqref{g1 def} and \eqref{gk def}. Then $\{g_k\}_{k\geq 1}$ is a Markov chain over $G$ and, for any $1\leq k\leq m$ and $g'\in G$, there exists a constant $C>0,$ depending on $\lambda$ and $L$, such that
    $$
    \sum_{g\in G}\left| \operatorname{Pr} (g_m = g | g_k=g') - \pi (g) \right| \leq C \left( 1 - (1-\lambda)^{L-1} \right)^{m-k},
    $$
    where $\pi$ is the probability distribution over $G$ given, for any $g=(x,y)\in G$, by
    $$
    \pi (x,y) :=
    \begin{cases}
        \dfrac{(1-\lambda)^L}{1 + (L-1)\lambda} & \text{if} \ (x,y)
 = (0,L) \\
 \dfrac{\lambda (1-\lambda)^{L - x}}{1+(L-1)\lambda} & \text{if} \ 1 \leq x <L \ \text{and} \ y = L \\
 \dfrac{\lambda (1-\lambda)^{y}}{1+ (L-1)\lambda} & \text{if} \ x = 0 \ \text{and} \ 0\leq y < L \\
 \dfrac{\lambda^2 (1 - \lambda)^{y - x}}{1 + (L-1)\lambda} & \text{if} \ 1 \leq x \leq y < L \\
 0 & \text{else.}
 
 \end{cases}
    $$
\end{lemma}

\begin{proof}
    Proving that $\{g_k\}_{k\geq 1}$ is a Markov chain over $G$ was done in Step 2 of the proof of Theorem \ref{thm:MTD-1D}, using the choice of $g_k^{(1)}$ and $g_k^{(2)}$ in \eqref{g1 def} and \eqref{gk def}. Recall that each $g_k$ depends only on $(\omega_k, \omega_{k-1})$, and it is easy to prove that $(\omega_k, \omega_{k-1})_{k\geq 1}$ is a Markov chain over $\{ 0,1, \ldots , L\}^2$. 
    
    It is well known that the stationary distribution of the Markov chain $\{\omega_k\}_{k\geq 0}$, that we represent by the vector $\rho := (\rho (0), \ldots , \rho(L))\in [0,1]^{L+1}$, is the normalized left-eigenvector of the matrix of transition probabilities associated to the eigenvalue $1$. In this case, by solving the system of equations $\rho P = \rho$, with $P$ given in \eqref{matrix of prob transitions}, we obtain
    \begin{equation}
    \rho(x) := \begin{cases}
        \dfrac{\lambda}{ 1 + (L-1) \lambda} & \text{for} \ x = 0, 1, \ldots , L-1 \\
        \dfrac{1-\lambda}{ 1+ (L-1)\lambda} & \text{for} \ x = L.
    \end{cases}
    \label{stationary rho explicit}
    \end{equation}

    Let us now define the probability distribution $\tilde{\rho}$ over $\{ 0,1, \ldots ,L\}^2$ given by
    \begin{equation}
    \label{stationary rho tilde explicit}
    \tilde{\rho} (\omega_k,\omega_{k-1}) = P(\omega_{k-1}, \omega_k) \rho (\omega_{k-1}), \qquad \text{for each} \quad (\omega_k, \omega_{k-1})\in \{ 0,1, \ldots, L\}^2,
    \end{equation}
    where $\rho(x)$ is given by \eqref{stationary rho explicit}, and $P(x,y)$ is given by the element in the position $(x,y)$ in the matrix of probability transitions $P$ defined in \eqref{matrix of prob transitions}.

    Using the properties of Markov chains, for any $m>k$, we obtain
    \begin{eqnarray*}
        \text{Pr}\left[(\omega_m, \omega_{m-1}) | (\omega_k, \omega_{k-1}) \right] &=& \text{Pr}\left[ \omega_m | \omega_{m-1}, (\omega_k, \omega_{k-1}) \right] \text{Pr}\left[\omega_{m-1} | (\omega_k, \omega_{k-1})  \right] \\
        &=& \text{Pr}\left[ \omega_m | \omega_{m-1} \right] \text{Pr}\left[\omega_{m-1} | \omega_k  \right], \\
        &=&
        P(\omega_{m-1},\omega_{m}) \mu_{m-1-k} (\omega_{m-1} | \omega_k),
    \end{eqnarray*}
     for any $(\omega_m, \omega_{m-1}), (\omega_k,\omega_{k-1})\in \{ 0,1, \ldots L\}^2$.
     Here, we used the notation 
     $$
     \mu_{m-1-k} (\omega_{m-1} | \omega_k) = \text{Pr}\left[\omega_{m-1} | \omega_k  \right]
     $$
     introduced in Step 3 of the proof of Theorem \ref{thm:MTD-1D}. Hence, combining this with the definition of $\tilde{\rho}$ in \eqref{stationary rho tilde explicit}, and the convergence estimate in \eqref{convergence MC proof}, we obtain the following convergence estimate for the Markov chain $(\omega_k, \omega_{k-1})$:
     \begin{eqnarray*}
         \dfrac{1}{2} \sum_{\omega_m, \omega_{m-1}} \left| \text{Pr} [(\omega_m, \omega_{m-1}) | (\omega_k, \omega_{k-1})]  - \tilde{\rho} (\omega_m, \omega_{m-1}) \right| \\
         = \dfrac{1}{2} \sum_{\omega_m, \omega_{m-1}} P(\omega_{m-1}, \omega_{m}) \left| \mu_{m-1-k} (\omega_{m-1}| \omega_k) - \rho(\omega_{m-1}) \right| \\
         \leq  \dfrac{L+1}{2} (1-\beta)^{m-1-k}.
     \end{eqnarray*}
    
    Finally, using the construction of $g_k = (g_k^{(1)}, g_k^{(2)})$ in terms of $(\omega_k, \omega_{k-1})$ made in \eqref{g1 def} and \eqref{gk def}, we deduce the convergence result for the Markov chain $\{ g_k\}_{k\geq 1}$ with $\pi$ given by
    $$
    \pi (g^{(1)}, g^{(2)}) = \begin{cases}
        \tilde{\rho} (g^{(2)}, g^{(1)}) & \text{if} \ (g^{(1)}, g^{(2)}) \in \{ 1, \ldots , L-1\} \times \{ 0, \ldots, L \} \\
        \tilde{\rho} (g^{(2)}, 0) + \tilde{\rho} (g^{(2)},L)  & \text{if} \ g^{(1)} = 0 \ \text{and} \ g^{(2)}\in \{ 0, \ldots, L\} \\
        0 & \text{else.}
    \end{cases}
    $$
    The final expression of $\pi$, given in the statement of the Lemma is obtained by replacing the values of $\rho$ in \eqref{stationary rho explicit} and the terms in the matrix $P$ in \eqref{matrix of prob transitions} in the expression of $\tilde{\rho}$ in \eqref{stationary rho tilde explicit}.
\end{proof}

\subsection{Proof of Theorem \ref{thm:2d-main-result}}\label{sec:proof-of-2d-converence}

In this proof, we use a similar structure to the one used in the proof of Theorem \ref{thm:MTD-1D}:
\begin{enumerate}
    \item First we prove that for any $k\in \mathcal{I}_M:= \{1,\ldots , M\}^2$, the patch $Z_{k}$ extracted from $Z$ as in \eqref{eq:2d-patches} can be written as
    $$
    Z_{k} = P (g_{k} \cdot \overline{X} ) + \varepsilon_{k},
    $$
    for some $g_k = \left( g_{k}^{(1)}, g_{k}^{(2)} , g_{k}^{(3)} , g_{k}^{(4)}  \right)\in G = \left( \Z_{2L} \times \Z_{2L} \right)^4$ and $\varepsilon_{k}\in \R^{L\times L}$. Recall the definition of $\overline{X}\in (\R^{2L\times 2L})^4$ in \eqref{eqn:barXdef 2D}.
    \item Then, for each patch $Z_k$, we prove that the group element $g_{k}\in G$ associated to the $k$-th patch can be written in terms of the hard-core model from Assumption \ref{assumption hard-core}, restricted to a rectangular neighborhood of the patch, that we denote by $V_k\subset V_M$, where $V_M$ are the vertices of the graph $\mathcal{G}_M$ in Definition \ref{def:hardCoreModelGraph}. This construction defines $\{g_k\}_{k\in \mathcal{I}_M}$ as a random field over $G$, indexed by $\mathcal{I}_M$.
    \item Using known results for the hard-core model we have that, for $\lambda$ sufficiently small, the hard-core model from Assumption \ref{assumption hard-core} satisfies a property known as strong spatial mixing. This property implies that, as the number of patches $M^2$ tends to infinity, the hard-core model converges to a unique Gibbs measure on the two-dimensional integer lattice $\Z^2$.
    The strong spatial mixing also implies that, for patches $Z_k$ far from the measurement boundary, the hard-core model restricted to $V_k$ and the Gibbs measure restricted to $V_k$ are exponentially close as probability measures.
    \item Since the Gibbs measure in $\mathbb{Z}^2$ is shift invariant, the restriction to any rectangle $V_k \subset V_M$ is the same for every $k\in \mathcal{I}_M$. Then, we use this restriction of the Gibbs measure to define, as in Step (ii), a probability distribution $\pi$ over the group $G$.
    Using the strong spatial mixing from the previous step, we prove that the random field $\left\{ g_{k}  \right\}_{k\in \mathcal{I}_M}$ and the probability distribution $\pi$ satisfy the exponential mixing property in Definition \ref{def:exponential-mixing}.
\end{enumerate}

\begin{remark}
    Differently to the Markov chain $\{g_k\}_{k\geq 1}$ introduced in the proof of Theorem \ref{thm:MTD-1D} for the 1D case, the two-dimensional random field $\{g_k\}_{k\in \mathcal{I}_M}$ constructed in this proof does not have a stationary distribution per se. However, for our convergence results in Theorems~\ref{thm:sampleComplexity1DMTD} and \ref{thm:2d-general-result}, we only need that
    \begin{enumerate}[(a)]
        \item the probability distribution of $g_k$'s in patches $Z_k$ which are far from the measurement boundary is exponentially close to a certain probability distribution $\pi$ over $G$;
        \item the correlation between $g_k$'s associated to patches which are far apart from each other decays exponentially.
    \end{enumerate}
    These two properties, which we together called exponential mixing in Definition \ref{def:exponential-mixing}, are satisfied by Markov chains with positive absolute spectral gap, taking $\pi$ as the stationary distribution (note that the boundary of the index set $\{1, \ldots, M\}$ in the 1D case is $0$ and $M+1$). In the 2D case, the random field $\{ g_k\}_{k\in \mathcal{I}_M}$ and the probability distribution $\pi$ constructed in this proof also satisfy these two properties, with the pushforward of the Gibbs measure on a given patch (which \emph{is} stationary due to shift invariance) taking on the role of $\pi$.
\end{remark}

Let us now proceed with the proof of Theorem \ref{thm:2d-main-result}.

\textbf{Step 1:} Let us recall the definition of the linear operator $P: (\R^{2L\times 2L})^4 \to \R^{L\times L}$ in \eqref{P def 2D}, which is of the form
\begin{equation}
\label{projectoin operator 2D}
P(\tilde{X}_0, \tilde{X}_1, \tilde{X}_2, \tilde{X}_3) = 
P_{00} \tilde{X}_0 + P_{01} \tilde{X}_1 
+ P_{10} \tilde{X}_2 + P_{11} \tilde{X}_3,
\end{equation}
for any $\left(\tilde{X}_0, \tilde{X}_1, \tilde{X}_2, \tilde{X}_3\right)\in (\R^{2L\times 2L})^4$.

Since the patches $Z_k$ have size $L\times L$, in each of them there can be, at most, one pixel containing the pixel $X[0,0]$ of the signal $X\in \R^{L\times L}$. In other words, for any $k = (k_1, k_2)\in \mathcal{I}_M$, there is at most one element in $\{ t_i \}_{i=1}^q$ such that $t_i - kL\in \{ 0, \ldots , L-1  \}^2$.
Moreover, in view of the definition of $Z_k$ in \eqref{eq:2d-patches}, the signal in $Z_k$ might also contain a piece of $X$ for which the pixel $X[0,0]$ is contained in the patches $Y_{k_1, k_2-1}, Y_{k_1-1, k_2}$ or $Y_{k_1-1, k_2-1}$.

Recall the group action of $\Z_{2L}\times \Z_{2L}$ on $\R^{2L \times 2L}$ defined in \eqref{group action 2d simple}, which cyclically shifts the coordinates of $2L\times 2L$ matrices coordinate wise, and also the group action of $G = (\Z_{2L}\times \Z_{2L})^4$ on $(\R^{2L \times 2L})^4$, which applies the previous action component wise, as defined in \eqref{eq:group-action-2d-double}.
Next we choose, for each $k=(k_1, k_2)\in \mathcal{I}_M$, a group element $g_k = (g_k^{(1)}, g_k^{(2)}, g_k^{(3)}, g_k^{(4)})\in G$ such that
\begin{enumerate}[i)]
    \item $P_{00} \left[g_k^{(1)}\cdot \tilde{X}\right]$ contains the part of $X$ in $Z_k$ for which the pixel $X[0,0]$ is in the patch $Z_k$.
    \item $P_{01} \left[g_k^{(2)}\cdot \tilde{X}\right]$ contains the part of $X$ in $Z_k$ for which the pixel $X[0,0]$ is in the patch $Y_{k_1, k_2-1}$.
    \item $P_{10} \left[g_k^{(3)}\cdot \tilde{X}\right]$ contains the part of $X$ in $Z_k$ for which the pixel $X[0,0]$ is in the patch $Y_{k_1-1, k_2}$.
    \item $P_{11} \left[g_k^{(4)}\cdot \tilde{X}\right]$ contains the part of $X$ in $Z_k$ for which the pixel $X[0,0]$ is in the patch $Y_{k_1-1, k_2-1}$.
\end{enumerate}

For each $k\in \mathcal{I}_M$, we define the parameter $\omega_k\in \{ 0, 1, \ldots , L-1\}^2 \cup \{ (L,L) \}$, which takes a value in $\{ 0, 1, \ldots , L-1\}^2$ if the patch $Z_k$ contains the pixel $X[0,0]$ of some signal apperance $X$ in $Z$, and $\omega_k = (L,L)$ if $Z_k$ does not contain the pixel $X[0,0]$.
More precisely, we define $\omega_k$ as
\begin{equation}
\label{omega_k def}
\omega_k := 
\begin{cases}
    t & \text{if }
    \{ t_i - kL\}_{i=1}^q \cap \{0, \ldots, L-1  \}^2 = \{ t\}, \\
    (L,L) & \text{if }
     \{ t_i - kL\}_{i=1}^q \cap \{0, \ldots, L-1  \}^2 = \emptyset.
\end{cases}
\end{equation}
Note that, due to the hard-core model assumption on $\{t_i\}_{i=1}^q$, the intersection $ \{ t_i - kL\}_{i=1}^q \cap \{0, \ldots, L-1  \}^2$ can contain, at most, one element.

Next we prove that, for any $k\in \mathcal{I}_M$, the group elements $g_k = (g_k^{(1)}, g_k^{(2)}, g_k^{(3)}, g_k^{(4)})\in (\Z_{2L\times 2L})^4$ can be chosen, respectively, in terms of the parameters 
$$\omega_k, \quad \omega_{k_1, k_2-1}, \quad \omega_{k_1-1, k_2} \quad \text{and} \quad \omega_{k_1-1, k_2-1}, \quad \text{defined in \eqref{omega_k def}.}
$$
Of course, if $k=(k_1, k_2)$ is such that $k_1=0$ or $k_2=0$, the sub-indices in the above parameters may take negative values. In this case, the associated parameter $\omega$ is taken as $\omega = (L,L)$, or equivalently, we consider that patches $Y_{k_1, k_2}$ with a negative sub-index are empty. 

Let us recall the definition of $\tilde{X}$ in \eqref{X tilde 2D}, which is given by
$$
\tilde{X} := \begin{bmatrix}
    X & 0_{L\times L} \\
    0_{L\times L} & 0_{L\times L}
\end{bmatrix}.
$$
If the parameter $\omega_k\in  \{ 0, 1, \ldots , L-1\}^2$, then the patch $Z_k$ contains $X[0,0]$ in its $\omega_k$ pixel. Therefore, the part of the signal $X$ that appears in the patch $Z_k$ is precisely $P_{00} \left[ \omega_k\cdot \tilde{X} \right]$. If the patch $Z_k$ does not contain the pixel $X[0,0]$, then we can simply take $g_k^{(1)} = \omega_k = (L,L)$, so that $P_{00} \left[ g_k^{(1)}\cdot \tilde{X} \right] = 0_{L\times L}$.
In conclusion, for any $k\in \mathcal{I}_M$, we can take $g_k^{(1)} = \omega_k$.

If the parameter $\omega_{k_1, k_2-1} \in \{ 0,\ldots, L-1\}^2$, then the patch $Y_{k_1, k_2-1}$ contains $X[0,0]$ in its $\omega_{k_1, k_2-1}$ pixel. Therefore, the remainder of this signal appearance in the patch $Y_{k_1, k_2}$ is given by $P_{01} \left[\omega_{k_1, k_2-1} \cdot \tilde{X} \right]$. If the patch $Y_{k_1, k_2-1}$ does not contain the pixel $X[0,0]$, then we can again take $g_k^{(2)} = \omega_{k_1, k_2-1} = (L,L)$, so that $P_{01} \left[ g_k^{(2)}\cdot \tilde{X} \right] = 0_{L\times L}$.
Hence, we can take $g_k^{(2)} = \omega_{k_1, k_2-1}$.

By a similar reasoning, we can take $g_k^{(3)} = \omega_{k_1-1, k_2}$, which implies that $P_{10} \left[ g_k^{(3)}\cdot \tilde{X}\right]$ contains the part of the signal $X$, appearing in $Z_k$, for which the pixel $X[0,0]$ is contained in the patch $Y_{k_1-1, k_2}$; and $P_{10} \left[ g_k^{(3)}\cdot \tilde{X}\right] = 0_{L\times L}$ if the patch $Y_{k_1-1, k_2}$ does not contain the pixel $X[0,0]$, i.e.\ if $\omega_{k_1-1, k_2} = (L,L)$.

If the patch $Y_{k_1-1, k_2-1}$ contains the pixel $X[0,0]$, then $\omega_{k_1-1, k_2-1}\in \{ 0, 1, \ldots, L-1\}^2$, and the remainder in $Z_k$ of this signal appearance is given by $P_{11}\left[ \omega_{k_1-1, k_2-1} \cdot \tilde{X} \right]$. If $Y_{k_1-1, k_2-1}$ does not contain the pixel $X[0,0]$, i.e., if $\omega_{k_1-1, k_2-1} = (L,L)$, then we can take $g_k^{(4)} = (0,0)$. 

In conclusion, for each $k\in \mathcal{I}_M$, the group element $g_k = \left( g_k^{(1)}, g_k^{(2)}, g_k^{(3)}, g_k^{(4)}\right)\in (\Z_{2L}\times \Z_{2L})^4$ can be taken as
\begin{equation}
\label{g_k in terms of omega}
g_k^{(1)} = \omega_k,
\quad
g_k^{(2)} = \omega_{k_1, k_2-1},
\quad
g_k^{(3)} = \omega_{k_1-1, k_2},
\quad
g_k^{(4)} =
\begin{cases}
    \omega_{k-1} & \text{if} \ \omega_{k-1}\in \{ 0, \ldots, L-1\}^2 \\
    (0,0) & \text{if} \ \omega_{k-1} = (L,L),
\end{cases}
\end{equation}
where the parameters $\omega_k$ are defined, in terms of $\{ t_i\}_{i=1}^q$ by \eqref{omega_k def}.
In the definition of $g_k^{(4)}$ above, $k-1$ denotes $(k_1- 1, k_2 -1)$.

\textbf{Step 2:}
Let us recall the graph $\mathcal{G}_M = (V_M, E_M)$, with vertices $V_M\subset \Z^2$ given by \eqref{graph vertex} and edges $E_M\subset V_M \times V_M$ defined in \eqref{graph edges}.
Let us denote by $\Theta$ the random variable associated to the hard-core model on the graph $\mathcal{G}_M$ with activity parameter $\lambda\in (0,1)$; see Definition \ref{def:hardCoreModelGraph}. 
The hard-core model $\Theta$ is a random variable that takes values in the admissible subgraphs of the graph $\mathcal{G}_M$, i.e.\ subsets of $V_M$ which contain no adjacent vertices.
From now on, we represent subsets of $V_M$ by means of a boolean matrix of the form $\eta = \{\eta_v\ : \ v\in V_M\} \in\{ 0,1\}^{V_M}$.
We recall that, for any $\eta\in \{0,1\}^{V_M}$, we say that $\eta$ is admissible if and only if
$$
\eta_v \eta_{v'} = 0 \qquad \forall (v,v')\in E_M,
$$
i.e.\ any two vertices $v$ and $v'$ which are connected by an edge, do not take the value 1 simultaneously.
The probability function for the random variable $\Theta$ can therefore be written as
$$
\mathbb{P} (\Theta = \eta) =
\begin{cases}
C^{-1} \lambda^{|\eta|} & \text{if $\eta$ is admissible}, \\
0 & \text{else,}
\end{cases}
$$
where, for each $\eta\in \{ 0,1\}^{V_M}$, we write 
$|\eta| = \sum_{v\in V_M} \eta_v$ and $C$ is the normalizing constant
$$
C := \sum_{\substack{\eta \in \{0,1\}^{V_M} \\ \text{indep.}}} \lambda^{|\eta|}. 
$$

In view of Assumption \ref{assumption hard-core}, the parameters $\{ t_i\}_{i=1}^q$ in \eqref{eq:2d-mtd} can be written, in terms of the random variable, $\Theta$ as
$$
\{ t_i\}_{i=1}^q = \{ v\in V_M \ : \quad \Theta_v = 1 \}.
$$
By the definition of the parameter $\omega_k$ in \eqref{omega_k def} for each $k = (k_1, k_2)\in \mathcal{I}_M$, we can write $\omega_k$ in terms of the restriction of the hard-core model $\Theta$ to the rectangle
$$
kL + \{0, \ldots, L-1\}^2 = \left\{ (Lk_1 + \xi_1, \, Lk_2 + \xi_2) \ : \quad (\xi_1,\xi_2)\in \{ 0, \ldots, L-1 \}^2 \right\}.
$$
More precisely, we can write $\omega_k$ as
\begin{equation}
\label{omega k in terms of hard core}
\omega_k :=  \{ v\in kL + \{ 0, \ldots, L-1 \}^2 \ : \quad \Theta_v = 1 \},
\end{equation}
with $\omega_k := (L,L)$ if $\Theta_v = 0$ for all $v\in kL + \{ 0, \ldots, L-1 \}$.

Also, in view of \eqref{g_k in terms of omega},  for each $k = (k_1, k_2)\in \mathcal{I}_M$, the group element $g_k = \left( 
g_k^{(1)}, g_k^{(2)}, g_k^{(3)}, g_k^{(4)} \right)$ can be written as a deterministic function of the parameters $\omega_k, \quad \omega_{k_1, k_2-1}, \quad \omega_{k_1-1, k_2}$ and $\omega_{k_1-1, k_2-1}$.
Hence, for any $k\in \mathcal{I}_M$, we can choose $g_k\in G$ uniquely in terms of the restriction of the hard-core model $\Theta$ to the rectangle $V_k\in \Z^2$ given by
\begin{equation}
\label{V_k def}
V_k := kL + \{-L, \ldots, L-1\}^2 = \left\{ (Lk_1 + \xi_1, \, Lk_2 + \xi_2) \ : \quad (\xi_1,\xi_2)\in \{ -L, \ldots, L-1 \}^2 \right\}.
\end{equation}
For each $k\in \mathcal{I}_M$, let us denote by $\Theta_k$ the restriction of the random variable $\Theta$ to the rectangle $V_k$.
Since the choice of $g_k\in G$ only depends on the random variable $\Theta_k$, we can view $\{ g_k \}_{k\in \mathcal{I}_M}$ as a random field over $G$, indexed by $\mathcal{I}_M$.

\textbf{Step 3:} 
Next, we claim that the hard-core model $\Theta$ on the graph $\mathcal{G}_M = (V_M,E_M)$ has a unique Gibbs measure on the infinite graph $\Z^2$, with the same criterion for edge connectivity as in \eqref{graph edges}. This is a consequence of the following property, known as strong spatial mixing. 
\begin{lemma}[Corollary 2.6 of \cite{weitz2006}]
\label{lem: strong mixing hard-core}
In a general hard-core model $\Theta$, defined on a graph $(V,E)$, there exists a critical value $\lambda_c>0$, depending on the maximal degree of the graph, such that for any $0<\lambda<\lambda_c$ there exist constants $c,\gamma$ depending on $\lambda$ and the maximal degree such that for any vertex $v\in V$, any set of vertices $S_1\subset V$, any $\tau\in \braces{0,1}$ and any configurations $\eta,\eta'\in \braces{0,1}^{S_1}$,
\begin{equation} 
 \label{strong mixing property}
     \abs{\mathbb{P}(\Theta_{v}=\tau \mid \Theta_{S_1} = \eta)-\mathbb {P}(\Theta_{v}=\tau \mid \Theta_{S_1}=\eta')} \leq c \exp (-\gamma \tilde{d}(v,\eta,\eta')),
\end{equation}
where $\tilde{d}(v,\eta,\eta')=\min \{ d(v,w) : \eta_w \neq \eta'_w \}$ is the smallest graph distance from $v$ to a vertex (in $S_1$) on which $\eta$ and $\eta'$ differ.
\end{lemma}

Note that the sufficient condition on $\lambda$ for strong spatial mixing to hold depends on the maximal degree of the graph, which in our case is $(2L+1)^2$, due to the definition of the edges in \eqref{graph edges}, and since every vertex $v\in V_M$ is connected to all the vertices in the rectangle $v + \{ -L, \ldots, L \}^2$.

It is well known that strong spatial mixing implies the existence of a Gibbs measure $\mu$ defined on the infinite graph $\Z^2$ (see \cite{georgii2011} and \cite{01771213-b0f2-3d96-9da6-34ca75b35228} for general theory on Gibbs measures, and \cite{weitz2006} for the special case of hard-core models).
The Gibbs measure $\mu$ is a probability distribution on the infinite-dimensional state space $\{ 0,1\}^{\Z^2}$ and can be used to write the hard-core model on the finite graph $\mathcal{G}_M = (V_M,E_M)$, as the restriction of the Gibbs measure to $V_M$, conditionally on all the elements on the complement $V_M^c := \Z^2\setminus V_M$ being $0$.
If we denote by $\overline{\Theta}$ the random variable on $\{0,1\}^{\Z^2}$ associated to the Gibbs measure $\mu$, then for any $\eta\in \{0,1\}^{V_M}$, we can write the probability $\mathbb{P} [\Theta = \eta]$ as
$$
\mathbb{P} [\Theta = \eta] = \mathbb{P} \left[ \overline{\Theta} = \overline{\eta} \mid \overline{\Theta}_{V^c_M} = 0\right] = \dfrac{\mu (\overline{\eta})}{\mathbb{P} \left[ \overline{\Theta}_{V^c_M} = 0 \right]}, \qquad \text{where} \quad
\mathbb{P} \left[ \overline{\Theta}_{V^c_M} = 0 \right] = \sum_{\substack{\tau\in \{0,1\}^{\Z^2} \\ \tau_{V^c_M} = 0}} \mu (\tau).
$$
Here, $\overline{\eta}$ is the extension by zeros of $\eta$ from $V_M$ to $\Z^2$, i.e., $\overline{\eta}_v = \eta_v$ if $v\in V_M$ and $\overline{\eta}_v = 0$ if $v\in V^c_M$,
and for any $\tau\in \{0, 1\}^{\Z^2}$, $\tau_{V^c_M}$ denotes the restriction of $\tau$ to the set $V^c_M = \Z^2\setminus V_M$.

\textbf{Step 4:} In Step 2, we proved that the group element $g_k\in G$, for each $k\in \mathcal{I}_M$, is determined by restricting the hard-core model $\Theta$ on the graph $\mathcal{G}_M$, to the rectangle $V_k$ defined in \eqref{V_k def}.
The probability distribution $\pi$ on $G$ satisfying \eqref{eq:mixing-property} will be constructed by restricting the Gibbs measure $\mu$ on $\Z^2$, introduced in Step 3, to the rectangle $V_k$. Recall that the existence of this Gibbs measure is guaranteed by Lemma \ref{lem: strong mixing hard-core}, provided $\lambda\in (0,\lambda_0)$, where $\lambda_0$ is a constant depending only on $L$.

For each $k\in \mathcal{I}_M$, let us denote by $\mu_k$ the marginal of the Gibbs measure $\mu$ on the rectangle $V_k$ defined in \eqref{V_k def}. Note that the support of $\mu_k$ is $\{ 0,1\}^{V_k}$.
Based on the construction of $g_k$ in \eqref{g_k in terms of omega} and \eqref{omega k in terms of hard core}, which is in terms of the restriction of the hard-core model to the rectangle $V_k$, there is a probability distribution $\pi_k$ on $G$ associated to the measure $\mu_k$. More precisely, $\pi_k$ is the pushforward measure of $\mu_k$ under the encoding map from $\{0,1\}^{V_k}$ to $G$ constructed in \eqref{g_k in terms of omega} and \eqref{omega k in terms of hard core}. Due to the shift invariance of the Gibbs measure $\mu$ over $\Z^2$, one can easily verify that all the marginals $\mu_k$ are equal, and hence, there exists a probability distribution $\pi$ on $G$  such that $\pi_k = \pi$ for all $k\in \mathcal{I}_M$. 

Let us now prove that the random field $\{g_k\}_{k\in \mathcal{I}_M}$ and the probability distribution $\pi$ satisfy \eqref{eq:mixing-property}.
Let $k\in \mathcal{I}_M$, and let $S\subset \mathcal{I}_M$ be a subset of indices. 
Denoting $V_k$ the rectangle defined in \eqref{V_k def}, let us define $V_S = \bigcup_{k'\in S} V_{k'}$.
By the choice of $g$ in \eqref{g_k in terms of omega}, for any $\varphi\in G$, there is a unique associated configuration $\eta\in \{0,1\}^{V_k}$ of the hard-core model $\Theta$ restricted to $V_k$, and for each admissible configuration $\Psi \in G^{|S|}$, there is a unique admissible configuration $\tau \in \{0,1\}^{V_S}$ of the hard-core model $\Theta$ restricted to $V_S$.
Then, we can write the probability of $g_k = \varphi$ and $g_S = \Psi$ in terms of the hard-core model $\Theta$ on $\mathcal{G}_M = (V_M,E_M)$ as
$$
\mathbb{P} \left[ g_k = \varphi \right] = \mathbb{P} \left[ \Theta_{V_k} = \eta \right]
\qquad \text{and} \qquad
\mathbb{P} \left[ g_S = \Psi \right] = \mathbb{P} \left[ \Theta_{V_S} = \tau \right] .
$$
Using the relation between the hard-core model $\Theta$ and the random variable $\overline{\Theta}$ associated to the Gibbs measure $\mu$ on $\{ 0,1 \}^{\Z^2}$, we can write
\begin{equation}
\label{probability of A and B}
\mathbb{P} \left[ g_k = \varphi \right] = \mathbb{P} \left[ \overline{\Theta}_{V_k} = \eta \mid \overline{\Theta}_{V^c_M} = 0 \right]
\qquad \text{and} \qquad
\mathbb{P} \left[ g_S = \Psi \right] = \mathbb{P} \left[ \overline{\Theta}_{V_S} = \tau \mid \overline{\Theta}_{V^c_M} = 0 \right].
\end{equation}
Let us now prove that
\begin{equation}
\label{cond prob 1}
\mathbb{P} \left[ g_k = \varphi \mid g_S = \Psi \right] =
\mathbb{P} \left[ \overline{\Theta}_{V_k} = \eta \mid (\overline{\Theta}_{V_S}, \overline{\Theta}_{V^c_M}) =(\tau, 0) \right].
\end{equation}
We define the events $A := [g_k = \varphi]$, $B:= [g_S = \Psi]$. The intersection of both events is given by $A\cap B = \left[ (g_k, g_S ) = (\varphi, \Psi) \right]$.
Using the second identity in \eqref{probability of A and B}
and the  formula of conditional probabilities, we can write $\mathbb{P}[B]$ as
$$
\mathbb{P} [B] = 
\mathbb{P} \left[ \overline{\Theta}_{V_S} = \tau \mid \overline{\Theta}_{V^c_M} = 0 \right] =
\dfrac{\mathbb{P} \left[ (\overline{\Theta}_{V_S}, \overline{\Theta}_{V_M^c}) = (\tau, 0) \right]}{\mathbb{P} \left[ \overline{\Theta}_{V_M^c} = 0  \right]}.
$$
Similarly, the probability of the intersection $A\cap B$ is given by
\begin{eqnarray*}
\mathbb{P}[A\cap B] &=& \mathbb{P}[(g_k ,g_S) = (\varphi, \Psi) ] = \mathbb{P} \left[ (\overline{\Theta}_{V_k}, \overline{\Theta}_{V_S} ) = (\eta, \tau) \mid \overline{\Theta}_{V_M^c} = 0 \right] \\
&=& \dfrac{\mathbb{P} \left[ (\overline{\Theta}_{V_k} ,\overline{\Theta}_{V_S}, \overline{\Theta}_{V_M^c} ) = (\eta, \tau, 0) \right]}{ \mathbb{P} \left[ \overline{\Theta}_{V_M^c} = 0 \right]}.
\end{eqnarray*}
Using the formula for the conditional probability, we have
\begin{eqnarray*}
\mathbb{P} \left[ A \mid B \right] = \dfrac{\mathbb{P} [A\cap B]}{\mathbb{P} [B]} &=& 
\dfrac{\mathbb{P} \left[ (\overline{\Theta}_{V_k} ,\overline{\Theta}_{V_S}, \overline{\Theta}_{V_M^c} ) = (\eta, \tau, 0) \right]}{\mathbb{P} \left[ (\overline{\Theta}_{V_S}, \overline{\Theta}_{V_M^c}) = (\tau, 0) \right]} \\
&=&
\mathbb{P}\left[ \overline{\Theta}_{V_k} = \eta \mid (\overline{\Theta}_{V_S}, \overline{\Theta}_{V_M^c}) = (\tau, 0) \right],
\end{eqnarray*}
and \eqref{cond prob 1} is proven.

On the other hand, using the law of total probability, we have
\begin{equation}
\label{cond prob 2}
\pi (\varphi) = 
\mathbb{P} \left[ \overline{\Theta}_{V_k} = \eta \right] = 
\sum_{\substack{\tau_1 \in \{0,1\}^{V_S} \\
\tau_2\in \{0,1\}^{V^c_M}}} \mathbb{P} \left[ \overline{\Theta}_{V_k} = \eta \mid (\overline{\Theta}_{V_S} , \overline{\Theta}_{V^c_M}) = (\tau_1, \tau_2) \right] \mathbb{P} \left[ (\overline{\Theta}_{V_S} , \overline{\Theta}_{V^c_M} ) = (\tau_1,\tau_2) \right].
\end{equation}
Combining \eqref{cond prob 1} and \eqref{cond prob 2}, with the strong spatial mixing property from Lemma \ref{lem: strong mixing hard-core}, we obtain
\begin{align*}
   & \big| \mathbb{P} \left[ g_k = \varphi \mid g_S = \Psi  \right] - \pi (\varphi) \big|
    \leq  \\ 
   & \sum_{\substack{\tau_1 \in \{0,1\}^{V_S} \\
\tau_2\in \{0,1\}^{V^c_M}}} 
\Big|
\mathbb{P} \left[ \overline{\Theta}_{V_k} = \eta \mid \overline{\Theta}_{V_S, V^c_M} = (\tau,0) \right] 
- 
\mathbb{P} \left[ \overline{\Theta}_{V_k} = \eta \mid \overline{\Theta}_{V_S, V^c_M} = (\tau_1, \tau_2) \right] 
\Big| \, 
\mathbb{P} \left[ \overline{\Theta}_{V_S, V^c_M} = (\tau_1, \tau_2) \right] \\
& \leq c_0 \exp \left( - \gamma d( V_k, V_S\cup V^c_M ) \right),
\end{align*}
where $d(V_k, V_S\cup V^c_M)$ denotes the distance between the sets $V_k$ and $V_S\cup V^c_M$ in the 2-dimensional lattice $\Z^2$.
In view of the construction of the rectangles $V_M$ in \eqref{graph vertex} and $V_k$ in \eqref{V_k def}, we see that $d(k, S\cup \mathcal{I}^c_M) \leq L d(V_k, V_S\cup V^c_M)$, and property \eqref{eq:mixing-property} then follows with exponent $\gamma/L$.
\hfill $\square$

\subsection{Proof of Theorems~\ref{thm:sampleComplexity1DMTD} and \ref{thm:2d-general-result}}\label{sec:proof-of-sample-complexity}

The proof is almost identical for both theorems, and relies on the mixing properties of the group elements $g_k$ which appear in the induced MRA model for the patches $Z_k$. The main difference is the index set, which is $k\in \{ 1, \ldots , M\}$ for the one-dimensional case (Theorem \ref{thm:sampleComplexity1DMTD}) and $k\in \{ 1, \ldots, M \}^2$ in the two-dimensional case (Theorem \ref{thm:2d-general-result}). In both cases, we denote the index set by $\mathcal{I}_M$, and $|\mathcal{I}_M|$ denotes the cardinality of $\mathcal{I}_M$ (i.e., $|\mathcal{I}_M| = M$ for the 1D case, and $|\mathcal{I}_M| = M^2$ for the 2D case).

Note that, in the 1D case, the Markov chain $\{ g_k\}_{k=1}^{M}$ from Theorem \ref{thm:MTD-1D} can be seen as a random field  over $G$ indexed by $\mathcal{I}_M = \{1,2,\ldots M \} \subset \Z$; see Definition \ref{def:random-field}. Moreover, the Markovian property, together with property \eqref{exponential mixing MC thm} proved in Theorem \ref{thm:MTD-1D}, imply the following mixing property: 
for any $k\in \mathcal{I}_M = \{1, 2, \ldots ,M\}$, any $S \subset \mathcal{I}_M$ with $s < k$ for all $s\in S$, and any admissible $\Psi\in G^S$, 
\begin{equation}
    \label{mixing property 1D}
    | \mathbb{P} (g_k = \varphi \mid  g_{|S}=\Psi) - \pi (\varphi)| \leq C \exp \left(-\gamma (k - s_0)\right),
\end{equation}
where  $s_0 = \max (S) $ or $s_0 = 0$ if $S= \emptyset$. Here, we take $\gamma = - \log (1-(1-\lambda)^{L-1}) >0$ and $C>0$ the same constant as in Theorem \ref{thm:MTD-1D}.
The mixing property \eqref{mixing property 1D} can be seen as a one-dimensional version of the mixing property in Definition \ref{def:exponential-mixing} noting that, in this case, the states $\{g_k\}_{k=1}^M$ have a natural ordering.

In view of Theorems \ref{thm:MTD-1D} and \ref{thm:2d-main-result}, in both cases, the patches $\{ Z_k\}_{k\in \mathcal{I}_M}$ extracted from the MTD measurement $Z$, can be written as
\begin{equation}
    \label{MRA def proof}
    Z_k = P (g_k\cdot \overline{X}) + \varepsilon_k, \qquad \text{for}\ k\in \mathcal{I}_M,
\end{equation}
for some projection operator $P$ and some group elements $\{ g_k \}_{k\in \mathcal{I}_M}$ in $G$, acting on the target signal $\overline{X}$, and $\varepsilon_k$ being i.i.d.\ Gaussian noise with variance $\sigma^2>0$.
We also recall, from the statement of the theorem, the definition of the i.i.d.\ measurements $\{Y_k\}_{k\in \mathcal{I}_M}$, which have the same form as \eqref{MRA def proof}, but with group elements $\{ \tilde{g}_k \}_{k\in \mathcal{I}_M}$ being i.i.d.\ from a suitable probability distribution $\pi$ over $G$.

The proof of the Theorems is naturally divided in two parts: part (i), which addresses the case of general estimators, considering an $m$-separated subsampling of the patches; and part (ii), which deals with estimators based on an empirical average, and do not require subsampling.

\textbf{Part (i):} We carry the proof in three steps. First, we unify the notation concerning the patch subsampling in the 1D and the 2D settings. Then, we prove an $\ell^1$-estimate for the probability distribution of the group elements associated the subsampled patches, which is a direct consequence of the mixing properties proved in Theorems~\ref{thm:MTD-1D} and \ref{thm:2d-main-result}. Finally, we combine the $\ell^1$-estimate with the law of total expectation to conclude the proof.

\textit{Step 1: Preparations.} Let $c>0$ be a constant to be chosen later. For any $M\in \mathbb{N}$, choose $m = \lfloor c \log(M) \rfloor$ and $M' = \lfloor M/m \rfloor$.
We define the index set
$$
\mathcal{I}_M' := \{ km\ , : \quad k\in \mathcal{I}_{M'} \}.
$$
The cardinality of the index set is $|\mathcal{I}_{M}'| = M'$ in the 1D case, and $|\mathcal{I}_{M}'| = (M')^2$ in the 2D case.
Note that $\mathcal{I}_M' \subset \mathcal{I}_M$. 
Moreover, in the 1D case we have
\begin{equation}
\label{index separation proof 1D}
k - k' \geq m = \floor{c \log (M)},
\qquad \forall k, k'\in \mathcal{I}_M', \ \text{with $k>k',$}
\end{equation}
whereas in the 2D case we have
\begin{equation}
\label{index separation proof}
d(k,k') \geq m = \floor{c \log (M)},
\qquad \forall k, k'\in \mathcal{I}_M',
\end{equation}
$d(k,k') = \max_{j\in\{1,2\}} |k_j-k'_j|$.
Considering $\mathcal{I}_M'$ as a subset of $\mathcal{I}_M$, we also have that
$d(k, \mathcal{I}_{M}^c)>m$ for all $k\in \mathcal{I}_M'$, where $\mathcal{I}_{M}^c = \Z^j \setminus \mathcal{I}_M$, with $j=1$ in the 1D case, and $j=2$ in the 2D case.

We are therefore in position to apply the mixing property \eqref{mixing property 1D} in the 1D case, and \eqref{eq:mixing-property} in the 2D case for any $k\in \mathcal{I}_M'$ with $S = \mathcal{I}_M'\setminus \{k\}$.

Let us denote by $\pi_M$ the probability density function associated to the marginal of the random field $\{ g_k\}_{k\in \mathcal{I}_M}$ to the subset of indices $\mathcal{I}_M'$, and let $\tilde{\pi}_M$ be the probability density function of the marginal of the random field $\{ \tilde{g}_k\}_{k\in \mathcal{I}_M}$ to $\mathcal{I}_M'$.
Note that the latter is an i.i.d.\ sampling from $\pi$, whereas the former is extracted from a random field satisfying the mixing property \eqref{mixing property 1D} and \eqref{eq:mixing-property} respectively for the 1D and the 2D cases.

Now we choose, in the 2D case, an arbitrary ordering of the index set $\mathcal{I}_{M}'$, so that any $\Psi \in G^{\mathcal{I}_M'}$ can be written as an ordered list of the form $\Psi =(\psi_1, \psi_2, \ldots, \psi_{|\mathcal{I}_M'|})\in G^{|\mathcal{I}_M'|}$. In the 1D case, we choose the natural ordering of the index set $\mathcal{I}_M' := \{m, 2m, 3m, \ldots , mM' \}$.

\textit{Step 2: Mixing property and $\ell^1$-estimate.}
We claim that there exists $c_1>0$, independent of $M'$, such that the $\ell_1$ distance between $\tilde{\pi}_M$ and $\pi_M$ satisfies
\begin{equation}
\label{claim convergence}
\norm{{\pi}_M-\tilde{\pi}_M}_1\coloneqq \sum_{\Psi} \abs{\pi_M(\Psi)-\tilde{\pi}_M(\Psi)} \leq (M')^j |G| c_1 e^{-\gamma m},
\end{equation}
for any $\Psi\in G^{|\mathcal{I}_M'|}$, and where $j=1$ in the 1D case and $j=2$ in the 2D case.

Using the notation $\Psi_{1:j}$ to denote the vector $(\psi_1,\dots,\psi_j)\in G^j$, define
\[A_j = \sum_{\Psi_{1:j}} \abs{ \mathbb{P}({g}_{1:j}=\Psi_{1:j})-\mathbb{P}(\tilde{g}_{1:j}=\Psi_{1:j})},\] 
so that $\norm{\tilde{\pi}_M-\pi_M}_1=A_{|\mathcal{I}_M'|}$. Due to the mixing properties\footnote{\eqref{mixing property 1D} and \eqref{index separation proof 1D} in the 1D case; and \eqref{eq:mixing-property} and \eqref{index separation proof} in the 2D case.}, and the fact that $\mathbb{P}\left(\tilde{g}_1 = \Psi_1\right) = \pi (\Psi_1),$
we have that $A_1\leq |G| c_1 e^{-\gamma m}$, where $|G|$ is the cardinality of the finite group $G$.

Now, using that 
$$
\mathbb{P} \left( {g}_{1:j}=\Psi_{1:j} \right)
=
\mathbb{P} \left( {g}_{j}=\Psi_{j} \mid {g}_{1:j-1}=\Psi_{1:j-1}\right)  \mathbb{P} \left( {g}_{1:j-1}=\Psi_{1:j-1}\right),
$$ 
and that
$$
\mathbb{P} \left( \tilde{g}_{1:j}=\Psi_{1:j} \right)
=
\pi\left( \Psi_j \right)  \mathbb{P} \left( \tilde{g}_{1:j-1}=\Psi_{1:j-1}\right),
$$ 
for all $j$ and $\Psi_{1:j}$, we can decompose 
\begin{align*}
\abs{\mathbb{P}\left( g_{1:j}=\Psi_{1:j})-\mathbb{P}(\tilde{g}_{1:j}=\Psi_{1:j}\right) }  
\leq  & \mathbb{P}\left( g_{1:j-1}=\Psi_{1:j-1} \right) \abs{\mathbb{P} \left(g_{j}=\Psi_{j} \mid g_{1:j-1}=\Psi_{1:j-1}\right) - \pi(\Psi_j)} \\
& + \pi(\Psi_j) \abs{\mathbb{P}\left( g_{1:j-1}=\Psi_{1:j-1}\right)-\mathbb{P} \left( \tilde{g}_{1:j-1}=\Psi_{1:j-1}\right) } 
\end{align*}
Hence, we obtain
\begin{align*}
A_j  
\leq  & \sum_{\Psi_{1:j}} \mathbb{P}\left( g_{1:j-1}=\Psi_{1:j-1} \right) \abs{\mathbb{P} \left(g_{j}=\Psi_{j} \mid g_{1:j-1}=\Psi_{1:j-1}\right) - \pi(\Psi_j)} \\
& + \sum_{\Psi_{1:j}} \pi(\Psi_j) \abs{\mathbb{P}\left( g_{1:j-1}=\Psi_{1:j-1}\right)-\mathbb{P} \left( \tilde{g}_{1:j-1}=\Psi_{1:j-1}\right)} \\
= & \, I + II.
\end{align*}
Using again the mixing properties\footnote{\eqref{mixing property 1D} and \eqref{index separation proof 1D} in the 1D case; and \eqref{eq:mixing-property} and \eqref{index separation proof} in the 2D case.}, we can estimate $I$ as
\begin{align*}
    I \leq & c_1 e^{-\gamma m} \sum_{\Psi_{1:j}} \mathbb{P}\left( g_{1:j-1}=\Psi_{1:j-1} \right) = c_1 e^{-\gamma m} \sum_{\Psi_j} \underbrace{\sum_{\Psi_{1:j-1}} \mathbb{P}\left( g_{1:j-1}=\Psi_{1:j-1} \right)}_{=1} = |G| c_1 e^{-\gamma m}.
\end{align*}
For the term $II$, we can factor out $\pi (\Psi_j)$ to obtain
$$
II = \underbrace{\sum_{\Psi_j} \pi(\Psi_j)}_{=1} \sum_{\Psi_{1:j-1}} \abs{\mathbb{P}\left( g_{1:j-1}=\Psi_{1:j-1}\right)-\mathbb{P} \left( \tilde{g}_{1:j-1}=\Psi_{1:j-1}\right)} = A_{j-1}.
$$
We therefore obtain $A_j \leq |G| c_1 e^{-\gamma m} + A_{j-1}$, and hence, $\| \pi_M - \tilde{\pi}_M\|_1 = A_{|\mathcal{I}_M'|} \leq |\mathcal{I}_M'| |G| c_1 e^{-\gamma m}$. Finally, since the cardinality of the index set $|\mathcal{I}_M'|$ is $M'$ in the 1D case, and $(M')^2$ in the 2D case, the claim \eqref{claim convergence} follows.

\textit{Step 3: Conclusion.}
Now, for any bounded estimator $\widehat{F}(\cdot)$, and using the law of total expectation, we can write the MSE for the samples $\{ Z_k\}_{k\in \mathcal{I}_M'}$ and $\{ Y_k\}_{k\in \mathcal{I}_M'}$, respectively, as
$$
\mathbb{E} \left[ \left\| \widehat{F} (\{ Z_k\}_{k\in \mathcal{I}_M'}) - \mathcal{F}(X) \right\|^2 \right] 
= \sum_\Psi \mathbb{E} \left[
 \left\| \widehat{F} (\{ Z_k\}_{k\in \mathcal{I}_M'}) - \mathcal{F}(X) \right\|^2 \mid \{ g_k\}_{k\in \mathcal{I}_M'} = \Psi
\right] \pi_M (\Psi)
$$
and
$$
\mathbb{E} \left[
 \left\| \widehat{F} (\{ Y_k\}_{k\in \mathcal{I}_M'}) - \mathcal{F}(X) \right\|^2 \right] 
= \sum_\Psi \mathbb{E} \left[
 \left\| \widehat{F} (\{ Y_k\}_{k\in \mathcal{I}_M'}) - \mathcal{F}(X) \right\|^2 \mid \{ \tilde{g}_k\}_{k\in \mathcal{I}_M'} = \Psi
\right] \tilde{\pi}_M (\Psi).
$$
Since the difference between $\{ Z_k\}_{k\in \mathcal{I}_M'}$ and $\{ Y_k\}_{k\in \mathcal{I}_M'}$ is only the probability distribution for the group elements $\{ {g}_k \}_{k\in \mathcal{I}_M'}$, we have that the MSE, conditionally on  $\{ {g}_k \}_{k\in \mathcal{I}_M'}$, is the same for both cases.
In other words, for any admissible $\Psi\in G^{\mathcal{I}_M'}$, we have
\begin{equation*}
\mathbb{E} \left[ \left\| \widehat{F} (\{ Z_k\}_{k\in \mathcal{I}_M'}) - \mathcal{F}(X) \right\|^2 \mid \{ g_k\}_{k\in \mathcal{I}_M'} = \Psi
\right] 
=
\mathbb{E} \left[ \left\| \widehat{F} (\{ Y_k\}_{k\in \mathcal{I}_M'}) - \mathcal{F}(X) \right\|^2 \mid \{ \tilde{g}_k\}_{k\in \mathcal{I}_M'}  =\Psi
\right] .
\end{equation*}
Hence, using the boundedness of $\widehat{F}(\cdot)$ and the estimate \eqref{claim convergence}, we deduce the existence of a constant $C>0$, which is independent of $M$, such that
\begin{eqnarray*}
\left|
\mathbb{E} \left[
 \left\| \widehat{F} (\{ Z_k\}_{k\in \mathcal{I}_M'}) - \mathcal{F}(X) \right\|^2 \right]
-
\mathbb{E} \left[ \left\| \widehat{F} (\{ Y_k\}_{k\in \mathcal{I}_M'}) - \mathcal{F}(X) \right\|^2 \right]
\right| &\leq & C \sum_\Psi | \pi_M (\Psi) - \tilde{\pi}_M (\Psi) | \\
&\leq & C |G| c_1 (M')^j e^{-\gamma m}
\end{eqnarray*}
Finally, for the convergence rate $a(\sigma)$ from~\eqref{eqn:2d-convergence-rate}, we can compute
\begin{eqnarray*}
    \abs{
    \dfrac{\mathbb{E} \left[
    \left\| \widehat{F} (\{ Z_k\}_{k\in \mathcal{I}_M'}) - \mathcal{F}(X) \right\|^2 \right]}{a(\sigma) / (M')^j} - 1
    } & \leq & C |G| c_1 \dfrac{(M')^{2j}}{a(\sigma)} e^{-\gamma m}
 + \abs{
    \dfrac{\mathbb{E} \left[
 \left\| \widehat{F} (\{ Y_k\}_{k\in \mathcal{I}_M'}) - \mathcal{F}(X) \right\|^2 \right]}{a(\sigma)/(M')^j} - 1
    }.
\end{eqnarray*}
Since we chose $m=\floor{c\log M}$ and $M' = \lfloor M/m\rfloor$ for a constant $c$, we deduce that the first term in the right-hand side vanishes as $M'\to \infty$ (which implies $M\to \infty$) provided the constant $c$ is large enough (we need $c\geq c_0 = 2/\gamma$ for the 1D case and $c\geq c_0 = 4/\gamma$ for the 2D case).
The second term in the right-hand side of the above inequality also vanishes by assumption, since the estimator $\widehat{F}(\cdot)$ has convergence rate $a(\sigma)$ for the i.i.d.\ measurements $\{ Y_k\}_{k\in \mathcal{I}_M'}$ (see Definition \ref{def:convergenceRate1D}).
The first statement of the theorem then follows.

\textbf{Part (ii):}
Recall that ${Z}^{(M)}:= \{ {Z}_k \}_{k\in \mathcal{I}_M}$ are given by \eqref{MRA def proof}, where ${g}^{(M)}:=\{ {g}_k \}_{k\in \mathcal{I}_M}$ is a random field over $G$ satisfying the mixing property \eqref{eq:mixing-property}, and that ${Y}^{(M)}:= \{ Y_k \}_{k\in \mathcal{I}_M}$ are given by \eqref{MRA def proof}, where $\tilde{g}^{(M)}:=\{ \tilde{g}_k \}_{k\in \mathcal{I}_M}$ are i.i.d.\ from the probability distribution $\pi$ over $G$.
Note that, differently to Part (i), we do not consider here a subset of indices.

The proof consists in proving that
\begin{equation}
\label{limsup ineq proof}
\limsup_{M\to \infty} |\mathcal{I}_M| \left| \mse \left( \widehat{F}(Z^{(M)}) \right) - \mse \left( \widehat{F}(Y^{(M)}) \right)  \right| \leq \| F\|_\infty K_1 + K_2,
\end{equation}
for some constants $K_1,K_2>0$,
which, in view of the convergence rate of $\widehat{F} (Y^{(M)})$ and the assumption $a(\sigma) \geq \tau>0$, implies
$$
\limsup_{M\to \infty} \dfrac{\mse \left( \widehat{F}(Z^{(M)}) \right)}{\frac{a(\sigma)}{|\mathcal{I}_M|}} \leq \dfrac{\| F\|_\infty K_1 + K_2}{a(\sigma)} +  \limsup_{M\to \infty} \dfrac{\mse \left( \widehat{F}(Y^{(M)}) \right)}{\frac{a(\sigma)}{|\mathcal{I}_M|}} \leq \dfrac{\| F\|_\infty K_1 + K_2}{\tau} + 1,
$$
and the conclusion follows.

The proof of \eqref{limsup ineq proof} is done in four steps. First we use the bias-variance decomposition to reduce the proof of \eqref{limsup ineq proof} to that of estimating the bias and the variance of $\widehat{F}(Z^{(M)})$. In steps 2 and 3, we estimate the bias and the variance respectively, and in step 4, we assemble everything together to conclude the proof.

\textit{Step 1: Bias-variance decomposition.}
For the i.i.d.\ measurements ${Y}^{(M)}$, we have the standard bias-variance decomposition
\[
\mse \left( \widehat{F}({Y}^{(M)}) \right) = \mathbb{E} \left[\left\|\widehat{F}({Y}^{(M)})- \mathcal{F}(X)\right\|^2\right] = \bias(\widehat{F}({Y}^{(M)}))^2+\Var(\widehat{F}({Y}^{(M)})),\]
where the bias of the estimator is defined as $\bias(\widehat{F}(Y^{(M)}))=\norm{\mathbb{E} [\widehat{F}(Y^{(M)})]-\mathcal{F}(X)}$. 
Since $\widehat{F} ({Y}^{(M)})$ is an i.i.d.\ average, we have
$$
\bias \left(\widehat{F}({Y}^{(M)})\right)=\bias \left({F}({Y}_{1})\right)
\quad \text{and} \quad 
\Var\left( \widehat{F}({Y}^{(M)})\right) =\dfrac{1}{|\mathcal{I}_M|} \Var\left({F}({Y}_1)\right).
$$
We assumed that the MSE with i.i.d.\ data vanishes (for fixed $\sigma$) as $M\to \infty$, which implies that the bias is zero with i.i.d.\ data. This implies in particular that $\mathbb{E} \left[ F (Y_1) \right] = \mathcal{F}(X)$. 

Using the bias-variance decomposition for estimators $\widehat{F} (Z^{(M)})$ and $\widehat{F}(Y^{(M)})$, we obtain
\begin{align*}
    & \left| \mse \left( \widehat{F}(Z^{(M)}) \right) - \mse \left( \widehat{F}(Y^{(M)}) \right)  \right| = \\
    & \qquad \qquad   =\left| \bias(\widehat{F}({Z}^{(M)}))^2 - \bias(\widehat{F}({Y}^{(M)}))^2 + \Var(\widehat{F}({Z}^{(M)})) - \Var(\widehat{F}({Y}^{(M)})) \right| \\
     & \qquad \qquad \leq  \left| \bias(\widehat{F}({Z}^{(M)}))^2 - \bias(\widehat{F}({Y}^{(M)}))^2 \right| + \left| \Var(\widehat{F}({Z}^{(M)})) - \Var(\widehat{F}({Y}^{(M)})) \right| \\
     & \qquad \qquad = \bias(\widehat{F}({Z}^{(M)}))^2 + \left| \Var(\widehat{F}({Z}^{(M)})) - \Var(\widehat{F}({Y}^{(M)})) \right|
\end{align*}

\textit{Step 2: Bias estimate.}
Now, for the dependent data $Z^{(M)}$, we can estimate the bias of the estimator $\widehat{F} (Z^{(M)})$ as
\begin{align}
\bias \left( \widehat{F} (Z^{(M)}) \right) = \left\| \dfrac{1}{|\mathcal{I}_M|} \sum_{k\in \mathcal{I}_M}\left( \mathbb{E} \left[ F(Z_k) \right] - \mathcal{F} (X) \right) \right\| 
 \leq \dfrac{1}{|\mathcal{I}_M|} \sum_{k\in \mathcal{I}_M}\left\| \mathbb{E} \left[ F(Z_k) \right] - \mathbb{E} \left[ F(Y_1) \right]   \right\|.
 \label{bias estimate}
\end{align}
Since each $Z_k$ has the same probability as $Y_1$ up to the distribution of the group element $g_k$, we have
$$
\mathbb{E} \left[ F(Z_k) \mid g_k = \eta \right] = \mathbb{E} \left[ F(Y_1) \mid \tilde{g}_1 = \eta \right], \qquad \forall \eta \in G,
$$
and using the mixing property \eqref{eq:mixing-property}, we can write
\begin{align*}
\left\| \mathbb{E} \left[ F(Z_k) \right] - \mathbb{E} \left[ F(Y_1) \right]   \right\| & = 
\left\| \sum_{\eta \in G}  \mathbb{P} (g_k = \eta) \mathbb{E} \left[ F(Z_k) \mid g_k = \eta \right] - \sum_{\eta \in G} \pi (\eta) \mathbb{E} \left[ F(Y_1) \mid g_1 = \eta \right]  \right\| \\
& \leq  \sum_{\eta \in G} \left| \P (g_k = \eta) - \pi (\eta)  \right| \left\| \mathbb{E} \left[ F(Y_1) \mid \tilde{g}_1 = \eta \right] \right\| \\
& \leq C \exp \left( -\gamma \dist (\{ k \}, \mathcal{I}_M^c) \right),
\end{align*}
for some constant $C$. Due to the exponential decay of $\exp \left( -\gamma \dist (\{ k \}, \mathcal{I}_M^c) \right)$ as $M\to \infty$, we have 
$$
\lim_{M\to \infty} \sum_{k\in\mathcal{I}_M} C \exp \left( -\gamma \dist (\{ k \}, \mathcal{I}_M^c) \right) = K_1,
$$ 
for some $K_1>0$.
Hence, \eqref{bias estimate} implies
\begin{equation}
\label{limsup bias}
\limsup_{M\to \infty} |\mathcal{I}_M| \bias\left(\widehat{F}(Z^{(M)})\right) \leq  K_1.
\end{equation}
It remains to control the variance of $\widehat{F}(Z^{(M)})$.

\textit{Step 3: Variance estimate.}
We can write
\begin{align}
\left| \Var\left(\widehat{F}(Z^{(M)})\right) - \Var \left( \widehat{F}(Y^{(M)}) \right) \right| &
= \left| \dfrac{1}{|\mathcal{I}_M|^2}\sum_{(i,j)\in \mathcal{I}_M^2} \Cov\brackets[\Big]{F(Z_i),F(Z_j)} -  \Var \left( \widehat{F}(Y^{(M)}) \right) \right|
\nonumber \\
&\leq \underbrace{ \left| \dfrac{1}{|\mathcal{I}_M|^2} \sum_{i\in \mathcal{I}_M} \Var \left( F(Z_i) \right)  - \Var \left( \widehat{F}(Y^{(M)}) \right) \right|}_{I} \nonumber\\
& \quad + \underbrace{\dfrac{1}{|\mathcal{I}_M|^2}\sum_{\substack{(i,j)\in \mathcal{I}_M^2 \\i\neq j}} \left| \Cov\left( F(Z_i),F(Z_j) \right) \right|}_{II}.
\label{var X hat Y}
\end{align}
\textit{Step 3.1: Controlling term I.}
Since ${Y}^{(M)} = \{ {Y}_i \}_{1\leq i\leq M^d}$ are i.i.d., we have
\begin{align*}
\left| 
 \dfrac{1}{|\mathcal{I}_M|^2} \sum_{i\in \mathcal{I}_M} \Var \left( F(Z_i) \right) - \Var \left(\widehat{F}({Y}^{(M)}) \right)
\right|  = & 
 \left| 
 \dfrac{1}{|\mathcal{I}_M|^2} \sum_{i\in \mathcal{I}_M}\left( \Var \left( F(Z_i) \right) - \Var({F}({Y}_1)
\right) \right|\\
\leq & 
 \dfrac{1}{|\mathcal{I}_M|^2} \sum_{i\in \mathcal{I}_M}\left| \Var \left( F(Z_i) \right) - \Var({F}({Y}_{1}))
\right|
\end{align*}
Due to the mixing property \eqref{eq:mixing-property}, $|\Var (F(Z_i)) - \Var (F(Y_1))| \to 0$ as the distance of the index $i$ from the boundary of $\mathcal{I}_M$ goes to infinity, which implies that
$$
\lim_{M\to \infty}\dfrac{1}{|\mathcal{I}_M|} \sum_{i\in \mathcal{I}_M}\left| \Var \left( F(Z_i) \right) - \Var({F}({Y}_{1}))
\right| = 0.
$$
Hence,
\begin{align}
    \label{lim Var Y}
\lim_{M\to\infty} |\mathcal{I}_M| \left| 
 \dfrac{1}{|\mathcal{I}_M|^2} \sum_{i\in \mathcal{I}_M} \Var \left( F(Z_i) \right) - \Var \left(\widehat{F}({Y}^{(M)}) \right)
\right|= & 0.
\end{align}

\textit{Step 3.2: Controlling term II.}
Next, we prove that, for any $(i,j)\in \mathcal{I}_M^2$ with $i\neq j$, it holds that
\begin{equation}
\label{E(Z_i Z_j) estimate}
\left| \mathbb{E} \left[ F(Z_i) F(Z_j) \right] -  \mathbb{E} \left[ F(Y_1) \right]^2 \right| \leq 
C \exp \left(-\gamma \min \{ |i-j|, \,  \dist (\{ i,j\}, \mathcal{I}_M^c) \} \right).
\end{equation}
for some constant $C>0$, and for $\gamma$ as in Assumption~\ref{def:exponential-mixing}. 

First, we note that since $Y_i$ and $Y_j$ are i.i.d., we have $\mathbb{E}\left[ F(Y_i)F({Y_j})\right] = \mathbb{E} [F(Y_1)]^2$, for all $i\neq j$.
Moreover, since the distribution of $Z_i, Z_j$ and $Y_i, Y_j$ are equal up to the distribution of the associated group actions, we have that
$$
\mathbb{E} \left[ F(Z_i)F(Z_j) | g_i = \eta_i, g_j = \eta_j \right] =
\mathbb{E} \left[ F(Y_i)F(Y_j) | \tilde{g}_i = \eta_i, \tilde{g}_j = \eta_j \right],
\qquad \forall \eta_i, \eta_j\in G.
$$
Therefore, we have that
\begin{align}
\mathbb{E} \left[ F(Z_i) F(Z_j) \right] & = \mathbb{E} \left[ F(Z_i) F(Z_j) -  F(Y_i) F(Y_j) \right] +   \mathbb{E} \left[ F(Y_1) \right]^2  \nonumber  \\ 
&=\sum_{\eta_i, \eta_j} \big( \mathbb{P} \left[ g_i=\eta_i, g_j = \eta_j \right] - \pi (\eta_i) \pi (\eta_j) \big) \mathbb{E}\left[ F(Z_i)F(Z_j) | g_i = \eta_i, g_j = \eta_j \right] \nonumber \\
& \quad + \mathbb{E} \left[ F(Y_1) \right]^2 . \label{expectation ZiZj}
\end{align}
Now, we can use the formula for conditional probabilities and the mixing property \eqref{eq:mixing-property} to obtain
\begin{align*}
    \left| \P \left[ g_i = \eta_i , g_j = \eta_j \right] - \pi (\eta_i) \pi (\eta_j) \right| &= 
    \left|  \P \left[ g_i = \eta_i \mid g_j = \eta_j \right] \P \left[ g_j = \eta_j \right] - \pi (\eta_i) \pi (\eta_j) \right| \\
    &\leq \P (g_j = \eta_j) \left| \P \left[ g_i = \eta_i \mid g_j = \eta_j \right] - \pi (\eta_i) \right| \\
    & \quad + \pi (\eta_i) \left|  \P \left[ g_j = \eta_j \right] - \pi (\eta_j) \right| \\
    & \leq 2 c_1 \exp \left(-\gamma \min \{ |i-j|, \,  \dist (\{ i,j\}, \mathcal{I}_M^c) \} \right)
\end{align*}
Plugging this estimate in \eqref{expectation ZiZj}, we obtain \eqref{E(Z_i Z_j) estimate}.
Similarly, one can prove that
\begin{equation}
\label{E(Z_i) E(Z_j) estimate}
\left| \mathbb{E} \left[ F(Z_i)\right] \mathbb{E} \left[ F(Z_j) \right] -  \mathbb{E} \left[ F(Y_1) \right]^2 \right| \leq 
C \exp \left(-\gamma \min \{ |i-j|, \,  \dist (\{ i,j\}, \mathcal{I}_M^c) \} \right),
\end{equation}
for all $i\neq j$.

Using \eqref{E(Z_i Z_j) estimate} and \eqref{E(Z_i) E(Z_j) estimate}, we can estimate $\Cov \left( F(Z_i), F(Z_j)\right)$ as
\begin{align}
    \left| \Cov \left( F(Z_i), F(Z_j)\right) \right| & = \left| \mathbb{E} \left[ F(Z_i) F(Z_j)\right] - \mathbb{E}\left[ F(Z_i) \right] \mathbb{E} \left[ F(Z_j) \right] \right| \nonumber  \\
    &  \leq  \left| \mathbb{E} \left[ F(Z_i) F(Z_j) \right] -  \mathbb{E} \left[ F(Y_1) \right]^2 \right| \nonumber  \\
    & \quad + \left| \mathbb{E} \left[ F(Y_1)\right]^2 - \mathbb{E}\left[ F(Z_i) \right] \mathbb{E} \left[ F( Z_j) \right] \right| \nonumber  \\
    &\quad \leq 2 C \exp \left( - \gamma \min \{ d(i,j), \, \dist (\{i,j\}, \mathcal{I}_M^c) \} \right).   
    \label{covariance estimate 0}
\end{align}
Due to the exponential decay with respect to the distance, for every fixed $i_0\in \mathcal{I}_M$, we have
$$
\sum_{\substack{j\in \mathcal{I}_M \\ j\neq i_0}} 2 C \exp \left( - \gamma \min \{ d(i_0,j), \, \dist (\{i_0,j\}, \mathcal{I}_M^c) \} \right) \leq K_2,
$$
for some $K_2>0$ independent of $M$. Hence, from \eqref{covariance estimate 0} we obtain
\begin{equation}
    \label{covariance estimate} 
    \dfrac{1}{|\mathcal{I}_M|^2} \sum_{\substack{(i,j)\in \mathcal{I}_M^2 \\ i\neq j}} |\Cov (F(Z_i), F(Z_j)) | \leq \dfrac{1}{|\mathcal{I}_M|^2}\sum_{i\in \mathcal{I}_M} K_2 = \dfrac{K_2}{|\mathcal{I}_M|}.
\end{equation}

\textit{Step 4: Conclusion.}
Finally, combining \eqref{lim Var Y} and \eqref{covariance estimate} with \eqref{var X hat Y}, we obtain
$$
\limsup_{M\to \infty} |\mathcal{I}_M| \left| \Var \left( \widehat{F} (Z^{(M)}) \right) - \Var \left( \widehat{F}( Y^{(M)} ) \right) \right| \leq K_2.
$$
The fact that $\bias (\widehat{F}(Y^{(M)})) = 0$ together with \eqref{limsup bias} imply that
$$
\limsup_{M\to \infty} |\mathcal{I}_M| \left| \bias (\widehat{F}(Z^{(M)})) - \bias (\widehat{F}(Y^{(M)}))  \right| \leq K_1.
$$
Hence, using the bias variance decomposition \
$$
\mse \left( \widehat{F} (Z^{(M)}) \right) = \bias \left( \widehat{F} (Z^{(M)}) \right)^2 + \Var \left(   \widehat{F} (Z^{(M)}) \right),
$$
we obtain
$$
\limsup_{M\to \infty} |\mathcal{I}_M| \left| \mse \left( \widehat{F}(Z^{(M)}) \right) - \mse \left( \widehat{F}(Y^{(M)}) \right)  \right| \leq \| F\|_\infty K_1 + K_2,
$$
which, in view of the convergence rate of $\widehat{F} (Y^{(M)})$ and the assumption $a(\sigma) \geq \tau>0$, implies
$$
\limsup_{M\to \infty} \dfrac{\mse \left( \widehat{F}(Z^{(M)}) \right)}{\frac{a(\sigma)}{|\mathcal{I}_M|}} \leq \dfrac{\| F\|_\infty K_1 + K_2}{a(\sigma)} +  \limsup_{M\to \infty} \dfrac{\mse \left( \widehat{F}(Y^{(M)}) \right)}{\frac{a(\sigma)}{|\mathcal{I}_M|}} \leq \dfrac{\| F\|_\infty K_1 + K_2}{\tau} + 1,
$$
and the conclusion follows.
\hfill $\square$

\end{appendix}

\end{document}